\documentclass[orivec]{llncs}

\usepackage{mathtools,amsmath,amssymb,mathabx,stmaryrd}
\usepackage{booktabs}
\usepackage{bussproofs}
\usepackage{tabulary,tabularx}
\usepackage{minibox}
\usepackage[formats]{listings}
\usepackage{soul,xcolor}
\usepackage{comment}
\usepackage[moderate]{savetrees}
\usepackage{txfonts}

\newcommand{\sra}{\mbox{{\footnotesize $\rightarrow$}}}

\newcommand{\goali}[4]{\ensuremath{\left[#4,#1,#2\right]\Vdash #3}}

\newcommand{\EPS}{\textbf{EPS}}
\newcommand{\GSI}{\textbf{GSI}}
\newcommand{\CAS}{\textbf{CAS}}
\newcommand{\CS}{\textbf{CS}}
\newcommand{\CVUL}{\textbf{CVUL}}
\newcommand{\VA}{\textbf{VA}}
\newcommand{\ERL}{\textbf{ERL}}
\newcommand{\ERR}{\textbf{ERR}}

\title{Inductive Reasoning with Equality Predicates,\\Contextual Rewriting and 
Variant-Based Simplification}
\author{Jos\'{e} Meseguer}
\institute{Department of Computer Science\\
University of Illinois at Urbana-Champaign, USA \\
	       \email{meseguer@illinois.edu}}

\authorrunning{}

\begin{document}
\maketitle
\pagestyle{plain}
\begin{abstract}
An inductive inference system
for proving validity of formulas in 
the initial algebra $T_{\mathcal{E}}$
of an order-sorted equational theory $\mathcal{E}$
is presented.  It has 20 inference rules, but only 9 of them 
require user interaction; the remaining 11
can be automated as \emph{simplification rules}.
In this way, a substantial fraction of the
proof effort can be automated.
The inference rules are based on
advanced equational reasoning techniques, including:
equationally defined equality predicates, narrowing,
constructor variant unification,
variant satisfiability,
order-sorted
congruence closure, contextual rewriting, ordered rewriting,
and recursive path orderings. All
these techniques work modulo axioms $B$,
for $B$ any combination of associativity and/or commutativity and/or identity axioms.
Most of these inference rules have  already been  implemented in Maude's
{\bf NuITP} inductive theorem prover.
\end{abstract}

%% \begin{center}
%% Version of June 10 2021
%% \end{center}

\section{Introduction} \label{INTRO-SECT}

In inductive theorem proving for equational specifications there is a
tension between automated approaches, e.g., 
\cite{DBLP:conf/popl/Musser80,ind,DBLP:journals/jcss/HuetH82,DBLP:journals/ai/KapurM87,DBLP:conf/lics/Bachmair88,DBLP:journals/jar/BouhoulaR95,DBLP:journals/iandc/ComonN00},
and explicit induction ones,
 e.g.,
\cite{DBLP:conf/rta/KapurZ89,DBLP:conf/alp/Goguen90,DBLP:series/mcs/GuttagHGJMW93,cafe-tools-paper,DBLP:conf/birthday/GAinALOF14,itp-manual,hendrix-thesis,itp/HendrixKM10,DBLP:journals/corr/abs-2101-02690}.
For two examples of automated equational inductive provers we can mention,
among various others, Spike \cite{DBLP:journals/jar/BouhoulaR95} and the
superposition-based ``inductionless induction'' prover in
\cite{DBLP:journals/iandc/ComonN00}; and for explicit induction
equational provers we can mention, again among various others, 
RRL \cite{DBLP:conf/rta/KapurZ89}, OBJ as a Theorem Prover
\cite{DBLP:conf/alp/Goguen90,DBLP:journals/corr/abs-2101-02690},
the OTS/CafeOBJ Method
\cite{Futatsugi10,DBLP:conf/birthday/OgataF14,DBLP:conf/birthday/GAinALOF14,DBLP:journals/scp/Futatsugi22}, 
and the Maude ITP
\cite{cafe-tools-paper,itp-manual,hendrix-thesis,itp/HendrixKM10}.  
The advantage of
automated provers is that they do not need user interaction, although they
often require proving auxiliary lemmas.
Explicit induction is less automated, but it allows users to be ``in the
driver's seat'' for finding a proof.  This work presents an approach that 
combines features from automated and explicit-induction theorem
proving in the context of proving validity in the initial algebra $T_{\mathcal{E}}$
of an order-sorted equational theory $\mathcal{E}$ for
both arbitrary quantifier-free (QF)
formulas (expressed as conjunctions of clauses, some 
of which can be combined together as ``multiclauses'' in the sense explained
in Section \ref{SuperClIndTh}) and for existential formulas
reduced to QF form by Skolemization (see Section \ref{Skolem-subsection}).

The combination of proving styles
 is achieved by an inference system having 20 inference
rules, where 11 of them are \emph{goal simplification rules} that can
be fully automated, whereas the remaining 9 require explicit user
interaction, unless they are also automated by tactics.
%
% In fact, we have combined 9 of those simplification rules
% into an automated simplification strategy that we call
% \emph{ISS}.  A more powerful $\mathit{ISS}^{+}$ combining 
%the 11  simplification rules could easily be developed.
%
The
simplification rules are powerful enough to be used on their
own as an automatic (partial) \emph{oracle} to answer inductive validity
questions.  For example, as mentioned in the earlier
paper \cite{ind-ctxt-rew}, a prototype version of
a subset of the simplification rules made it possible
to discharge a huge number of inductive validity verification
conditions (VCs) that were generated in the deductive verification
proof in constructor-based reachability logic of the security
properties of the IBOS Browser described in \cite{skeirik-thesis,ibos-wrla20}.
For a more recent example, the simplification rules implemented
in the {\bf NuITP} prover are invoked as an oracle to discharge
VCs generated by the {\bf DM-Check} tool when proving invariants of
infinite-state systems \cite{DM-Check-WRLA24}.

The effectiveness of these simplification rules seems to be
due to a novel \emph{combination} of 
\emph{automatable} equational reasoning techniques.
Although several are well-known
and some are widely used in superposition-based
\cite{DBLP:journals/logcom/BachmairG94}
automatic first-order theorem provers such as, e.g.,
\cite{DBLP:journals/aicom/Schulz02,weidenbach-SPASS},
to the best of my knowledge they have not been previously combined
for inductive theorem proving purposes with the
extensiveness and generality presented
here.  They include:
(1) equationally defined equality predicates
\cite{DBLP:journals/scp/GutierrezMR15};
(2) constrained narrowing \cite{conditional-narrowing-SCP};
(3) constructor variant unification
\cite{var-sat-scp,skeirik-meseguer-var-sat-JLAMP};
(4) variant satisfiability
\cite{var-sat-scp,skeirik-meseguer-var-sat-JLAMP};
 (5) recursive path orderings  \cite{rubio-thesis,MTA};
(6) order-sorted
congruence closure \cite{DBLP:conf/fossacs/Meseguer16};
(7) contextual rewriting \cite{DBLP:journals/fuin/Zhang95};
and (8) ordered rewriting, e.g.,
\cite{DBLP:conf/cade/MartinN90,Bachmair91,DBLP:books/el/RV01/NieuwenhuisR01}).
Furthermore, in this work all these techniques work \emph{modulo} axioms $B$,
for $B$ any combination of associativity and/or commutativity and/or
identity axioms.

Since the paper \cite{ind-ctxt-rew} was published, two important
developments have taken place.  First, as further explained 
below, the inference system has been extended with new 
inference rules to make it more effective and scalable, and several of the
rules in \cite{ind-ctxt-rew}  have been extended
in their scope and applicability for the same purpose.  Furthermore,
the mathematical foundations of the inference system have been
developed and its completeness has been proved.  Second, since mid
2021 to the present I have been working with Francisco Dur\'{a}n,
Santiago Escobar and Julia Sapi\~na on the design and implementation of, and
experimentation with, an inductive theorem prover for Maude equational programs
called Maude's {\bf NuITP} that implements
(most of) the inference system presented in this paper and whose
user manual, code and examples are
publicly available \cite{NuITP}.  The {\bf NuITP} has been used in formal
methods courses since 2022 at several universities, including
the University of Illinois at Urbana-Champaign and the University of
M\'{a}laga, and  a substantial body of examples and applications
have been developed.  Furthermore, as already mentioned,
the {\bf NuITP} is used as a backend
to dischage inductive verification conditions generated by  the
{\bf DM-Check} tool  \cite{DM-Check,DM-Check-WRLA24}, a deductive model checker under development
by Kyungmin Bae, Santiago Escobar, Ra\'{u}l  L\'{o}pez, Jos\'{e}
Meseguer and Julia Sapi\~{n}a, which
can be used to verify invariants of infinite-state concurrent systems.
The theoretical and tool developments have been closely related.  For
example, based on the experience gained with the {\bf NuITP},
two new rules have been added to the inference system to increase the
effectiveness of the {\bf NuITP}.

\subsection*{Relationship of this Paper to an Earlier Conference Paper} \label{ADDTL-CONTS}

In 2020 Stephen Skeirik and I published
a preliminary set of inference
rules and some preliminary experiments in
the conference paper \cite{ind-ctxt-rew}.
 Since then,  I have 
 advanced the theoretical foundations and have improved and
 substantially extended the inference system presented 
 in \cite{ind-ctxt-rew}.  Also, as mentioned above,
since mid 2021 these advances have been used both as
the theoretical basis of  Maude's
{\bf NuITP} tool and as a means to make  the {\bf NuITP} a more effective prover.

Besides making substantial improvements to the inference rules in
\cite{ind-ctxt-rew} and developing  new formula simplification
 techniques (see below), four
 new inference rules, namely,
 the
 \emph{narrowing induction} (${\bf NI}$),
\emph{narrowing simplification} (${\bf NS}$),
\emph{equality} (${\bf Eq}$)
and \emph{cut} (${\bf Cut}$) rules,
make the inference system more powerful and
 versatile.

${\bf NI}$ and ${\bf NS}$ both have the effect of
fully evaluating at the symbolic level one step of computation with a
defined function symbol $f$ for all the ground instances of
a given expression $f(\vec{v})$ appearing in a conjecture.
${\bf NI}$ and ${\bf NS}$ complement each other and serve
somewhat different purposes.  ${\bf NI}$ is a full-blown induction
rule that inducts on smaller calls to $f$ appearing in the righthand
sides or conditions of the equations defining $f$. 
Instead, ${\bf NS}$ is used mainly for goal simplification purposes
and is a key contributor to the ``virtuous circle'' of simplification
by which ``chain reactions'' of formula simplification are triggered.
As I show in several examples, ${\bf NI}$ compares favorably
with more standard rules like structural induction or its 
more flexible generalization
to \emph{generator set induction} (the ${\bf GSI}$ rule in the present
system).  Another attractive feature of both 
 the ${\bf NI}$ and  ${\bf NS}$  rules is that both
can be used in a fully automatic way.  For ${\bf NS}$ this is
to be expected, but for  ${\bf NI}$ this
opens up the future prospect of an automated
use of the inference system
where the {\bf NuITP} could be used as a backend by many other formal tools.
The ${\bf Eq}$ inference rule allows the application of
a conditional equation, either in the original theory, an induction hypothesis,
or an already proved lemma, to a subterm of the current goal
with a user-provided (partial) substitution.
The ${\bf Cut}$ rule has the usual meaning in many
other inference systems. In the context of reasoning with multiclauses,
which are \emph{conditional} formulas of the form
$\Gamma \rightarrow \Lambda$ with $\Gamma$ a conjunction of equalities
(more on this in \S \ref{SCBR}),
it serves the purpose of a generalized
\emph{modus ponens} inference rule, since modus
ponens becomes the special case when $\Gamma = \top$.

\vspace{1ex}

\noindent Another important new addition is the use of
\emph{ordered rewriting modulo axioms} $B$ ---including conditional 
ordered rewriting--- to substantially increase the possibility of applying
induction hypotheses as rewrite rules for conjecture simplification.
Two entirely new sections  have  been added on this topic (Sections
\ref{OR-section} and \ref{IndH-REW}),
and a considerably more nuanced classification of different types
of induction hypotheses is made in Section \ref{SuperClIndTh}
in order to more effectively apply induction hypotheses to simplify
conjectures.  Furthermore, to boost the effectiveness
of ordered rewriting for conjecture simplification purposes, 
Section \ref{OR-section} concludes with a new semantics for ordered
rewriting that, 
by making ordered rewriting ``background theory aware,''
is strictly more powerful in simplifying conjectures than the
standard semantics of ordered rewriting.  To the best of my
knowledge this new semantics seems to
be a novel contribution.

\vspace{1ex}

\noindent Four other entirely new sections have been added: 
Section \ref{RWL-BACKGR} on rewriting logic; 
Section \ref{GS-FubCall-Subsection}, explaining in detail some
nuances about \emph{generator sets} modulo axioms $B$ and all the
technicalities about \emph{function subcalls} needed to support the new 
 ${\bf NI}$ rule described above; Section \ref{Skolem-subsection},
explaining in detail the model-theoretic semantics of 
\emph{Skolemization}
for satisfaction of formulas existentially quantified on Skolem
function symbols and constants, which supports a more
general $\exists$ rule now included in the inference system
to support inductive reasoning about quantified first-order formulas;
and Section \ref{inf-ex}, illustrating
the use of the inference
system with a collection of examples
(all new, except one borrowed from \cite{ind-ctxt-rew})
that go beyond
the numerous other examples given in Section \ref{ind-inf-rules}
to illustrate the use of individual rules, which are also new.
The examples in  Section \ref{inf-ex}
illustrate both the meaning and the application of 
important inference rules ---including  ${\bf GSI}$,  ${\bf NI}$, 
${\bf LE}$,  ${\bf CAS}$, ${\bf VA}$ and  ${\bf Cut}$,
and the new, generalized versions of ${\bf EPS}$,
${\bf CVUL}$ and 
 ${\bf ICC}$--- as well as 
the remarkable power and economy of thought 
 of reasoning with the \emph{multiclause} representation 
of formulas supported by the inference system, as opposed to the 
usual clause representation.
 Last but not least,  the proof of soundness of the inference
system is given in Appendix 
\ref{Proof-Sound-Theo}.

\vspace{1ex}

\noindent Besides all these entirely new additions and inference rules, the
following substantial extensions of previous
inference rules further  differentiate the
current inference system from its erlier, preliminary version in \cite{ind-ctxt-rew}:

\begin{enumerate}
\item The equality predicate simplification rule
(\textbf{EPS}) and the 
 inductive congruence closure rule (${\bf ICC}$) have been
  substantially generalized to increase their effectiveness: they now can
  also use ordered rewriting to apply a much wider
range of induction hypotheses in the simplification process that could
not be applied in the previous
formulations of \textbf{EPS} and
${\bf ICC}$.

\item The two previous constructor variant unification left
simplification rules, namely, {\bf CVUL} and {\bf CVUFL},
have been substantially generalized and subsumed
by a single, more widely applicable {\bf CVUL} rule.

\item The substitution left (\textbf{SUBL}) and right
  (\textbf{SUBR}) rules have also been substantially
generalized: now they can, not only substitute equations of the
form $x=u$ with $x$ a variable as before, but  can also substitute
equations of the form $\overline{x}=u$, with $\overline{x}$ a
(universal) Skolem constant.

\item The generator set induction (\textbf{GSI}) rule now generates 
  stronger induction hypotheses.

\item The  case (\textbf{CAS}) rule has also been generalized,
so that it can be applied, not just on a variable, but also on a 
(universal) Skolem constant.

\item The existential ($\exists$), lemma enrichment (\textbf{LE}),
split  ($\textbf{SP}$) and  variable abstraction  ($\textbf{VA}$)  rules have been 
substantially generalized.
\end{enumerate}

\section{Preliminaries} \label{PRELIM-SECT}

Since many different techniques are combined in the inductive
inference system, to make the paper self-contained and
as easy to understand as possible,
an unusually long list of preliminaries need to be covered
on: (i) order-sorted first-order logic, (ii) convergent theories
and constructors; (iii) rewriting logic; (iv) narrowing;
(v) equationally-defined equality predicates; (vi) congruence closure;
(vii) contextual rewriting; (viii) variant unification and satisfiability;
(ix) generating sets and function subcalls; (x) ordered rewriting; 
and (xi) Skolemization.  Readers already familiar with some of these
topics may skip the corresponding sections or,  perhaps better,  read them quickly
to become familiar with the notation used.

\subsection{Background on Order-Sorted First-Order Logic}

Familiarity with the notions of an order-sorted signature
$\Sigma$ on a poset of sorts $(S,\leq)$,
an order-sorted $\Sigma$-algebra $A$, and the term $\Sigma$-algebras 
$T_{\Sigma}$ and $T_{\Sigma}(X)$ for $X$ an $S$-sorted set of
variables is assumed, as well as with the notions of:
(i) $\Sigma$-homomorphism $h : A \rightarrow B$
between $\Sigma$-algebras $A$ and $B$,
so that $\Sigma$-algebras 
and $\Sigma$-homomorphisms
form a category ${\bf OSAlg}_{\Sigma}$;
 (ii) order-sorted (i.e., sort-preserving) substitution
$\theta$, its domain $\mathit{dom}(\theta)$ and range
$\mathit{ran}(\theta)$, and its application $t \theta$ to a term $t$;
 (iii) \emph{preregular} order-sorted
signature $\Sigma$, i.e., a signature such that each term $t$ has a least sort,
denoted $\mathit{ls}(t)$; (iv) the set $\widehat{S}=S/(\geq \cup
\leq)^{+}$ of \emph{connected components} of a poset $(S,\leq)$
viewed as a DAG; and
(v) for $A$ a $\Sigma$-algebra, the set $A_{s}$ of it elements
of sort $s \in S$, and the set $A_{[s]}=\bigcup_{s' \in [s]}A_{s'}$
of all elements in a connected component
$[s] \in \widehat{S}$.  It is furthermore assumed that all
signatures $\Sigma$ have \emph{non-empty sorts}, i.e.,
$T_{\Sigma,s}\not= \emptyset$ for each $s \in S$.
All these
notions are explained in detail in \cite{tarquinia,osa1}.  The material below
is adapted from \cite{var-sat-scp}.

The first-order language of \emph{equational} 
$\Sigma$-\emph{formulas}
is defined in the usual way: its atoms\footnote{As explained in \cite{var-sat-scp},
there is no real loss of generality
  in assuming that all atomic formulas are equations: predicates can
  be specified by equational formulas using additional function
  symbols of a fresh new sort
  $\mathit{Pred}$ with a constant $\mathit{tt}$, so that a predicate atom
$p(t_{1},\ldots, t_{n})$ becomes
$p(t_{1},\ldots, t_{n})=\mathit{tt}$.} are
$\Sigma$-\emph{equations} $t=t'$, where
$t,t' \in T_{\Sigma}(X)_{[s]}$ for some $[s] \in \widehat{S}$ and
each $X_{s}$ is assumed countably infinite.  The
set $\mathit{Form}(\Sigma)$ of \emph{equational}
$\Sigma$-\emph{formulas}
 is then
inductively built
from atoms by: conjunction ($\wedge$), disjunction ($\vee$),
negation ($\neg$), and universal ($\forall x_{1}\!\!:\!\!s_{1},\ldots, x_{n}\!\!:\!\!s_{n}
$) and existential
($\exists x_{1}\!\!:\!\!s_{1},\ldots, x_{n}\!\!:\!\!s_{n}$)
quantification with distinct sorted variables
$x_{1}\!\!:\!\!s_{1},\ldots, x_{n}\!\!:\!\!s_{n}$, with
$s_{1},\ldots,s_{n} \in S$
(by convention, for $\emptyset$ the empty set of variables and
$\varphi$ a formula, we define
$(\forall \emptyset)\; \varphi
\equiv (\exists \emptyset)\; \varphi
\equiv \varphi$).
A literal $\neg(t=t')$ is denoted $t \not= t'$.
Given a $\Sigma$-algebra $A$, a formula $\varphi \in \mathit{Form}(\Sigma)$,
and an assignment $\alpha \in [Y \sra A]$, where
$Y \supseteq \mathit{fvars}(\varphi)$, with $\mathit{fvars}(\varphi)$
the free variables of $\varphi$,
 the \emph{satisfaction
relation} $A,\alpha \models \varphi$ is defined inductively as usual:
for atoms, $A,\alpha \models t = t'$ iff $t \alpha = t' \alpha$;
for Boolean connectives it is the corresponding Boolean combination
of the satisfaction relations for subformulas; and for
quantifiers: $A,\alpha \models (\forall x_{1}\!\!:\!\!s_{1},\ldots, x_{n}\!\!:\!\!s_{n})\; \varphi$
 (resp.~$A,\alpha \models (\exists x_{1}\!\!:\!\!s_{1},\ldots, x_{n}\!\!:\!\!s_{n})\; \varphi$)
holds  iff for all $(a_1 , \ldots, a_n ) \in A_{s_{1}} \times \ldots
\times A_{s_{n}}$
(resp.~for some  $(a_1 , \ldots, a_n ) \in A_{s_{1}} \times \ldots
\times A_{s_{n}}$) we have
$A,\alpha[x_{1}\!\!:\!\!s_{1} := a_{1},\ldots,
x_{n}\!\!:\!\!s_{n} := a_{n}] \models \varphi$, where
if $\alpha \in [Y \sra A]$, then
$\alpha[x_{1}\!\!:\!\!s_{1} := a_{1},\ldots,
x_{n}\!\!:\!\!s_{n} := a_{n}]
\in [(Y  \cup \{x_{1}\!\!:\!\!s_{1},\ldots, x_{n}\!\!:\!\!s_{n}\})
\sra A]$ and is such that for $y \!\!:\!\!s \in (Y  \setminus
\{x_{1}\!\!:\!\!s_{1},\ldots, x_{n}\!\!:\!\!s_{n}\})$,
$\alpha[x_{1}\!\!:\!\!s_{1} := a_{1},\ldots,
x_{n}\!\!:\!\!s_{n} := a_{n}](y \!\!:\!\!s) = \alpha(y \!\!:\!\!s)$, and
$\alpha[x_{1}\!\!:\!\!s_{1} := a_{1},\ldots,
x_{n}\!\!:\!\!s_{n} := a_{n}](x_{i}\!\!:\!\!s_{i}) = a_{i}$, $1 \leq i
\leq n$.
 $\varphi$ is called  \emph{valid} in $A$
(resp. \emph{satisfiable} in A) iff
$A, \underline{\emptyset} \models (\forall Y)\; \varphi$
(resp. $A, \underline{\emptyset} \models (\exists Y)\; \varphi$),
where $Y=\mathit{fvars}(\varphi)$ and
$\underline{\emptyset}\in [\underline{\emptyset} \sra A]$
denotes the empty $S$-sorted assignment of values in $A$
to the empty $S$-sorted family $\underline{\emptyset}$ of variables.
The notation $A \models \varphi$ abbreviates validity of
$\varphi$ in $A$.  More generally, a set of formulas
$\Gamma \subseteq \mathit{Form}(\Sigma)$
is called \emph{valid} in $A$, denoted $A \models \Gamma$,
iff $A \models \varphi$ for each $\varphi \in \Gamma$.
For a subsignature $\Omega \subseteq
\Sigma$ and $A \in {\bf
  OSAlg}_{\Sigma}$, the \emph{reduct} $A|_{\Omega} \in {\bf
  OSAlg}_{\Omega}$ agrees with
$A$ in the interpretation of all sorts and operations in $\Omega$ and
discards everything in $\Sigma \setminus \Omega$.
If $\varphi \in \mathit{Form}(\Omega)$ we have
the equivalence $A \models \varphi \; \Leftrightarrow \; A|_{\Omega} \models \varphi$.

An OS \emph{equational theory}
is a pair $T=(\Sigma,E)$, with $E$ a set of (possibly conditional)  $\Sigma$-equations.
${\bf OSAlg}_{(\Sigma,E)}$ denotes the full subcategory
of ${\bf OSAlg}_{\Sigma}$ 
with objects those $A \in {\bf
  OSAlg}_{\Sigma}$ such that $A \models E$, called the $(\Sigma,E)$-\emph{algebras}.
${\bf OSAlg}_{(\Sigma,E)}$ has an
\emph{initial algebra} $T_{\Sigma/E}$ \cite{tarquinia}.
Given $T=(\Sigma,E)$ and $\varphi \in \mathit{Form}(\Sigma)$, we call
$\varphi$ $T$-\emph{valid}, written $E
\models \varphi$, iff $A \models \varphi$ for all $A \in {\bf
  OSAlg}_{(\Sigma,E)}$. We call $\varphi$ $T$-\emph{satisfiable}
iff there exists $A \in {\bf
  OSAlg}_{(\Sigma,E)}$ with $\varphi$ satisfiable in A.  Note that $\varphi$ 
is $T$-\emph{valid} iff  $\neg \varphi$ 
is $T$-\emph{unsatisfiable}. 
The inference system in \cite{tarquinia} is
\emph{sound and complete} for OS  equational deduction, i.e., for
any OS equational theory
$(\Sigma,E)$, and $\Sigma$-equation $u=v$
we have an equivalence $E \vdash u=v \;\; \Leftrightarrow \;\; E
\models u=v$.  Deducibility $E \vdash u=v$ is abbreviated 
as $u =_{E} v$, called $E$-\emph{equality}.
An $E$-\emph{unifier} of a system of
$\Sigma$-equations, i.e., of a conjunction 
$\phi=u_{1}=v_{1} \, \wedge \, \ldots \, \wedge \,  u_{n}=v_{n}$ of
$\Sigma$-equations, is a substitution $\sigma$ such that $u_{i} \sigma=_{E} v_{i}\sigma$,
$1 \leq i \leq n$.  An $E$-\emph{unification algorithm} for
$(\Sigma,E)$ 
is an algorithm generating a \emph{complete set} of $E$-unifiers
$\mathit{Unif}_{E}(\phi)$ for any system of $\Sigma$ equations $\phi$,
where ``complete'' means that for any $E$-unifier $\sigma$ of $\phi$ there
is a $\tau \in \mathit{Unif}_{E}(\phi)$ and a substitution
$\rho$ such that $\sigma =_{E} 
(\tau \rho)|_{\mathit{dom}(\sigma) \cup \mathit{dom}(\tau)}$,
 where $=_{E}$ here means that for any variable $x$ we have 
$x\sigma =_{E} x (\tau \rho)|_{\mathit{dom}(\sigma) \cup \mathit{dom}(\tau)}$.
The algorithm is \emph{finitary} if it  always terminates
with a \emph{finite set} $\mathit{Unif}_{E}(\phi)$ for any $\phi$.

Given a set of equations $B$ used for deduction modulo $B$,
a preregular OS signature $\Sigma$ is called
$B$-\emph{preregular}\footnote{\label{broad-b-prereg} If $B = B_{0} \uplus U$, with $B_{0}$ 
associativity and/or commutativity axioms, and $U$ identity axioms,
the $B$-preregularity notion can be \emph{broadened} by requiring
only that:
(i) $\Sigma$ is $B_{0}$-preregular in the standard sense that 
$\mathit{ls}(u\rho)=\mathit{ls}(v\rho)$ for all $u=v \in B_{0}$
and substitutions $\rho$; and (ii) the axioms $U$ oriented as rules $\vec{U}$ are
\emph{sort-decreasing} in the sense explained in Section
\ref{convergent-backgr}.}
iff for each $u =v \in B$ and substitutions $\rho$,
$\mathit{ls}(u\rho)=\mathit{ls}(v\rho)$.

\subsection{Background on Convergent Theories and Constructors}
\label{convergent-backgr}

Given an order-sorted equational theory $\mathcal{E}=(\Sigma,E \cup
B)$, where $B$ is a collection of associativity and/or commutativity
and/or identity axioms and $\Sigma$ is $B$-preregular,
we can associate to it a corresponding
\emph{rewrite theory} \cite{unified-tcs}
$\vec{\mathcal{E}}=(\Sigma,B,\vec{E})$
by orienting the equations $E$ as left-to right rewrite
rules.  That is, each $(u=v) \in E$ is transformed
into a rewrite rule $u \rightarrow v$.  For simplicity I
recall here the case of unconditional equations;
for how conditional equations (whose conditions are
conjunctions of equalities)
are likewise transformed into
conditional rewrite rules see, e.g.,
\cite{lucas-meseguer-normal-th-JLAMP}.
The main purpose of the rewrite theory $\vec{\mathcal{E}}$ is to
make equational logic efficiently computable by
reducing the complex bidirectional reasoning with equations to
the much simpler unidirectional reasoning with rules under suitable
assumptions.  I assume familiarity with the notion of subterm
$t|_{p}$ of $t$ at a term position $p$ and of term replacement
$t[w]_{p}$ of $t|_{p}$ by $w$ at position $p$ (see, e.g., \cite{dershowitz-jouannaud}).
The rewrite relation $t \rightarrow_{\vec{E},B} t'$ holds
iff there is a subterm $t|_{p}$ of $t$, a rule $(u
\rightarrow v) \in \vec{E}$ and a substitution $\theta$
such that $u \theta =_{B} t|_{p}$, and $t'=t[v\theta]_{p}$.
$\rightarrow^{*}_{\vec{E},B}$ denotes the reflexive-transitive
closure of $\rightarrow_{\vec{E},B}$.
The requirements on $\vec{\mathcal{E}}$ that allow
reducing equational reasoning to rewriting are the following:
 (i) $\mathit{vars}(v) \subseteq \mathit{vars}(u)$;
(ii) \emph{sort-decreasingness}:  for each 
substitution $\theta$ we must have $\mathit{ls}(u \theta) \geq
\mathit{ls}(v \theta)$; (iii) \emph{strict} $B$-\emph{coherence}:
if $t_1 \rightarrow_{\vec{E},B} t'_1$ and $t_1 =_{B} t_2$
then there exists $t_2 \rightarrow_{\vec{E},B} t'_2$ with $t'_1 =_{B} t'_2$;
(iv) \emph{confluence} (resp. \emph{ground confluence}) modulo $B$:
for each term $t$ (resp. ground  term $t$)
if $t\rightarrow^{*}_{\vec{E},B}v_1$ and 
$t\rightarrow^{*}_{\vec{E},B}v_2$,
then there exist rewrite sequences
$v_1\rightarrow^{*}_{\vec{E},B}w_1$ and 
$v_2\rightarrow^{*}_{\vec{E},B}w_2$ such that $w_1 =_{B} w_2$;
(v) \emph{termination}: the relation $\rightarrow_{\vec{E},B}$ is
well-founded (for $\vec{E}$ conditional, we require \emph{operational termination}
\cite{lucas-meseguer-normal-th-JLAMP}).
 If $\vec{\mathcal{E}}$ satisfies
conditions (i)--(v) (resp. the same, but (iv) weakened to ground
confluence modulo $B$), then it is called \emph{convergent} (resp. \emph{ground
  convergent}).  The key point is that then, given a term
(resp. ground term) $t$,
all terminating rewrite sequences $t \rightarrow^{*}_{\vec{E},B}w$ 
end in a term $w$, denoted  $t!_{\vec{\mathcal{E}}}$,
that is unique up to $B$-equality,
and it is called $t$'s \emph{canonical form}.
Three major results then follow for the ground convergent case:
(1) for any ground terms $t,t'$ we have $t =_{E \cup B} t'$
iff $t!_{\vec{\mathcal{E}}} =_{B} t'!_{\vec{\mathcal{E}}}$, 
(2) the $B$-equivalence classes of canonical forms are the elements of the \emph{canonical term
  algebra} $C_{\Sigma/E,B}$, where for each $f:s_1 \ldots s_n
\rightarrow s$ in $\Sigma$ and $B$-equivalence classes of  canonical terms $[t_1], \ldots, [t_n]$
with $ls(t_i) \leq s_i$ the operation $f_{C_{\Sigma/E,B}}$ is defined
by the identity:  $f_{C_{\Sigma/E,B}}([t_1] \ldots [t_n]) = [f(t_1\ldots t_n)!_{\vec{\mathcal{E}}}]$,
and (3) we have an isomorphism 
$T_{\mathcal{E}}\cong C_{\Sigma/E,B}$.

A ground convergent rewrite theory 
$\vec{\mathcal{E}}=(\Sigma,B,\vec{E})$ 
is called \emph{sufficiently complete} with respect to a
subsignature  $\Omega$, whose operators are then called \emph{constructors},
iff for each ground $\Sigma$-term $t$, $t!_{\vec{\mathcal{E}}} \in T_{\Omega}$.
Furthermore, for $\vec{\mathcal{E}}=(\Sigma,B,\vec{E})$ 
sufficiently complete w.r.t. $\Omega$,
a ground convergent rewrite subtheory
$(\Omega,B_{\Omega},\vec{E}_{\Omega}) \subseteq (\Sigma,B,\vec{E})$
is called a \emph{constructor subspecification} iff
 $T_{\mathcal{E}}|_{\Omega} \cong T_{\Omega/ E_{\Omega} \cup
  B_{\Omega}}$.  If $E_{\Omega} = \emptyset$, then
$\Omega$ is called a signature
of \emph{free constructors modulo axioms} $B_{\Omega}$.

\subsection{Background on Rewriting Logic} \label{RWL-BACKGR}

Convergent rewrite theories are ideal to specify \emph{deterministic}
systems, where computations have a unique final result.  To formally
specify \emph{non-deterministic} and possibly concurrent
systems, rewriting logic \cite{unified-tcs,20-years} is a 
suitable framework that contains convergent rewrite theories
as a special case.  A \emph{rewrite theory} is a 
triple $\mathcal{R}=(\Sigma,E \cup B,R)$, where 
$(\Sigma,E \cup B)$ is an equational theory, and $R$ is a
collection of, possibly conditional, rewrite rules.  The intended
meaning of $\mathcal{R}$ is to specify a system whose
states are elements of the initial algebra $T_{\Sigma/E \cup B}$,
and whose non-deterministic \emph{transitions} are specified by
the rules $R$.  

Since the states are $E \cup B$-equivalence classes,
conceptually, rewriting with $\mathcal{R}$ is rewriting modulo
$E \cup B$, i.e., the relation $\rightarrow_{R/E \cup B}$.  However,
$\rightarrow_{R/E \cup B}$ is a hard to implement, complicated relation.
In practice, the effect of rewriting with $\rightarrow_{R/E \cup B}$
is achieved by a much simpler relation $\rightarrow_{R\!:\!E ,B}$
under the assumptions that: (i) $\vec{E}$ is convergent modulo $B$,
and (ii) the rules $R$ are ``coherent'' with $\vec{E}$ modulo $B$
\cite{viry-tcs,crc-alp}.  $\rightarrow_{R\!:\!E ,B}$ is
the composed relation $\rightarrow_{R\!:\!E ,B} \; = \;
\rightarrow_{\vec{E},B}! \; ; \; \rightarrow_{R,B} \; ; \;
\rightarrow_{\vec{E},B}!$.  One can think of $\rightarrow_{R\!:\!E ,B}$
 as executing the relation $\rightarrow_{R,B} \cup \rightarrow_{\vec{E},B}$ with an eager
strategy for $\rightarrow_{\vec{E},B}$.

In this paper, the state transitions that we will be interested
in will be deduction steps in inductive reasoning, so that the
``states'' are formulas in a deductive proof.  Two different
 kinds of rewrite theories will be particularly useful for this purpose.

First, we will use \emph{object-level} rewrite theories
associated to specific \emph{inductive goals} stating some property
of the initial algebra $T_{\Sigma/E \cup B}$ of an equational
theory $(\Sigma,E \cup B)$.  These theories
will be of the general form $\mathcal{R}=(\Sigma,E \cup B,H)$,
where $H$ are the \emph{induction hypotheses} associated to
the goal we want to prove, used as rewrite rules.
In their simplest form, hypotheses in $H$ may be equations
or conditional equations.  So, why not using instead the
equational theory $(\Sigma,E \cup H \cup B)$?  The key reason
is that the equations $E$ will always be assumed to be 
\emph{ground convergent} modulo $B$ when oriented as
rules $\vec{E}$, but in general the rules $H$ need not be
convergent.  Nevertheless, as discussed
in Section \ref{IndH-REW}, they may be oriented in
a manner ensuring their termination.  The key advantage of
the rewrite theory $\mathcal{R}=(\Sigma,E \cup B,H)$
versus the equational theory $(\Sigma,E \cup H \cup B)$,
is that $\mathcal{R}$, by sharply distinguishing between $E$ and $H$,
can consider a variety of  non-deterministic ways in which
the induction hypotheses $H$ can be applied.  This is
important for inductive theorem proving, because: 
 (i) in a terminating rewrite theory, 
there is in general no such thing as \emph{the} normal form
of a term; in general there will be a \emph{set} of such normal
forms; and (ii) in fact, the signature $\Sigma$ that will be used
will not be that of the given equational theory
$\Sigma$, but, as explained in Section \ref{EQ-PREDS},
its extension to a signature $\Sigma^{=}$ whose terms
denote \emph{quantifier-free} $\Sigma$-formulas.
That is,  $\mathcal{R}$ will not be rewriting $\Sigma$-terms,
but $\Sigma$-formulas, and it may very well be, that one
of the normal forms of the goal formula $\varphi$ we
are rewriting is $\top$, so that the goal is proved, while
another normal form $\varphi'$ has only made partial progress
towards proving $\varphi$.  By considering  the various
ways in which the induction hypotheses $H$ can be applied,
we maximize the chances of success in proving $\varphi$.

There is, however, a second way in which rewriting logic
will play an important role in this work \emph{at the meta-level}.
The point is that rewriting logic is a flexible \emph{logical framework}
in which many other logics can be naturally represented
\cite{rwl-fwk,20-years}.  In such representations,
an inference rule in a given logic $\mathcal{L}$ becomes
a possibly conditional rewrite rule.  In particular, the
inference rules in the inductive inference system presented in 
Section \ref{ind-inf-rules} will be the rewrite rules
of a \emph{meta-level} rewrite theory.  Why meta-level?
Because such inference rules are \emph{parametric}
on the given equational theory $(\Sigma,E \cup B)$
we are reasoning on and, furthermore,
may \emph{transform} $(\Sigma,E \cup B)$ into
another theory by adding, for example, induction hypotheses
to a new goal.  This is not just a theoretically pleasing way to
\emph{think} about the inductive inference system
of Section \ref{ind-inf-rules}, but
 an eminently practical way
to \emph{implement} such an inference system in Maude
using Maude's efficient support of meta-level reasoning
through its \texttt{META-LEVEL} module.  
 Most of the inference rules in this paper have  already
been  implemented in Maude's {\bf NuITP} in exactly this way.

\subsection{Narrowing in a Nutshell}

Narrowing modulo axioms $B$ generalizes
rewriting modulo $B$ in the following, natural
way.  When we rewrite a term $t$ at position $p$
with a rule $l \rightarrow r$, the subterm $t|_{p}$ 
\emph{matches} modulo $B$ the lefthand side $l$.
Instead, when we narrow a term $t$ at position $p$
with a rule $l \rightarrow r$ in $\vec{E}$, the subterm $t|_{p}$ 
\emph{unifies} modulo $B$ with the lefthand side $l$,
a more general notion which allows us to \emph{symbolically
evaluate} the term $t$ with the given oriented
equations $\vec{E}$. That is, ``narrowing'' is just technical
jargon for symbolic evaluation with equations.
 Here is the precise definition:
Given a ground convergent theory $\vec{\mathcal{E}}=(\Sigma,B,\vec{E})$,
the \emph{narrowing relation} $t \leadsto^{\alpha}_{\vec{E},B}u$
holds between $\Sigma$-terms $t$ and $u$
iff there is a non-variable position $p$ in $t$, a rule
$l \rightarrow r$ in $\vec{E}$, and
a $B$-unifier $\alpha$ of the equation $t|_{p}=l$
such that $u=(t[r]_{p})\alpha$, where we assume that the variables of
$l \rightarrow r$ have been renamed if necessary so that
no variables are shared with those of $t$.  Since in this
work we will consider ground convergent rewrite theories
that are sufficiently complete with respect to
free constructors modulo axioms $B_{\Omega} \subseteq B$,
we can discuss a key property of narrowing, called ``lifting,''  in this
specific setting.  For any non-constructor $\Sigma$ term $t$ and ground
constructor substitution $\rho$, the ground term $t\rho$
can always  be rewritten in one step to, say, $t\rho \rightarrow_{\vec{E},B}v$.
 But any such rewriting computation is \emph{covered} as
an instance by a corresponding narrowing step
$t \leadsto^{\alpha}_{\vec{E},B}u$, in the sense that there
is a ground constructor substitution $\mu$ such that
$u \mu =_{B} v$.  

Note that there may be infinitely many
such $\rho$'s instantiating $t$, and therefore infinitely
many one-step rewrites $t\rho \rightarrow_{\vec{E},B}v$, but
\emph{all of them are simultaneously covered as instances}
by a \emph{finite} number of one-step narrowing 
computations from $t$ of the form $t \leadsto^{\alpha}_{\vec{E},B}u$,  
assuming that: (i) $B$ has a finitary unification algorithm; and (ii) $\vec{E}$ is
a finite set. As we shall see in Section \ref{ind-inf-rules}, the
${\bf NS}$ and ${\bf NI}$ inference rules exploit this lifting
property of narrowing for inductive reasoning purposes.
More precisely, they do so using so-called \emph{constrained narrowing}
with possibly conditional rules $\vec{E}$, in the sense
of \cite{conditional-narrowing-SCP}.

\subsection{Equationally Defined Equality Predicates in a Nutshell} 
\label{EQ-PREDS}

Equationally-defined equality predicates \cite{DBLP:journals/scp/GutierrezMR15} achieve a
remarkable feat for QF formulas in
initial algebras under reasonable executability conditions:
they \emph{reduce} first-order logic
satisfaction of QF formulas in an initial algebra
$T_{\mathcal{E}}$ to \emph{purely equational reasoning}.  This is achieved
by a theory transformation\footnote{In
  \cite{DBLP:journals/scp/GutierrezMR15} the equality predicate is
  denoted $\_\sim\_$, instead of the standard notation
  $\_=\_$.  Here we use  $\_=\_$ throughout.  This has the
  pleasant effect that a QF formula $\varphi$ is both
a formula and a Boolean expression, which of course amounts to mechanizing
by equational rewriting the Tarskian semantics of QF formulas in first-order-logic for
initial algebras.}
 $\mathcal{E} \mapsto \mathcal{E}^{=}$ such that, provided: (i)
$\mathcal{E} = (\Sigma,E \cup B)$, with $B$ any combination
of associativity and/or commutativity axioms, is ground convergent 
and operationally terminating modulo $B$, and 
is sufficiently complete with respect to a subsignature $\Omega$
of constructors such that $T_{\mathcal{E}}|_{\Omega} \cong
T_{\Omega/B_{\Omega}}$, with $B_{\Omega} \subseteq B$, then:
(ii) $\mathcal{E}^{=}$ is ground convergent operationally terminating
and sufficiently complete and protects\footnote{That
is,  $T_{\mathcal{E}^{=}}|_{\Sigma} \cong T_{\mathcal{E}}$, and
there is a subtheory inclusion $\mathcal{B} \subseteq \mathcal{E}^{=}$,
with $\mathcal{B}$ having signature $\Sigma_{\mathcal{B}}$ and only sort
 $\mathit{NewBool}$ such that: (i)
 $T_{\mathcal{B}}$ the initial algebra of the Booleans,
and (ii) $T_{\mathcal{E}^{=}}|_{\Sigma_{\mathcal{B}}} \cong T_{\mathcal{B}}$.}
both $\mathcal{E}$ and 
 a new copy of the Booleans,
of sort $\mathit{NewBool}$, where true and false are respectively
denoted
$\top$, $\bot$, conjunction and disjunction are respectively denoted
$\wedge$, $\vee$, negation is denoted $\neg$, and a
QF $\Sigma$-formula $\varphi$ is a term of sort $\mathit{NewBool}$.
Furthermore, for any ground QF $\Sigma$-formulas $\varphi,\psi$
we have:
\[T_{\mathcal{E}} \models \varphi 
\;\;\; \;\;\; \mathit{iff} \;\;\; \;\;\;
\varphi!_{\vec{\mathcal{E}}^{=}}= \top.
\]
and, in particular,
\[T_{\mathcal{E}} \models \varphi \;\; \Leftrightarrow \;\; 
  \psi
\;\;\; \;\;\; \mathit{iff} \;\;\; \;\;\;
(\varphi \Leftrightarrow  \psi)!_{\vec{\mathcal{E}}^{=}}= \top
\;\;\; \;\;\; \mathit{iff} \;\;\; \;\;\;
T_{\mathcal{E}^{=}} \models \varphi = \psi
\]
where $\varphi = \psi$ is an equality of Boolean terms of sort
$\mathit{NewBool}$.  Therefore, one can \emph{decide} both the validity 
and the semantic equivalence of ground QF $\Sigma$-formulas in
$T_{\mathcal{E}}$ by reducing them to canonical form with the ground convergent rules in
$\vec{\mathcal{E}}^{=}$.  In particular, and this is a  property that
will be systematically exploited in Section \ref{ind-inf-rules}, 
for any QF $\Sigma$-formula  $\varphi$, possibly  with variables, we have
$T_{\mathcal{E}} \models \varphi \Leftrightarrow (\varphi!_{\vec{\mathcal{E}}^{=}})$, where
$\varphi!_{\vec{\mathcal{E}}^{=}}$ may be a much simpler formula, 
sometimes just $\top$ or $\bot$. Since the 
$\mathcal{E} \mapsto \mathcal{E}^{=}$ transformation excludes identity axioms
from $\mathcal{E}$, one lingering doubt is what to do when
$\mathcal{E}$ has also identity axioms $U$.  The answer is that we can
use the semantics-preserving theory transformation $\mathcal{E}
\mapsto \mathcal{E}_{U}$
defined in \cite{frocos09}, which
 turns identity axioms $U$ into rules $\vec{U}$ and  preserves ground convergence,
to reduce to the case $U= \emptyset$,
provided we have $T_{\mathcal{E}_{U}}|_{\Omega} \cong
T_{\Omega/B_{\Omega}}$.

 Since the rewrite theory 
$\vec{\mathcal{E}}_{U}^{=}$
will be a key workhorse for the 
simplification rules described in Section
\ref{ind-inf-rules}, let
us describe it in more detail, as well as  the above-mentioned 
$\vec{\mathcal{E}}
\mapsto \vec{\mathcal{E}}_{U}$ transformation
with which the $\vec{\mathcal{E}} \mapsto \vec{\mathcal{E}}^{=}$
transformation is composed.
That is, let us look in more detail at the
composition of transformations:
\[\vec{\mathcal{E}}
\mapsto \vec{\mathcal{E}}_{U} \mapsto \vec{\mathcal{E}}_{U}^{=}.
\]
The rewrite theory $\vec{\mathcal{E}}_{_U}$ 
is defined as follows.  $\mathcal{E}$'s axioms $B$
can be decomposed as
$B=B_0\cup U$, where $U$ are the unit axioms,
and $B_0$ are the remaining associative and/or commutative axioms.
Then,
$\vec{\mathcal{E}}_{_U}=
(\Sigma,B_0,\vec{E}_U\cup\vec{U})$
is the semantically equivalent rewrite theory obtained from 
$\vec{\mathcal{E}} =(\Sigma,B,\vec{E})$
by: (i) adding
the axioms $U$ as rules $\vec{U}$; and (ii)
transforming the rules $\vec{E}$ into $\vec{E}_U$
by mapping each rule $(l \rightarrow r) \in \vec{E}$ 
to the set of rules $\{l_{i} \rightarrow r \alpha_{i} \mid 1 \leq i
\leq n\}$, where $\{(l_{i},\alpha_{i})\}_{1 \leq i \leq n}$
is the finite set of $\vec{U},B_0$-\emph{variants} of
$l$  (the notion of variant is explained
in Section \ref{var-sat-nut}).
  For example, if $\_,\_$ is an ACU multiset union
operator of sort $\mathit{MSet}$
with identity $\emptyset$ and  with subsort 
$\mathit{Elt}$ of elements, a membership
rewrite rule $x \in x,S \rightarrow \mathit{true}$
modulo ACU with $x$ of sort $\mathit{Elt}$
and $S$ of sort $\mathit{MSet}$
 is mapped to the set of rules
$\{x \in x \rightarrow \mathit{true}, x \in x,S \rightarrow \mathit{true}\}$
modulo AC.
Since these theories are semantically equivalent, we have
$T_{\mathcal{E}} \cong
T_{\mathcal{E}_{U}}$.
The second step
$\vec{\mathcal{E}}_{U} \mapsto \vec{\mathcal{E}}_{U}^{=}$
then adds the equationally-defined equality predicates
to $\vec{\mathcal{E}}_{U}$ in the usual manner
specified in \cite{DBLP:journals/scp/GutierrezMR15}.
The easiest equality predicate defining equations to illustrate 
are those for free constructors.  For example, for the naturals
in Peano notation they are: $(n = n) = \mathit{true}$, $(s(n) = s(m)) = (n =
m)$, and $(s(n)=0)= \mathit{false}$ ($\_=\_$  is declared as 
a \emph{commutative} operator of (fresh) sort $\mathit{NewBool}$ ).

\subsection{Order-Sorted Congruence Closure in a Nutshell} \label{OS-CC}

Let $(\Sigma,B)$ be an order-sorted theory where the
axioms $B$ are only associativity-commutativity (AC) axioms
and $\Sigma$ is $B$-preregular.  Now let $\Gamma$ be 
a set of ground $\Sigma$-equations.  The question is:
is $B \cup \Gamma$-equality \emph{decidable}?
(when $\Sigma$ has just a binary AC operator, this is called
the ``word problem for commutative semigroups'').
The answer, provided in \cite{DBLP:conf/fossacs/Meseguer16},
is yes!  We can perform a ground Knuth-Bendix
completion of $\Gamma$ into an equivalent (modulo $B$) set of
ground rewrite rules $\mathit{cc}^{\succ}_{B}(\Gamma)$ that is
convergent modulo $B$, so that 
$t =_{B \cup \Gamma} t'$ iff
$t!_{\vec{\mathcal{E}}_{\mathit{cc}^{\succ}_{B}(\Gamma)}} =_{B^{\square}}
t'!_{\vec{\mathcal{E}}_{\mathit{cc}^{\succ}_{B}(\Gamma)}}$,
where $\vec{\mathcal{E}}_{\mathit{cc}^{\succ}_{B}(\Gamma)}$
is the rewrite theory
$\vec{\mathcal{E}}_{\mathit{cc}^{\succ}_{B}(\Gamma)}=(\Sigma^{\square},B^{\square},
\mathit{cc}^{\succ}_{B}(\Gamma))$, with
$\Sigma^\square$ the ``kind completion'' of $\Sigma$,
which is automatically computed by Maude by adding a
so-called ``kind'' sort $\top_{[s]}$ above each connected component
$[s] \in \widehat{S}$ of $(S,\leq)$ and lifting each
operation $f:s_1\cdots s_n\rightarrow s$ to its kinded version
$f : \top_{[s_1]} \cdots \top_{[s_n]} \rightarrow \top_{[s]}$,
and where $B^{\square}$ is obtained from $B$ by replacing each variable
of sort $s$ in $B$ by a corresponding variable of sort $\top_{[s]}$.
The symbol $\succ$ in $\mathit{cc}^{\succ}_{B}(\Gamma)$ is a total
well-founded order on ground terms modulo $B$ that is used to orient
the equations into rules.  In all our uses $\succ$ will be
an AC \emph{recursive path order}
(RPO)\footnote{The \emph{recursive path order} (RPO) is a
  well-founded simplification order on terms parametric on an order
on function symbols  (see, e.g., \cite{baader-nipkow}).
RPO is a \emph{total} order on terms if the order on symbols is so.
RPO has been extended to RPO modulo AC in various papers, including
\cite{DBLP:journals/iandc/Rubio02}.}
based on a total order on function symbols
\cite{DBLP:journals/iandc/Rubio02}.  The need to
extend $\Sigma$ to $\Sigma^\square$ is due to the fact that
some terms in $\mathit{cc}^{\succ}_{B}(\Gamma)$ may be
$\Sigma^\square$-terms that fail to be $\Sigma$-terms.

Extending the above congruence closure framework
from $AC$ axioms to axioms $B$ that contain any combination of
 associativity and/or
commutativity axioms is quite smooth, but requires a crucial caveat:
if some operator $f \in \Sigma$ is only associative, then
$\mathit{cc}^{\succ}_{B}(\Gamma)$ may be an infinite set that cannot
be computed in practice.  This is due to the  undecidability of
the ``word problem for semigroups.''  The Maude implementation
of $\mathit{cc}^{\succ}_{B}(\Gamma)$ used in the {\bf NuITP}
supports this more general combination of axioms $B$; but when some
$f \in \Sigma$ that is only associative appears somewhere in $\Gamma$, 
a bound is imposed on the number of iterations of the ground completion cycle.
This means
that, if the completion process has not terminated before the bound is
reached, the above decidability result does not hold.  However,
for  inductive simplification purposes it is enough to obtain a
set of ground rules $\mathit{cc}^{\succ}_{B}(\Gamma)$ that is
guaranteed to be \emph{terminating} modulo $B$, and that, thanks to the, perhaps partial,
completion, ``approximates convergence'' much better than the original
$\Gamma$.

\subsection{Contextual Rewriting in a Nutshell} \label{CTX-RW}

Let $(\Sigma,B)$ be an order-sorted theory where the
axioms $B$  contain any combination of
 associativity and/or
commutativity axioms.  What can we do to prove that in $(\Sigma,B)$
an implication of the form $\Gamma \rightarrow u = v$, with variables
$\mathit{vars}(\Gamma \rightarrow u = v)=X$ and $\Gamma$ a conjunction
of equations, is valid?  We can:
 (i) add to $\Sigma$ a set of fresh new constants
$\overline{X}$ obtained from $X$ by changing each $x \in X$ into a
constant $\overline{x} \in \overline{X}$ of same sort as $x$, (ii)
replace the conjunction $\Gamma$ by the ground conjunction
$\overline{\Gamma}$ obtained by replacing each $x \in X$ in $\Gamma$ by its
corresponding $\overline{x} \in \overline{X}$, and obtaining likewise
the ground equation $\overline{u} = \overline{v}$.  
By the Lemma of Constants and the Deduction Theorem we have \cite{tarquinia}:
\[(\Sigma,B) \vdash \Gamma \rightarrow u = v
\;\;\; \Leftrightarrow \;\;\;
(\Sigma(\overline{X}), \cup B \cup \{\overline{\Gamma}\}) \vdash \overline{u} = \overline{v}
\]
where $\Sigma(\overline{X})$ is obtained from $\Sigma$ by adding the
fresh new constants $\overline{X}$, and
$\{\overline{\Gamma}\}$ denotes the set of ground equations
associated to the conjunction
$\overline{\Gamma}$.
 But, disregarding the difference
between $\overline{\Gamma}$ and $\{\overline{\Gamma}\}$,
and realizing that
$\mathit{cc}^{\succ}_{B}(\overline{\Gamma})$ is equivalent modulo
$B^{\square}$ to $\overline{\Gamma}$, if we can prove
$\overline{u}!_{\vec{\mathcal{E}}_{\mathit{cc}^{\succ}_{B}(\overline{\Gamma})}}
=_{B^{\square}}
\overline{v}!_{\vec{\mathcal{E}}_{\mathit{cc}^{\succ}_{B}(\overline{\Gamma})}}$,
then we have proved
$(\Sigma(\overline{X}), \cup B \cup \{\overline{\Gamma}\}) \vdash
\overline{u} = \overline{v}$ and therefore
$(\Sigma,B) \vdash \Gamma \rightarrow u = v$, where
$\vec{\mathcal{E}}_{\mathit{cc}^{\succ}_{B}(\overline{\Gamma})}=(\Sigma(\overline{X})^{\square},B^{\square},
\mathit{cc}^{\succ}_{B}(\overline{\Gamma}))$.  Furthermore, if
$\vec{\mathcal{E}}_{\mathit{cc}^{\succ}_{B}(\overline{\Gamma})}$ is
\emph{convergent} (this may only fail to be the case if some
$f \in \Sigma$ is associative but not commutative) this is an
\emph{equivalence}: $(\Sigma,B) \vdash \Gamma \rightarrow u = v$ iff
$\overline{u}!_{\vec{\mathcal{E}}_{\mathit{cc}^{\succ}_{B}(\overline{\Gamma})}}
=_{B^{\square}}
\overline{v}!_{\vec{\mathcal{E}}_{\mathit{cc}^{\succ}_{B}(\overline{\Gamma})}}$,
and therefore a decision procedure.
Rewriting with
$\vec{\mathcal{E}}_{\mathit{cc}^{\succ}_{B}(\overline{\Gamma})}$ is
called \emph{contextual rewriting} \cite{DBLP:journals/fuin/Zhang95}, since we are using the ``context''
$\overline{\Gamma}$ suitably transformed into
$\mathit{cc}^{\succ}_{B}(\overline{\Gamma})$.  Many increasingly more
powerful variations on this method are possible.  For example, we may
replace $(\Sigma,B)$ by $\mathcal{E}=(\Sigma,E \cup B)$, with
$\vec{\mathcal{E}}$ ground convergent and then rewrite  $\overline{u}
= \overline{v}$ not only with
$\mathit{cc}^{\succ}_{B}(\overline{\Gamma})$ but also with
$\vec{\mathcal{E}}$.  Likewise, when performing inductive reasoning
we may consider not just a ground
equation $\overline{u} = \overline{v}$, but a ground QF formula
$\overline{\varphi}$ (whose premise is $\overline{\Gamma}$),
 and rewrite $\overline{\varphi}$ not only with
$\mathit{cc}^{\succ}_{B}(\overline{\Gamma})$ but also with
$\vec{\mathcal{E}_{U}}^{=}$.

\subsection{Variant Unification and Variant Satisfiability in a Nutshell} \label{var-sat-nut}

Consider an order-sorted equational theory $\mathcal{E}=(\Sigma,E \cup
B)$ such that $\vec{\mathcal{E}}$ is ground convergent and
suppose we have a constructor subspecification
$(\Omega,B_{\Omega},\emptyset) \subseteq (\Sigma,B,\vec{E})$,
so that $T_{\mathcal{E}}|_{\Omega} \cong T_{\Omega/ B_{\Omega}}$.
Suppose, further, that we have a subtheory
$\mathcal{E}_{1} \subseteq \mathcal{E}$
such that: (i) $\vec{\mathcal{E}}_{1}$ is convergent and has the finite
variant property\footnote{An 
$\vec{\mathcal{E}}_{1}$-\emph{variant}
(or $\vec{E}_{1},B_{1}$-\emph{variant}) of a 
$\Sigma_{1}$-term $t$ is a pair $(v,\theta)$, where
$\theta$ is a substitution in canonical form, i.e.,
$\theta = \theta!_{\vec{\mathcal{E}}_{1}}$, and
$v =_{B_{1}} (t\theta)!_{\vec{\mathcal{E}}_{1}}$.
$\vec{\mathcal{E}_{1}}$ is FVP iff any such $t$ has
a finite set of variants
$\{(u_{1},\alpha_{1}),\ldots, (u_{n},\alpha_{n})\}$ which are
``most general possible'' in the precise
sense that for any variant $(v,\theta)$  of $t$
there exist $i$, $1 \leq i \leq n$, and
substitution $\gamma$ such that:
(i) $v =_{B_{1}} u_{i}\gamma$, and (ii) 
$\theta =_{B_{1}} \alpha_{i}\gamma$.}
 (FVP) \cite{variant-JLAP}, 
(ii) $\vec{\mathcal{E}_{1}}$ can be ``sandwiched''
between $\vec{\mathcal{E}}$ and the constructors
as $(\Omega,B_{\Omega},\emptyset) \subseteq (\Sigma_{1},B_{1},\vec{E}_{1})
\subseteq (\Sigma,B,\vec{E})$, 
(iii) $B_{1}$ can involve any combination 
of associativity and/or commutativity and/or identity 
axioms;\footnote{\label{var-A-foot} The notion of 
an FVP theory has been recently extended  in
  \cite{DBLP:journals/jlap/Meseguer23} 
to allow also axioms
of associativity without commutativity in $B_{1}$.}
and  (iv)
$T_{\mathcal{E}}|_{\Sigma_{1}} \cong T_{\mathcal{E}_{1}}$,
which forces $T_{\mathcal{E}_{1}}|_{\Omega} \cong T_{\Omega/B_{\Omega}}$.

Then, 
if $\Gamma$ is a conjunction of $\Sigma_{1}$-equations, since
$T_{\mathcal{E}}|_{\Omega} \cong T_{\Omega/ B_{\Omega}}$,
a ground $\mathcal{E}$-unifier $\rho$ of $\Gamma$ is always 
$\mathcal{E}$-equivalent to its normal form
$\rho!_{\vec{\mathcal{E}}}$, which, by definition,
is the substitution
$\lambda x \in \mathit{dom}(\rho).\; \; \rho(x)!_{\vec{\mathcal{E}}}$.
Therefore, by ground convergence and sufficient completeness,
 $\rho!_{\vec{\mathcal{E}}}$ is a ground
$\Omega$-substitution, that is, a \emph{constructor}
ground $\mathcal{E}$-unifier of $\Gamma$.
But since $\Omega \subseteq \Sigma_{1}$
and $T_{\mathcal{E}}|_{\Sigma_{1}} \cong T_{\mathcal{E}_{1}}$,
which implies $C_{\Sigma/E,B}|_{\Sigma_{1}} \cong C_{\Sigma_{1}/E_{1},B_{1}}$,
this makes $\rho!_{\vec{\mathcal{E}}}$ a \emph{constructor}
ground $\mathcal{E}_{1}$-unifier of $\Gamma$.
But, under the assumptions for $B_{1}$, by the results in
\cite{var-sat-scp,skeirik-meseguer-var-sat-JLAMP,DBLP:journals/jlap/Meseguer23}
one can compute a complete,
finite\footnote{When $B_{1}$ contains associativity without 
commutativity axioms, the number of unifiers in 
$\mathit{Unif}^{\Omega}_{\mathcal{E}_{1}}(\Gamma)$ may be infinite.
However, Maude's $B_{1}$-unification algorithm
used in the process of computing $\mathit{Unif}^{\Omega}_{\mathcal{E}_{1}}(\Gamma)$
 will always produce a finite set of solutions, together
with an incompleteness warning in case some solutions are missing.
This warning will imply that the finite set of unifiers
 $\mathit{Unif}^{\Omega}_{\mathcal{E}_{1}}(\Gamma)$
is incomplete.}
set $\mathit{Unif}^{\Omega}_{\mathcal{E}_{1}}(\Gamma)$
of constructor $\mathcal{E}_{1}$-unifiers of $\Gamma$,
so that any constructor 
ground $\mathcal{E}_{1}$-unifier of $\Gamma$, and therefore up to 
$\mathcal{E}$-equivalence any ground $\mathcal{E}$-unifier of
$\Gamma$, is an instance of a unifier in $\mathit{Unif}^{\Omega}_{\mathcal{E}_{1}}(\Gamma)$.

Note, furthermore, that if  $B_{1}$ does not contain 
 associativity without  commutativity axioms,
$(\Omega,B_{\Omega})$ is an OS-compact theory \cite{var-sat-scp}
(resp. a weakly OS-compact theory if $B_{1}$ does contain associativity without 
commutativity axioms and $(\Omega,B_{\Omega})$ satisfies some additional assumptions
 \cite{DBLP:journals/jlap/Meseguer23}).
Therefore, again by \cite{var-sat-scp} (resp.  \cite{DBLP:journals/jlap/Meseguer23})
satisfiability (and therefore validity) of any
QF $\Sigma_{1}$-formula in $T_{\mathcal{E}_{1}}$, and by
$T_{\mathcal{E}}|_{\Sigma_{1}} \cong T_{\mathcal{E}_{1}}$ also in
$T_{\mathcal{E}}$, is \emph{decidable}
(resp. has a partial decidability algorithm under some 
assumptions as explained in \cite{DBLP:journals/jlap/Meseguer23}).
In the inductive inference system
presented in Section \ref{ind-inf-rules}
this decidability is exploited by the
{\bf VARSAT} simplification rule.

\subsection{Generator Sets and Function Subcalls} \label{GS-FubCall-Subsection}

Generator sets generalize standard structural induction on the
constructors of a sort, in the sense that structural
induction inducts on a specific generator set.
They are particularly useful for inductive
reasoning when constructors obey structural axioms $B$ including
associativity or associativity-commutativity for which structural
induction may prove ill-suited; but, even for free
constructors,  structural
induction may still be ill-suited for inducting on a  variable under
a function equationally defined in  a manner other than
by standard primitive recursion.  

A \emph{generator set} for a sort $s$ is just a set of constructor terms
patterns ---of sort $s$ or smaller--- such that, up to $B$-equality,
 any ground constructor
term of sort $s$ is a ground substitution instance of one of the
patterns in the generator set.  For example, structural induction on
the Peano natural numbers is achieved by the generator set $\{0,s(n)\}$
 for sort  $\mathit{Nat}$, but the alternative generator set
$\{0,s(0),s(s(n))\}$ may be better suited to reason about functions
defined by the two ``base cases'' $0$ and $s(0)$, and inducting on
$s(s(n))$, such as, for example, addition defined by the equations: $x+0=x,\;\;
x+s(0)=s(x),\;\; x+s(s(y)) = s(s(x+y))$, which can add numbers roughly
twice as fast as the standard definition of addition.  

The usefulness of generator sets is even greater
for constructors that satisfy structural axioms
such as $A$ or $AC$.  Let us see an example.

\begin{example} \label{multiset-ctors}
Let us consider a generator set for $\_\cup\_$
a (subsort-overloaded) \emph{associative-commutative} ($AC$)
 multiset union operator on multisets of elements, involving
subsort inclusions $\mathit{Elt} < \mathit{NeMSet} < \mathit{MSet}$,
with constructors $\Omega$ including $\emptyset$
of sort $\mathit{MSet}$ and the overloading 
$\_\cup\_ : \mathit{NeMSet} \;\; \mathit{NeMSet} \rightarrow \mathit{NeMSet}$
of $\_\cup\_$ for
the sort $\mathit{NeMSet}$ of non-empty multisets,
whereas the overloading 
$\_\cup\_ : \mathit{MSet} \;\; \mathit{MSet} \rightarrow \mathit{MSet}$
of  $\_\cup\_$  for the sort $\mathit{MSet}$
 is a \emph{defined symbol}, namely,
defined by the equation:
$S \cup \emptyset = S$,
with $S$ of sort $\mathit{MSet}$.
 Then,
the set $\{\emptyset,x,U \cup
V\}$, with $x$ of sort $\mathit{Elt}$ and $U,V$ of sort $\mathit{NeMSet}$,
can be seen as the generator set for sort $\mathit{MSet}$
corresponding to
structural induction.  But the generator set 
$\{\emptyset,x,y \cup U\}$, with $y$ also of sort $\mathit{Elt}$,
is much better suited to reason about a multiset membership
predicate $\_\in\_$ defined by equations including, among others,
$x \in \emptyset = \mathit{false}$, 
$x \in x = \mathit{true}$,
$x \in x \cup U = \mathit{true}$.
\end{example}

If the set $B$ of
structural axioms among constructors $\Omega$
decomposes as a disjoint union $B=B_{0} \uplus  U$,
with $B_{0}$ associative and/or commutative axioms,
and $U$ unit (identity) axioms,  it is better to consider
generator sets modulo $B_{0}$, instead of modulo $B$.  The reason for
this choice is that the axioms $B_{0}$ are term-size-preserving,
but the axioms $U$ are not so, and this can cause technical
complications when reasoning inductively about proper constructor subterms
of a term.  For example, if we had added an axiom $U_{\emptyset}$
 for $\emptyset$ as
a \emph{unit element} for $\_\cup\_$ instead of specifying
$S \cup \emptyset = S$ as a
 defining equation,
then $\{\emptyset,x\cup U\}$ would
indeed be a generator set modulo $\mathit{ACU}$ for sort $\mathit{MSet}$,
 but a term like
$a \cup \emptyset$  \emph{collapses} modulo $B$
to its proper subterm $a$, and is, for that reason, problematic for
inductive reasoning on \emph{subterms}.  Here is the precise
definition that will be used
for inductive reasoning purposes:

\begin{definition}  For $\Omega$ an order-sorted signature
of constructors satisfying axioms $B$ decomposable
as $B=B_{0} \uplus  U$, as explained above, and $s$ a sort in
$\Omega$, a  $B_{0}$-\emph{generator set} for sort $s$ is
a finite set of terms $\{u_{1},\ldots, u_{k}\}$, with
$u_{1},\ldots, u_{k} \in T_{\Omega}(X)_{s}$ and such that
\[ T_{\Omega/B_{0} ,s}=\{[u_{i} \;\rho] \in  T_{\Omega/B_{0} ,s} \mid 1 \leq i \leq k,\;\;\; \rho \in
[X \rightarrow T_{\Omega}]\}.
\]
\end{definition}

In practice one may use \emph{several} generator sets for the
same sort $s$, depending on the various functions that one may 
wish to reason inductively about.  Furthermore,
 one may be able to generate such sets
automatically from the function definitions 
themselves,\footnote{\label{collapse-foot} More
precisely, from their definition, not in the original theory
$\mathcal{E}$, but in the transformed theory
$\mathcal{E}_{U}$, making sure that no lefthandside
for an equation defining $f$ collapses
under the $\mathcal{E} \mapsto \mathcal{E}_{U}$
transformation to a lefthand side not topped by $f$
due to an $U$-axiom for $f$.}
assuming that defined functions
have been proved sufficiently complete with respect to $\Omega$.

\vspace{2ex}

\noindent {\bf Checking the Correctness of Generator Sets}. 
How can we \emph{know} that a proposed generator
set $\{u_{1},\ldots, u_{k}\}$ it truly one modulo axioms
$B_{0}$ for a given sort $s$
and constructors $\Omega$?  Assuming that the
terms $u_{1},\ldots, u_{k}$ are all \emph{linear}, i.e., have no
repeated variables ---which is the usual case for generator sets---
this check can be reduced to an automatic \emph{sufficient
  completeness check} with Maude's Sufficient Completeness Checker
 (SCC) tool 
\cite{hendrix-meseguer-ohsaki-ijcar06}, which is based on
tree automata decision procedures 
modulo axioms $B_{0}$.  The reduction is extremely
simple:  define a new unary predicate $s: s \rightarrow \mathit{Bool}$
with equations $s(u_{i}) = \mathit{true}$, $1 \leq i \leq k$.
Then, $\{u_{1},\ldots, u_{k}\}$  is a correct generator set for
sort $s$ modulo  $B_{0}$ for the constructor signature
$\Omega$ iff the predicate $s$ is sufficiently complete,
which can be automatically checked by the SCC tool.
Furthermore, if $\{u_{1},\ldots, u_{k}\}$
is \emph{not} a generator set
for sort $s$, the  SCC tool will output a useful counterexample.

\vspace{2ex}

\noindent {\bf Function Subcalls}. Function definitions usually involve
recursive (and possibly conditional) equations
oriented as rules of the form:
\[[l]:
f(\vec{u}) \rightarrow  t[f(\vec{u_{1}}),\ldots,
f(\vec{u_{k}})]_{p_{1},\ldots,p_{k}}\;\;\; \mathit{if} \;\;
\mathit{cond}[f(\vec{v_{1}}),\ldots, f(\vec{v_{k'}})]_{q_{1},\ldots,q_{k'}}
\]
where all the terms in $\vec{u}, \vec{u_{1}},\ldots, \vec{u_{k}}$
and $\vec{v_{1}},\ldots, \vec{v_{k'}}$ are
constructor terms, and
 at positions $p_{1},\ldots, p_{k}$ in the righthand side $t$, and
$q_{1},\ldots, q_{k'}$ in the
condition $\mathit{cond}$, with $k+k' \geq 1$,
new calls $f(\vec{u_{1}}),\ldots, f(\vec{u_{k}})$
and $f(\vec{v_{1}}),\ldots, f(\vec{v_{k'}})$
to $f$, which we shall call the \emph{subcalls 
 with constructor arguments}\footnote{Rule $[l]$ may also  have
other subcalls where some arguments are terms with defined
symbols.}  of the rule,
appear.   In reasoning inductively about $f$ after applying
such a rule, it can be highly advantageous to have
\emph{inductive hypotheses} available for such subcalls.

For the same technical reasons already mentioned, of avoiding 
$U$-collapses of a constructor term to a proper subterm,
we shall consider such rules $[l]$, not in the
original ground convergent rewrite theory $\vec{\mathcal{E}}$
of interest, but in the semantically-equivalent
transformed theory\footnote{However, as mentioned
in Footnote \ref{collapse-foot}, we should check
that in the 
$\vec{\mathcal{E}} \mapsto \vec{\mathcal{E}}_{U}$ 
transformation,
no lefthandside
of a rule in $\vec{\mathcal{E}}$ defining $f$ 
has as one of its transformed rules 
a rule whose
lefthand side is not topped by $f$
due to an $U$-axiom for $f$.}
 $\vec{\mathcal{E}}_{U}$.
Under such assumptions, 
$\mathit{SC}([l])$ will denote the set
\[
\mathit{SC}([l])=\{f(\vec{u_{1}}),\ldots, f(\vec{u_{k}}),
f(\vec{v_{1}}),\ldots, f(\vec{v_{k'}})\}
\]
of \emph{subcalls with constructor arguments} of rule $[l]$.

Intuitively, we should be able to induct on a subcall  $f(\vec{v})$ of a main
call $f(\vec{u})$ if $f(\vec{v})$ is \emph{smaller} than $f(\vec{u})$
 just in the same way than, if $s(x)$ belongs to a generator set for
the sort $\mathit{Nat}$, we can induct on the smaller term $x$.
This leads us to the notion of a (proper) \emph{subterm subcall} of
a main call $f(\vec{u})$.  For example, $f(s(x),s(y))$ may be the
main call in  the lefthand side of one of the rules defining $f$,
which may have two subcalls in its righthand side and/or its
condition, say, $f(x,s(y))$, and $f(x,x)$.  Then, $f(x,s(y))$
\emph{is}  a (proper) subterm subcall of $f(s(x),s(y))$ ---because
$x$ is a proper subterm of $s(x)$, and $s(y)$
is a (non-proper) subterm of $s(y)$--- but $f(x,x)$ is
\emph{not} a subterm subcall of $f(s(x),s(y))$ ---because $x$ is
not a subterm of $s(y)$.

In the presence of structural axioms $B_{0}$
of associativity and/or commutativity among
constructors, the notions of:
(1) (proper)  subterm (modulo  $B_{0}$); and (2)
(proper)
subterm subcall  (modulo  $B_{0}$) are somewhat more subtle.  
We first need to recall the well-founded immediate subterm
relation $\lhd$, where, for any $\Sigma$-term $f(t_{1},\ldots,t_{n})$,
the relation
$t_{i} \lhd f(t_{1},\ldots,t_{n})$ holds, by definition,
for $1 \leq i \leq n$, as well as its transitive,  $\lhd^{+}$,
  and reflexive-transitive,  $\lhd^{*}$, closures.

\begin{definition} \label{B0-subterm-def}
Assuming that $B_{0}$ are axioms of associativity and/or
 commutativity, a term $u$ is called a $B_{0}$-\emph{proper subterm}
of $v$, denoted  $u \lhd_{B_{0}}^{+} v$, 
 iff there exists a term $v'$
such that $v =_{B_{0}} v'$ and $u \lhd^{+} v'$.
The set $\mathit{PST}_{B_{0}}(v)$ of all
\emph{proper} $B_{0}$-\emph{subterms} of 
the $\Sigma$-term $v$ is then,
$\mathit{PST}_{B_{0}}(v)=\{u \mid  u \lhd_{B_{0}}^{+} v \}$.
\end{definition}

\begin{remark} \label{B0-subterm-remark}
Note that the set $\mathit{PST}_{B_{0}}(t)$ can be quite large.
In practice, what we are interested in is the much
smaller set $\{[v] \mid  v \in \mathit{PST}_{B_{0}}(t) \}$,
where $[v]$ denotes the $B_{0}$-equivalence class
of $[v]$.  That is, we identify any two proper
$B_{0}$-subterms of $t$ up to $B_{0}$-equality.
Without further ado, 
$\mathit{PST}_{B_{0}}(t)$ will be understood in this sense, that is,
as a \emph{choice} in the equivalence class
$[v]$ of any $v \in \mathit{PST}_{B_{0}}(t)$,
with all other equivalent choices discarded as redundant.
\end{remark}

We can now characterize notion (2), namely, that of a proper
subterm subcall,  as follows:

\begin{definition} Let $\vec{\mathcal{E}}$ be a ground convergent
  rewrite theory with
 $\vec{\mathcal{E}}_{U}$ its semantically equivalent
transformed theory where identity axioms in $\mathcal{E}$ have been
transformed into rules in $\vec{\mathcal{E}}_{U}$,
so that $\vec{\mathcal{E}}_{U}$ only satisfies
associativity and/or commutativity axioms $B_{0}$.
Let $[l]$ be a possibly
conditional rule in  $\vec{\mathcal{E}}_{U}$ with constructor
argument calls whose lefthand side is $f(\vec{u})$.  Then, the
set $\mathit{SSC}([l])$ of its (proper) \emph{subterm subcalls} is, by
definition, the set:
\[
\mathit{SSC}([l]) = 
\{ f(\vec{w}) \in  \mathit{SC}([l]) \mid
\exists f(\vec{u'})\; s.t. \; f(\vec{u})=_{B_{0}} f(\vec{u'}) 
\wedge \vec{w} \lhd^{+} \vec{u'}\}
\]
where, by definition, if  $\vec{w} = w_{1},\ldots,w_{n}$
and $\vec{u'} = u'_{1},\ldots,u'_{n}$, then
$\vec{w} \lhd^{+} \vec{u'}$ iff (i) $w_{i} \lhd^{*} u'_{i}$, $1 \leq i \leq n$, and (ii) exists $j$, 
$1 \leq j \leq n$, such that $w_{j} \lhd^{+} u'_{j}$.
\end{definition}

\noindent Consider, for example, a Boolean predicate 
$\mathit{ordered}$
on lists of natural numbers with an \emph{associative} ($A$) constructor symbol
$\_\cdot\_$  involving, among others, the rule:
\[[l]:
\mathit{ordered}(n \; \cdot m \; \cdot \;  L) \rightarrow 
\mathit{if} \;\; n \leq m \;\;
\mathit{then} \;\; \mathit{ordered}(m \; \cdot \; L)
\;\; \mathit{else} \; \; \mathit{false} \;\; \mathit{fi}.
\]
Then, $\mathit{SSC}([l]) =\mathit{SC}([l]) = \{ \mathit{ordered}(m \; \cdot L)\}$.

\subsection{Ordered Rewriting} \label{OR-section}

In inductive theorem proving one should try to use induction
hypotheses as much as possible to \emph{simplify} conjectures.  But this raises two
obvious questions: (1) In what sense should a conjecture be made
\emph{simpler}? (2) \emph{How} can an induction hypothesis
be used for this purpose?  There are various answers given to 
questions (1) and (2).  Ordered rewriting 
answers them as follows: Question (1) is answered by suggesting
that a conjecture $\varphi$ should become simpler by
being transformed into a \emph{smaller} conjecture $\varphi'$, in the
sense that $\varphi \succ \varphi'$, for $\succ$ a suitable reduction
order (for reduction orders see, e.g.,  \cite{baader-nipkow}).
Question (2) is answered by proposing the ordered rewriting
relation\footnote{Ordered rewriting can be defined even
more broadly than as presented here. However, the
cases we discuss apply very broadly ---and modulo axioms $B$---
and are quite  easy to implement.} 
$\varphi \rightarrow_{\succ} \varphi'$ for this purpose.

 Assuming our induction
hypotheses are clauses, $\Gamma \rightarrow \Delta$,
with $\Gamma$ a conjunction 
and $\Delta$ a disjunction of equations,
it is worth considering a (by no means exhaustive)
 natural taxonomy of clauses associated to a reduction order $\succ$ on terms:
\begin{enumerate}
\item \textbf{Reductive equation}, $w_{0}=w_{1}$ such that, for some $i\in\{0,1\}$,
$\mathit{vars}(w_{i}) \supseteq \mathit{vars}(w_{i+_{2}1})$ and  
$w_{i} \succ w_{i+_{2}1}$, in which case,
the equation can  be oriented
as a  rewrite rule $w_{i} \rightarrow w_{i+_{2}1}$, where
$+_{2}$ denotes addition modulo 2.
We then call $w_{i}$ the equation's \emph{lefthand side}.

\item \textbf{Reductive conditional equation},  $u_{1}=v_{1},\ldots, u_{n}=v_{n}
\rightarrow w_{0}=w_{1}$ such that: (i)  $w_{0}=w_{1}$ is reductive
with lefthand side $w_{i}$,
(ii) $\mathit{vars}(w_{i}) \supseteq 
\mathit{vars}(u_{1}) \cup \mathit{vars}(v_{1}) 
\cup \ldots \cup 
\mathit{vars}(u_{n}) \cup \mathit{vars}(v_{n})$, and 
(iii) $w_{i} \succ u_{1},\; v_{1},\ldots, u_{n},\;v_{n}$.
A reductive conditional equation can be
oriented as a conditional rewrite rule
$w_{i} \rightarrow w_{i+_{2}1} \;\; \mathit{if}\;\;
 u_{1}= v_{1} \wedge \ldots,\wedge\; u_{n}=v_{n}$.

\item \textbf{Usable equation}, $w_{0}=w_{1}$,
a non-reductive
equation such that for some $i\in\{0,1\}$,
$\mathit{vars}(w_{i}) \supseteq \mathit{vars}(w_{i+_{2}1})$.
The set of $w_{i}$ satisfying this property is
called the set of \emph{candidate lefthand sides},
denoted $\mathit{cand}(w_{0}=w_{1})$. 

\item \textbf{Usable conditional equation}, a non-reductive
$u_{1}=v_{1},\ldots, u_{n}=v_{n}
\rightarrow w_{0}=w_{1}$ 
such that either:  (i) $w_{0}=w_{1}$ is reductive
with lefthand side $w_{i}$, and
$\mathit{vars}(w_{i}) \supseteq 
\mathit{vars}(u_{1}) \cup \mathit{vars}(v_{1}) 
\cup \ldots \cup 
\mathit{vars}(u_{n}) \cup \mathit{vars}(v_{n})$;
or  (ii) $w_{0}=w_{1}$ is usable, and
if $w_{i}\in \mathit{cand}(w_{0}=w_{1})$, then
$\mathit{vars}(w_{i}) \supseteq 
\mathit{vars}(u_{1}) \cup \mathit{vars}(v_{1}) 
\cup \ldots \cup 
\mathit{vars}(u_{n}) \cup \mathit{vars}(v_{n})$.
The set of \emph{candidate lefthand sides}
for such a conditional equation is
$\{w_{i}\}$ in case (i), and
$\mathit{cand}(w_{0}=w_{1})$
in case (ii).

\item \textbf{Unusable equation}, conditional or not.  But this
does not mean that they could not be used \emph{in other ways}, even
for rewriting in a user-guided maner.

\item \textbf{Clauses}  $\Gamma \rightarrow \Delta$ with more
than one disjunct in $\Delta$.  
Not being equations, they cannot be oriented
as $\Sigma$-rewrite rules; but we shall see in 
Section \ref{IndH-REW} that, under mild conditions,
they can be oriented as $\Sigma^{=}$-rewrite rules
(recall Section \ref{EQ-PREDS}).
\end{enumerate}

\vspace{1ex}

\noindent What is attractive about this taxonomy is its simplicity and
the easy, syntactic
way one can, given a reduction order $\succ$ on terms,
determine where a given clause falls in the above pecking order.  
Since induction hypotheses
are generated dynamically as part of the inductive deduction process,
the above taxonomy provides an inexpensive way to determine if
and how any such hypothesis can be used for simplification.
Specifically, equations falling in categories (1) and (2) are
\emph{immediately usable by standard rewriting} (modulo
axioms $B$ if they are present in the given specification
and $\succ$ is $B$-\emph{compatible}, i.e., if it defines an order,
not just on terms $u \succ v$, but on $B$-equivalence classes
$[u] \succ [v]$).  Furthermore, rewriting modulo axioms $B$ with reductive equations
oriented as rules is \emph{guaranteed to terminate} thanks to $\succ$.
Specifically, unconditional (resp. conditional) rewriting with reductive rules
is terminating (resp. operationally terminating in the sense of 
\cite{Lucas-Marche-Meseguer-IPL}).  Termination is
 extremely useful for
any simplification process.
The next two categories, (3) and (4), are the
whole point and focus of ordered rewriting: how to use them for simplification
with the order $\succ$, an idea coming from both term rewriting and
 paramodulation-based theorem proving (see, e.g., 
\cite{DBLP:conf/cade/MartinN90,Bachmair91,DBLP:books/el/RV01/NieuwenhuisR01}).
Let us begin with category (3):

\begin{definition} (Unconditional Ordered Rewriting Modulo).
For $\Sigma$ a kind-complete and $B$-preregular ordered-sorted signature,
and $B$ a set of equational $\Sigma$-axioms, given a set
$G$ of usable $\Sigma$-equations
and a reduction order $\succ$ on $\Sigma$-terms compatible with $B$,
the \emph{ordered rewriting relation
modulo} $B$, denoted $t \rightarrow_{G,B\succ} t'$, or sometimes
$t \stackrel{\theta}{\rightarrow_{G,B\succ}} t'$,
holds between $\Sigma$-terms $t$ and $t'$ iff there exist:
(i) an equation $(w_{0}=w_{1}) \in G$, (ii) a position $p$ in $t$, (iii) a candidate
lefthand side $w_{i} \in \mathit{cand}(e)$, and (iv) a substitution $\theta$,
such that: (a) $t|_{p} =_{B} w_{i}  \theta$, (b) $w_{i} \theta \succ w_{i+_{2}1} 
\theta$, and (c)  $t' \equiv t[w_{i+_{2}1}\theta]_{p}$.

The relation $t \rightarrow^{\circledast}_{G,B\succ} t'$ holds
between $t$ and $t'$ iff either (i) $t =_{B} t'$, or
(ii) $t \rightarrow^{+}_{G,B\succ} t'$.  Note that, since $\succ$ is a
reduction order modulo $B$,
$\rightarrow_{G,B\succ}$ is always a
\emph{terminating} relation.
\end{definition}

\noindent Intuitively, if, say, $\{w_{0},w_{1}\}=\mathit{cand}(w_{0}=w_{1})$, we can think of
ordered rewriting with $w_{0}=w_{1}$ as equivalent to rewriting 
with two conditional rules:
 (i) $w_{0} \rightarrow w_{1}\;\; \mathit{if}\;\; w_{0} \succ w_{1}$,
and (ii) $w_{1} \rightarrow w_{0}\;\; \mathit{if}\;\; w_{1} \succ w_{0}$,
or just one such rule if $\mathit{cand}(w_{0}=w_{1})$
 is a singleton set.
It is a form of ``constrained  rewriting,'' where the
constraint is specified by means of the reduction order
$\succ$.  

\vspace{1ex}

\noindent For category (4) the notion of conditional ordered rewriting
is somewhat  more involved.
However, the effectiveness in  both categories (3) and (4)
can be significantly increased by making ordered rewriting
(conditional or not) \emph{background theory aware}.  As it turns out,
background theory awareness will make the definition of conditional ordered
rewriting simpler, since no extra requirements on termination of condition
evaluation will be needed.

\vspace{2ex}

\noindent {\bf Making Order Rewriting Background-Theory-Aware}.
Sometimes, ordered rewriting \emph{should} deliver the goods, but it
doesn't.  Let us consider a simple example.

\begin{example} \label{comm-ord-rew} Consider the theory
$\mathcal{N}$ giving the usual
definition of natural number addition in Peano notation:
$n+0=n,\; n+s(m)=s(n+m)$, and suppose we want to
prove the commutativity law $x+y=y+x$ by standard
induction, inducting on $x$.  Let us focus on the
proof for the induction step, i.e., on proving the subgoal
$s(\overline{x}) + y = y + s(\overline{x})$
with induction hypothesis $\overline{x} + y = y + \overline{x}$
(where $\overline{x}$ is a fresh constant),
which simplifies
to $s(\overline{x}) + y = s( y + \overline{x})$
because ordered rewriting with
 the induction hypothesis is ineffective at this
point.   We induct again, now on $y$.  Let us focus
on the base case, where we get the sub-subgoal
$s(\overline{x}) + 0 = s( 0 + \overline{x})$,
which simplifies to
$s(\overline{x}) = s( 0 + \overline{x})$,
and, using the theory $\mathcal{N}^{=}$
that adds an equationally-defined
equality predicate to $\mathcal{N}$ (see Section \ref{EQ-PREDS}),
it further simplifies\footnote{In the inductive inference system presented
in Section \ref{ind-inf-rules}, the first and second simplifications
can both be achieved
by the equality predicate simplification rule (\textbf{EPS}).}
to $\overline{x} = 0 + \overline{x}$.
We would now like to use ordered rewriting
with the induction hypothesis $\overline{x} + y = y + \overline{x}$
to further simplify this sub-subgoal with matching substitution
$\{y \mapsto 0\}$.  But, since when extending the
total RPO order based on the total order on symbols
$+ \succ s \succ 0$,  its natural extension
is based on $+ \succ s \succ \overline{x} \succ 0$,
we cannot apply order rewriting here,
since  $0 + \overline{x} \not\succ \overline{x} + 0$
in the standard  RPO order.

This is  frustrating, because the $0$ in the
righthand side's instance
 is a silly obstacle blocking the rule's application,
which could, for example, be removed  by first
instantiating the induction hypothesis to this case
as $\overline{x} + 0 = 0 + \overline{x}$ and then
simplifying it to $\overline{x} = 0 + \overline{x}$,
so that with this simplified instance of
the hypothesis ordered rewriting 
\emph{would actually succeed} in proving this sub-subgoal.

This raises the following question: is there a way
of making ordered rewriting \emph{backgroud theory aware},
so as to substantially increase its chances of success
in situations such as the one above?  In this case, the
``background theory'' is of course
$\mathcal{N}$ or, more generally, $\mathcal{N}^{=}$.
In general, it may also include other
induction hypotheses and possibly
some already-proved lemmas.  

The answer to this
question is an emphatic Yes!  Recall that, conceptually,
ordered rewriting with $\overline{x} + y = y + \overline{x}$
is rewriting with the conditional constrained rules:
\[
\overline{x} + y \rightarrow y + \overline{x}
\;\;\;\; \mathit{if} \;\;\;\; \overline{x} + y \succ y + \overline{x}
\;\;\;\;\; \;\;\;\; \; \;\;\;\;\; \; \;\;\;\;\; 
\mathit{and}
\;\;\;\;\; \;\;\;\;\; \;\;\;\;\; \; \;\;\;\;\; 
y + \overline{x} \rightarrow \overline{x} + y 
\;\;\;\; \mathit{if} \;\;\;\; y + \overline{x} \succ  \overline{x} + y.
\]
Of course, as given, these rules cannot exploit at all
the presence of a background theory.  But they can do so
in a more subtle version of ordered rewriting, as the rules:
\[
\overline{x} + y \rightarrow z
\;\;\;\; \mathit{if} \;\;\;\; z :\!=  y + \overline{x} \; \wedge\;
\overline{x} + y \succ z
\;\;\;\;\; \;\;\;\; \; \;\;\;\;\; \; \;\;\;\;\; 
\mathit{and}
\;\;\;\;\; \;\;\;\;\; \;\;\;\;\; \; \;\;\;\;\; 
y + \overline{x} \rightarrow z
\;\;\;\; \mathit{if} 
\;\;\;\; z :\!= \overline{x} + y \; \wedge \;
y + \overline{x} \succ  z
\]
where we are using Maude's ``matching condition'' notation,
so that, for example,  evaluating, say, the first rule's  condition
$z :\!=  y + \overline{x}$ for a 
lefthand side's matching substitution $\theta$ means:
(i) normalizing $(y + \overline{x})\theta$ to, say,
$u$, and (ii) then performing the comparison
$(\overline{x} + y)\theta \succ u$.  In the notation
of conditional term rewriting systems (see, e.g., \cite{ohlebusch-book})
this would be described as a ``deterministic 
3-CTRS'' conditional rewrite
rule of the form:
$\overline{x} + y \rightarrow z
\;\;\;\; \mathit{if} \;\;\;\; y + \overline{x} \rightarrow z \; \wedge\;
\overline{x} + y \succ z$.  However, this description is
still a bit too simplistic, since it fails to distinguish between
equational and non-equational rewrite rules (see Section
\ref{RWL-BACKGR}  and  Footnote \ref{eq-rules-foot}
below).

This new method of contextual rewriting is indeed
``background theory aware,'' for the simple reason
that the normalization of expressions
in a matching condition uses the 
equational\footnote{\label{eq-rules-foot}In a rewrite
theory $\mathcal{R}=(\Sigma,E \cup B,R)$ where the equational rules
$\vec{E}$ are convergent modulo $B$, a non-equational and possibly
conditional rule $l \rightarrow r\;\; \mathit{if} \; \; \mathit{cond}$
can have an \emph{equational condition} $\mathit{cond}$.  This means
that only the equational rules $\vec{E}$ will be applied (modulo $B$)
to evaluate such an equational condition.  In our example,
the matching condition $z :\!=  y + \overline{x}$ is
an equational condition, which will be evaluated
using only equational rules such as those in the background theory
$\mathcal{N}$.  This distinction between equational rules and
non-equational ones will be exploited in Sections \ref{SCBR}--\ref{ind-inf-rules}.       
It is useful not just to make ordered rewriting ``theory aware,'' but
for the much broader reason that, while the rules $\vec{E}$
are assumed ground confluent,  confluence is often
lost when induction hypotheses are added, so that there is
no such thing anymore as \emph{the} canonical form of a term.
Exploring all the different canonical forms of a term is what the
distinction between equational and non-equational rules makes possible.}
rewrite rules available, which include those in the
background equational  theory.  In fact, the above recalcitrant sub-subgoal
$\overline{x} = 0 + \overline{x}$ can be immediately
discharged with the rule
$y + \overline{x} \rightarrow z
\;\;\;\; \mathit{if} 
\;\;\;\; z :\!= \overline{x} + y \; \wedge \;
y + \overline{x} \succ  z$, since
$\overline{x} + 0$ normalizes to $\overline{x}$,
and $0+ \overline{x} \succ  \overline{x}$.
\end{example}

The moral to be drawn from this
example is that there is a general \emph{theory transformation}
from usable equations
(conditional or not, i.e., from equations in categories (3)--(4))
to conditional rewrite rules that
supports background theory awareness: given
a usable conditional equation of the form:
$u=v \;\;\;\; \mathit{if} 
\;\;\;\; \mathit{cond}$ (in category (3), $\mathit{cond}= \top$),
assuming, say, that both $u$ and $v$ are
candidate lefthand sides,
 we transform it into the
constrained conditional rewrite rules:
\[
u \rightarrow z
\;\;\;\; \mathit{if} \;\;\;\; z :\!=  v  \; \wedge\;
u \succ z \; \wedge\; \mathit{cond}
\;\;\;\;\; \;\;\;\; \; \;\;\;\;\; \; \;\;\;\;\; 
\mathit{and}
\;\;\;\;\; \;\;\;\;\; \;\;\;\;\; \; \;\;\;\;\; 
v \rightarrow z
\;\;\;\; \mathit{if} 
\;\;\;\; z :\!= u \; \wedge \;
v \succ  z \; \wedge\; \mathit{cond}
\]
with $z$ a fresh variable whose sort is the top sort in the
connected component of sorts including those
of $u$ and $v$. 
Of course, if there is only one candidate lefthand side,
only one such rule is generated.  And if
$u \succ v$ (resp. $v \succ u$),
we only need to generate the simpler rule:
$u \rightarrow v \;\;\;\; \mathit{if} 
\;\;\;\; \mathit{cond}$
(resp. $v \rightarrow u \;\;\;\; \mathit{if} 
\;\;\;\; \mathit{cond}$).
To the best of my knowledge this form of
ordered rewriting  appears to be new.  
Background theory awareness
can  make ordered rewriting
considerably  more likely to succeed.  Success can be even more dramatic
in the case of conditional usable equations (category (4)), since, 
assuming the background equational theory is terminating,
the evaluations of conditions \emph{always} terminates.
That is, \emph{all} usable equations in categories (3)--(4)
become in this way executable in a terminating way.

But we can do even better: (i) we can associate 
to a conditional equation
 $u_{1}=v_{1},\ldots, u_{n}=v_{n}
\rightarrow w_{0}=w_{1}$ that is  reductive
with lefthand side $w_{i}$ (category (2))
the rule: 
\[
w_{i} \rightarrow w_{i+_{2}1} \;\; \mathit{if}\;\;
\bigwedge_{1 \leq i \leq n} (u_{i} \equiv v_{i}) \rightarrow tt
\]
where we have added a fresh new sort $\mathit{Pred}$ with constant $\mathit{tt}$ 
to the signature $\Sigma$
as well as  a new operator $\_\equiv\_$ of
sort $\mathit{Pred}$, and
new rules $x \equiv x \rightarrow tt$
for each $x$ whose sort is the top of a connected component
of $\Sigma$; (ii) likewise,  for  a
 usable conditional equation
$w_{1} = w_{2}  \;\;\;\; \mathit{if} 
\;\;\;\; \bigwedge_{1 \leq i \leq n} u_{i} = v_{i}$
 in category (4) with $w_{i} \in \mathit{cand}(w_{0}=w_{1})$
such that $w_{i}  \succ u_{1},\; v_{1},\ldots, u_{n},\;v_{n}$,
we can modify  the above transformation from a usable conditional
equation into a conditional rewrite rule by
 associating to the choice of lefthand
side $w_{i}$ the rule:
\[
w_{i} \rightarrow z
\;\;\;\; \mathit{if} \;\;\;\; z :\!=  v  \; \wedge\;
u \succ z \; \wedge\; \bigwedge_{1 \leq i \leq n} (u_{i} \equiv v_{i}) \rightarrow tt
\]
The main difference in case (ii)
 is that, if the above requirement
$w_{i}  \succ u_{1},\; v_{1},\ldots, u_{n},\;v_{n}$ does not hold, in
a rule application
the instance of the condition $\bigwedge_{1 \leq i \leq n} u_{i} = v_{i}$
is only evaluated \emph{with equations}.  Instead, now, the
\emph{rewrite condition} $\bigwedge_{1 \leq i \leq n} (u_{i} \equiv v_{i})
\rightarrow tt$ is evaluated with \emph{both equations and rules}.
The same applies to case (i): the condition is now evaluated
with equations and rules.
Since in our inductive theorem proving applications 
both reductive conditional equations and
usable equations will be \emph{induction hypotheses} used for
goal simplification, what this means in practice
is that, when applying  a conditional equational
hypothesis to simplify a goal,
in cases (i)--(ii) above
 we can use (executable) induction hypotheses
themselves, and not just the ground convergent equations
in the original theory, to evaluate the hypothesis' condition,
which gives us a better chance of success.
This both ensures
termination of goal simplification with equational hypotheses,
 and, in cases (i)--(ii) above,
 increases the chances that such hypotheses will be applied, since,
the evaluation of their conditions will be more likely to succeed.

The main import of this entire discussion for the rest of this
paper is that, without further ado, in any future
mention and use of the words ``ordered rewriting'' we will
\emph{always} mean ordered rewriting in
the \emph{background theory aware} version  defined above.

\subsection{Existential Formula Quantification with Skolem Signatures}
\label{Skolem-subsection}

What is the model-theoretic meaning of a universally (resp. existentially)
quantified sentence $(\forall{x_{1},\ldots,x_{n}})\; \varphi$
(resp. $(\exists{x_{1},\ldots,x_{n}})\; \varphi$) being satisfied
on a $\Sigma$-algebra $A$?  It exactly means that the sentence
$\overline{\varphi}$,  where $\overline{\varphi} \equiv
\varphi\{x_{1} \mapsto  \overline{x}_{1},\ldots,\; x_{n}\mapsto
\overline{x}_{n}\}$, is satisfied in \emph{all}  (resp. \emph{some})
$\Sigma(\{\overline{x}_{1},\ldots,\overline{x}_{n}\})$-algebras
$B$ that have $A$ as their $\Sigma$-reduct, i.e., such that
$B|_{\Sigma}=A$.
The constants $\overline{x}_{1},\ldots,\overline{x}_{n}$
are called the \emph{Skolem constants} involved in this business,
in honor of Thoralf Skolem.  But why quantifying
(universally or existentially) only on (Skolemized) \emph{variables}?
After all, we know, thanks to Skolem, that we can also quantify 
on \emph{function symbols}; and that this is the essential trick
(``Skolemization'') allowing us  to
reduce the satisfaction of
\emph{all} first-order formulas (first put in prenex normal form)
to that of quantifier-free formulas 
(for detailed treatments of prenex normal form and Skolemization see,
e.g., \cite{harrison-book,Baaz-etal-formula-normal-forms}).
 All this is eminently relevant for our purposes
in this paper, since inductive theorem proving is just proving the
satisfaction of a formula $\varphi$ in the initial algebra
$T_{\mathcal{E}}$ of an equational theory
$\mathcal{E}$ with, say, signature $\Sigma$.  But which formulas?
There are very good practical reasons to stick to quantifier-free ones;
but we can introduce quantified formulas through the back door
of Skolemization.  

A simple example may help at this point.
Suppose that we have specified the Peano natural numbers 
in an unsorted equational theory $\mathcal{N}$ with
constructors $0$ and $s$, the obvious equations for $+$,
$*$, and natural division $\mathit{div}$, and a Boolean predicate
$\mathit{even}$, defined by the equations: $\mathit{even}(0) = \mathit{true}$,
$\mathit{even}(s(0)) = \mathit{false}$, and 
$\mathit{even}(s(s(n))) = \mathit{even}(n)$.  A simple property we might like 
to prove is the following:
\[\mathit{even}(n)= \mathit{true} \; \Rightarrow\;
  (\exists m)\;\; n = s(s(0)) * m
\]
which Skolemized becomes:
\[
 (\exists f) \;\; \mathit{even}(n)= \mathit{true} \; \Rightarrow\; n = s(s(0)) * f(n).
\]
In general, of course, there may be, not just one, but several Skolem
constants and functions in a formula like this.  But what does this
mean? In an order-sorted setting, it just means that we are
quantifying (existentially in this case) over an order-sorted \emph{signature}
---let us call it a \emph{Skolem signature}, and denote it by $\chi$---
with the exact same poset of sorts $(S,\leq)$ as our original signature $\Sigma$, so
that a quantifier-free Skolemized formula
 is just a quantifier-free
$\Sigma \cup \chi$-formula $\psi$, and the satisfaction of its
 universal (resp.
existential) quantification, $(\forall \chi)\; \psi$ (resp. $(\exists \chi)\;
\psi$) in a $\Sigma$-algebra $A$
 has the expected meaning:  $A \models (\forall \chi)\; \psi$
(resp. $A \models (\exists \chi)\; \psi$)
iff  it is the case that in \emph{all} (resp. in \emph{some}) 
of the  $\Sigma \cup \chi$-algebras
$B$ such that $B|_{\Sigma}=A$, $B \models  \psi$ holds.

\vspace{1ex}

But then, what does it mean to \emph{prove}, that a $\Sigma$-algebra
$A$ satisfies an existentially quantified Skolemized formula?
For example, that 
the initial algebra $T_{\mathcal{N}}$ of our example satisfies
$(\exists f) \;\; \mathit{even}(n)= \mathit{true} \; \Rightarrow\; n = s(s(0)) * f(n)$?
It exactly means to show 
 that there is an \emph{interpretation} of $f$ as
a unary \emph{function} in $T_{\mathcal{N}}$, let us denote it
$f_{T_{\mathcal{N}}}:  T_{\mathcal{N}} \rightarrow T_{\mathcal{N}}$,
satisfying the formula $\mathit{even}(n)= \mathit{true} \; \Rightarrow\; n = s(s(0)) * f(n)$
in $T_{\mathcal{N}}$.
But how can we \emph{specify}  such a function $f_{T_{\mathcal{N}}}$?  Well, the best possible
situation is when we are lucky enough that $f_{T_{\mathcal{N}}}$
is \emph{definable} as a $\Sigma$-functional expression, i.e.,
when we can define $f_{T_{\mathcal{N}}}$ by means of
 a definitional extension of our theory $\mathcal{N}$
of the form: $f(n)=t(n)$, for $t(n)$ a $\Sigma$-term.
In this example we are lucky, since we can define
$f(n)=\mathit{div}(n,s(s(0))$, so that we just have to prove the
inductive theorem:
\[ \mathit{even}(n)= \mathit{true} \; \Rightarrow\; n = s(s(0)) * \mathit{div}(n,s(s(0)).
\]
In general, of course, this may not be the case (for example, if we
had not yet specified the $\mathit{div}$ function in $\mathcal{N}$).
However, thanks to the Bergstra-Tucker Theorem \cite{bt80},
\emph{any} total computable function can be 
specified in a finitary theory extension protecting the given data type
by a finite number of confluent and terminating equations.
What is the upshot of all this in practice?  The following:
To prove that an initial algebra $T_{\mathcal{E}}$ with signature $\Sigma$
satisfies an existentially quantified Skolemized formula
of the form $(\exists \chi)\; \psi$, with $\psi$ quantifier-free we can:
\begin{enumerate}
\item First, find, if possible,
a \emph{theory interpretation} (called a \emph{view} in Maude
\cite{maude-book})
$I: \chi \rightarrow \mathcal{E}$
that is the identity on sorts, and maps
each $f: s_{1},\ldots s_{n} \rightarrow s$ (including consants) in $\chi$
to a $\Sigma$-term $t_{f}(x_{1}: s_{1},\ldots,\; x_{n}: s_{n})$
of sort $s$ in $\mathcal{E}$.
This is equivalent to adding to 
$\mathcal{E}$ the family of definitional extensions
$\{ f(\vec{x})=t_{f}(\vec{x}) \}_{f \in \chi}$, but $I$ makes
those explicit extensions unnecessary.

\item Then, prove that
  $T_{\mathcal{E}}\models I(\psi)$ by standard inductive theorem
  proving methods, where $I(\psi)$ is the homomorphic extension
of the theory interpretation $I$ to a formula $\psi$.
If the choice of $I$ was correct, i.e., if indeed
$T_{\mathcal{E}}\models I(\psi)$, we can then think of $I$ as
a \emph{witness}, giving us a \emph{constructive proof}
of  the existential formula $(\exists \chi)\; \psi$ for $T_{\mathcal{E}}$.

\item However, if some of the function symbols
in $\chi$ are not definable as $\Sigma$-terms,
then \emph{extend} $\mathcal{E}$
to a theory $\mathcal{E}'$  (with auxiliary function symbols, and perhaps with
auxiliary sorts) that
\emph{protects} $T_{\mathcal{E}}$, i.e., such that
$T_{\mathcal{E}'}|_{\Sigma} \cong T_{\mathcal{E}}$.
Then, define $I$ as a theory interpretation
$I: \chi \rightarrow \mathcal{E}'$ and follow steps (1)--(2) above.
\end{enumerate}
For any formula  $(\exists \chi)\; \psi$
such that the function symbols in $\chi$
can be witnessed by \emph{computable}
functions in the initial algebra of interest 
$T_{\mathcal{E}}$ (which seems reasonable for
many Computer Science applications and agrees with a constructive view
of existential quantification),
 steps (1)--(3) above
reduce ---thanks to the Bergstra-Tucker Theorem---
 the theoremhood problem 
 $T_{\mathcal{E}}  \models (\exists \chi)\; \psi$
to the theoremhood problem
$T_{\mathcal{E}}  \models I(\psi)$,
with $I(\psi)$  a quantifier-free 
$\Sigma$-formula amenable to standard inductive
theorem proving methods.

This of course does no
guarantee that we can always \emph{prove}
the formula $I(\psi)$, even if the witness $I$ was correctly 
guessed.  After all, our example theory
$\mathcal{N}$ contains natural number addition
and multiplication, so that, by G\"{o}del's Incompleteness
Theorem, there is no hope of having an inference system
capable of  proving all theorems in $T_{\mathcal{N}}$.

\section{Multiclause-Based Inductive Theories and Induction Hypotheses} \label{SCBR}

The syntax and semantics of the inductive logic, including the
notions of \emph{multiclause} and of \emph{inductive theory},
are first presented in Section \ref{SuperClIndTh}. Then,
the important topic of how to effectively use induction
hypotheses for conjecture simplification purposes is
treated in Section \ref{IndH-REW}.
%
% The inference system itself, consisting of 11 
% simplification rules and 7  inductive rules,
% is presented and illustrated with examples in Section
% \ref{ind-inf-rules}.  Finally, six more inductive proof examples
% are added in Section \ref{inf-ex} to help the reader gain a more
% complete understanding of the inference system and its possibilities.

\subsection{Multiclauses and Inductive Theories} \label{SuperClIndTh}

Since predicate symbols can always be transformed into function symbols
by adding a fresh new sort \textit{Pred}, we can reduce all of order-sorted
first-order logic to just reasoning about equational formulas whose only atoms
are equations.
Any quantifier-free formula $\phi$ can therefore be put in conjunctive normal form (CNF)
as a conjunction of equational clauses
$\phi \equiv \bigwedge_{i\in I} \Gamma_i \rightarrow \Delta_i$,
where $\Gamma_i$, denoted $u_1 = v_1,\cdots,u_n = v_n$, is a \emph{conjunction}
of equations $\bigwedge_{1\leq i \leq n} u_i = v_i$ and $\Delta_i$, denoted
$w_1 = w'_1,\cdots, w_m = w'_m$, is a \emph{disjunction} of equations
$\bigvee_{1\leq k\leq m} w_k = w'_k$.
%
%
% so that $\Gamma_i \rightarrow \Delta_i$
% represents the equational clause:
%
% \begin{center}
% $\bigvee_{1\leq i \leq n} u_i \neq v_i \vee \bigvee_{1\leq k\leq m} w_k = w'_k$.
%\end{center}
%
Higher effectiveness can be gained by applying
inductive inference rules not to a single clause, but to a conjunction of
related clauses sharing the same condition $\Gamma$.  Thus, we will
assume that all clauses $\{\Gamma \rightarrow \Delta_l\}_{l\in L}$
with the same condition $\Gamma$ in the CNF of $\phi$ have been
gathered together into a semantically equivalent formula of the form
$\Gamma \rightarrow \bigwedge_{l\in L} \Delta_l$,
which I will call a \emph{multiclause}.\footnote{Using the multiclause
representation is optional: a user of the inference system
may choose to stick to the less general clause representation.
However, I explain in Section \ref{inf-ex} that using
multiclauses can afford a substantial economy of thought and lead to
shorter proofs.}
I will use the
notation $\Lambda$ to abbreviate $\bigwedge_{l\in L} \Delta_l$.
Therefore, $\Lambda$ denotes a conjunction of disjunctions of
equations.  Multiclauses, of course, generalize 
clauses, which
generalize both disjunctions of equations and 
conditional equations, which, in turn, generalize equations.
Likewise, multiclauses generalize conjunctions of
disjunctions of equations, which generalize conjunctions of
equations, which generalize equations.
Thus, multiclauses give us a very general setting for inductive reasoning.

What is an \emph{inductive theory}?
In an order-sorted equational logic framework,
the simplest possible inductive theories we can consider are order-sorted conditional
equational theories $\mathcal{E} = (\Sigma, E\cup B)$, where
$E$ is a set of conditional
equations (i.e., Horn clauses) of the form
$u_1 = v_1,\; \cdots,\; u_n=v_n \rightarrow w=w'$,
 and $B$ is a set of equational axioms such as associativity and/or
commutativity and/or identity.  Inductive properties are then
properties satisfied in the \emph{initial algebra}
$T_\mathcal{E}$ associated to $\mathcal{E}$.  Note that this is
exactly the initial semantics of functional modules in the 
Maude language.  So, as a first approximation, a Maude user can think of
an inductive theory as an order-sorted functional module.  
The problem, however, is that
as we perform inductive reasoning, the inductive theory of
$\mathcal{E}$, let us denote it $[\mathcal{E}]$ to emphasize
its initial semantics,  needs to be extended
by: (i) extra constants, and; (ii) extra formulas such as: (a)
induction hypotheses, (b) lemmas, and (c) hypotheses associated to
modus ponens reasoning.  Thus, general inductive
theories will have the form
$[\overline{X},\mathcal{E},H]$,
where $\overline{X}$ is a fresh set of constants having sorts in $\Sigma$
and $H$ is a set of $\Sigma(\overline{X})$ clauses,\footnotemark  corresponding to
formulas of types (a)--(c) above.
\footnotetext{Even when, say, an induction hypothesis
 in H might originally be a
multiclause 
$\Gamma\rightarrow \bigwedge_{l\in L}\Delta_l$, for executability
reasons it will always be decomposed into its corresponding set of clauses
$\{\Gamma\rightarrow \Delta_l\}_{l\in L}$.}
The \emph{models} of an inductive theory $[\overline{X},\mathcal{E},H]$
are exactly the $\Sigma^\square(\overline{X})$-algebras $A$ such that
$A|_{\Sigma^\square} \cong T_{\mathcal{E}^\square}$ and $A \models H$,
where $\mathcal{E}^\square=(\Sigma^\square,E\cup B^{\square})$ and
$\Sigma^\square$ is the kind completion\footnote{The reader might
wonder about the reasons for extending
$\mathcal{E}$ to $\mathcal{E}^\square$.  There are at least two
good reasons.  First, as shown in \cite{DBLP:conf/fossacs/Meseguer16},
the congruence closure modulo axioms $B$ of a set of ground
$\Sigma$-equations is not in general a set of ground
$\Sigma$-rewrite rules, but only a set of ground $\Sigma^\square$-rewrite
rules.  Therefore, even if the original goal to be proved is a 
$\Sigma$-multiclause, some of its subgoals may be
$\Sigma^\square$-multiclauses ---obtained, for example,
by application of the \textbf{ICC} simplification rule (see Section
\ref{ind-inf-rules}).  Second, the extra generality of
the theory $\mathcal{E}^\square$ supports the verification
of sufficient completeness properties  as inductive theorems
 for specifications $\mathcal{E}$ outside the scope of 
tree-automata techniques 
such as \cite{hendrix-meseguer-ohsaki-ijcar06},
a topic developed elsewhere \cite{DBLP:conf/wrla/Meseguer22}.}
of $\Sigma$
defined in Section \ref{OS-CC}.
Note that, since $T_{\mathcal{E}^\square}|_\Sigma = T_\mathcal{E}$
\cite{DBLP:conf/fossacs/Meseguer16}, the key relation for reasoning is
$A|_\Sigma \cong T_\mathcal{E}$.
Such algebras $A$ have a very simple description: they are pairs
$(T_{\mathcal{E}^\square},a)$ where $a : \overline{X}\rightarrow T_\mathcal{E}$
is the assignment interpreting constants $\overline{X}$.
In Maude, such inductive theories $[\overline{X},\mathcal{E},H]$
can be defined as \emph{functional theories} which \emph{protect}
the functional module $[\mathcal{E}]$, which in our expanded notation
is identified with the inductive theory
$[\emptyset,\mathcal{E},\emptyset]$.
I.e., $[\emptyset,\mathcal{E},\emptyset]$ has no
extra constants, hypotheses, lemmas, or assumptions.  It is
the inductive theory from which we will start our reasoning.

I will furthermore assume that
$\vec{\mathcal{E}}=(\Sigma,B,\vec{E})$, with $B = B_{0} \uplus U$,
is ground convergent, 
with a total\footnote{\label{total-RPO} A crucial property of a total
RPO order modulo $B_{0}$, for $B_{0}$ any combination
of associativity and/or commutativity axioms,
is that it defines a \emph{total} order on $B_{0}$-equivalence
classes of ground terms 
\cite{DBLP:journals/iandc/Rubio02,DBLP:conf/csl/Rubio95}.  This property
has many useful consequences for the present work.}
 RPO order $\succ$ modulo $B_{0}$ making $\vec{\mathcal{E}}_{U}$
operationally terminating,
and that there is a ``sandwich''
$(\Omega,B_{\Omega},\emptyset) \subseteq (\Sigma_{1},B_{1},\vec{E}_{1})
\subseteq (\Sigma,B,\vec{E})$ with $B_{\Omega} \subseteq B_{0}$
satisfying all the requirements in Section \ref{var-sat-nut},
including the sufficient completeness of $\vec{\mathcal{E}}$
w.r.t. $\Omega$, the
 finite variant property of $\mathcal{E}_{1}$,
and the OS-compactness of $(\Omega,B_{\Omega})$.
Under the above assumptions,
up to isomorphism, and identifying $T_{\mathcal{E}^\square}$ 
with $\mathcal{C}_{\Sigma^{\square}/E,B}$,
a model $(A,a)$ of an inductive theory
$[\overline{X},\mathcal{E},H]$ has a very simple
description as a pair $(T_{\mathcal{E}^\square},[\overline{\alpha}])$,
where $\alpha: X \rightarrow T_{\Omega}$ is a ground constructor
substitution, $[\alpha]$ denotes the composition
$X \stackrel{\alpha}{\rightarrow} T_{\Omega}
\stackrel{[\_]}{\rightarrow}T_{\Omega/B_{\Omega}}$, with $[\_]$  the unique
$\Omega$-homomorphism mapping each term $t$ to its
$B_{\Omega}$-equivalence class $[t]$, and where
 $[\overline{\alpha}] : \overline{X} \rightarrow T_{\mathcal{E}^\square}$
maps each $\overline{x} \in \overline{X}$ to $[\alpha](x)$,
where $x \in X$ is the variable with same sort associated to  
$\overline{x} \in \overline{X}$.
The fact 
that for each clause $\Gamma\rightarrow \Delta$ in $H$
we must have $(T_{\mathcal{E}^\square},[\overline{\alpha}]) 
\models \Gamma\rightarrow \Delta$ has also a very simple
expression.  Let $Y = \mathit{vars}(\Gamma\rightarrow \Delta)$.
Then, $(T_{\mathcal{E}^\square},[\overline{\alpha}]) 
\models \Gamma\rightarrow \Delta$ exactly means that
$T_{\mathcal{E}^\square} \models  (\Gamma\rightarrow  \Delta)^{\circ}\alpha$, 
where $(\Gamma\rightarrow \Delta)^{\circ}$ is obtained from
$\Gamma\rightarrow \Delta$ by replacing each 
constant $\overline{x} \in \overline{X}$ appearing in it by its 
corresponding variable $x \in X$.  Equivalently, it also
means that  for each ground constructor substitution 
$\beta: Y \rightarrow T_{\Omega}$
we have $T_{\mathcal{E}^\square}  \models
(\Gamma\rightarrow  \Delta)^{\circ}(\alpha \uplus \beta)$.
This entire discussion can be summarized by the equivalence:
$(T_{\mathcal{E}^\square},[\overline{\alpha}])
\models \Gamma \rightarrow \Lambda$
iff 
$T_{\mathcal{E}^\square} \models (\Gamma \rightarrow \Lambda)^{\circ}\,\alpha$.
In particular, $(T_{\mathcal{E}^\square},[\overline{\alpha}])$
 is a model
of the  inductive theory 
$[\overline{X},\mathcal{E},H]$ iff 
$T_{\mathcal{E}^\square} \models H^{\circ}\,\alpha$.

We will finally assume that the clauses $H$ in an inductive theory
$[\overline{X},\mathcal{E},H]$ have, first of all, been
simplified with the theory  $\vec{\mathcal{E}}_{\overline{X}_U}^=$
that adds equationally-defined equality predicates to
$\vec{\mathcal{E}}_{\overline{X}_U}$ (see Section \ref{EQ-PREDS}),
and then
classified, according
to a  slightly refined version of 
the taxonomy described in Section \ref{OR-section},
into the disjoint union
$H = H_e \uplus H_{\vee,e} \uplus H_{ne}$,  
where: (i)  $H_{e}$, the \emph{executable} hypotheses,
are (possibly conditional) equations that are either reductive (classes (1)--(2)), or
usable unconditional or conditional equations (classes (3)--(4))
executable by ordered rewriting 
using the same RPO order $\succ$ modulo $B_{0}$
that makes $\vec{\mathcal{E}}=(\Sigma,B,\vec{E})$ ground convergent;
(ii) $H^{=}_{\vee,e}$, which we call the
$(=)$-\emph{executable non-Horn hypotheses}
(because they can be executed
as conditional $\Sigma^{=}$-rewrite rules, as explained below),
are clauses of the form $\Gamma \rightarrow \Delta$
where $\Delta$ has more than one disjunct
(therefore in class (6)) and such that 
$\mathit{vars}(\Delta) \supseteq \mathit{vars}(\Gamma)$;
finally, $H_{ne}$,  the \emph{non-executable} hypotheses,
are all other clauses in $H$ (necessarily in classes (5) or (6)).
Although clauses in $H_{ne}$ cannot be
used as rewrite rules to simplify conjectures,
they can still be used for goal simplification
in two other ways, namely, by means of
the clause subsumption simplification
rule ({\bf CS}) and also  (if in category (3))
in a user-guided manner
by the equality rule  ({\bf Eq}), both
described in Section \ref{ind-inf-rules}.

\subsection{Using Induction Hypotheses as Rewrite Rules} 
\label{IndH-REW}

The effectiveness of an inductive inference system can be significantly
increased by maximizing the ways in which
induction hypotheses can be applied as rewrite rules
 to simplify inductive  goals.\footnote{That is, formulas
appearing in a conjectured inductive theorem or in subgoals of
such a conjecture. Note that  rewriting
modulo axioms with both the module's ground
convergent oriented equations $\vec{E}$ and executable hypotheses in $H$
can simplify both \emph{terms} in a goal and the goal $\Gamma \rightarrow \Delta$
itself.  However, as further explained below, 
since we regard goal $\Gamma \rightarrow \Delta$ as a \emph{Boolean term}
in the convergent equational theory extending the 
goal's theory $\mathcal{E}$ with equationally-defined equality
predicates, formula simplification can be
even more powerful than term simplification, since 
we can also use the oriented equations defining equality predicates
to simplify not only terms but also \emph{equations} inside a goal.}
In this regard, even for the most obvious canditates,
namely, the equations $H_{e}$ in the above classification,
which are easily orientable as rewrite rules
$\vec{H_{e}}$ because they have either lefthand sides or
candidate lefthand sides (when having two candidate lefthand sides, 
they are orientable as two
$\succ$-\emph{constrained} conditional rewrite rules: see Section 
\ref{OR-section}),  since such rules
may contain operators enjoying unit axioms  $U$
in $B=B_{0} \cup U$, we should not use $\vec{H_{e}}$ itself,
but its $U$-transformation $\vec{H}_{e_{U}}$, which achieves modulo
$B_{0}$ the same semantic effect as the rules $\vec{H_{e}}$ modulo
$B=B_{0} \cup U$. A crucial reason for performing the 
 $U$-transformation on such hypothesis rules is that we want to use
the induction hypotheses not just for \emph{term} simplification,
but, more generally, for \emph{formula} simplification.  But since:
(i) the rules in 
$\vec{\mathcal{E}}^{=}_{\overline{X}_U}$ can simplify exactly the same
$\Sigma$-\emph{terms} as those in $\vec{\mathcal{E}}_{\overline{X}}$
(and to the same normal forms), but (ii) they  can \emph{also} simplify many
QF $\Sigma$-\emph{formulas}, we will always want to
use the rules $\vec{H}_{e_{U}}$ \emph{together} with the
rules $\vec{\mathcal{E}}^{=}_{\overline{X}_U}$ instead of
using the combination $\vec{\mathcal{E}}_{\overline{X}} \cup\vec{H_{e}}$, which is strictly weaker for \emph{formula}
simplification.  How to most effectively do this, so as to
maximize our chances to prove conjectures while avoiding
non-termination problems is the matter I turn to
next.

Note that formula simplification will always take place
in the context of a goal associated to an inductive theory
$[\overline{X},\mathcal{E},H]$.  We can spell out how formula
simplification will be achieved by describing such formula
simplifications as the computations of an associated rewrite
theory in the sense of Section \ref{RWL-BACKGR}, namely, the
rewrite theory
\[\mathcal{R}_{[\overline{X},\mathcal{E},H]}=(\Sigma^{=}(\overline{X}),
E^{=} \cup B_{0}^{=}, \vec{H}^{+}_{e_U})
\]
where ---as already mentioned, and it will be further explained in what follows---
the induction hypotheses $H$ have been suitably \emph{simplified}
beforehand to increase their effectiveness.
The rewrite rules $\vec{H}^{+}_{e_U} =  \vec{H}_{e_U}  \cup
\vec{H}^{=}_{\vee,e_{U}}$ are defined as follows:
\begin{enumerate}
\item $\vec{H}_{e_U}$ are the $U$-transformed rewrite rules
corresponding to induction hypotheses in $H_{e}$.  These
are hypotheses orientable as either (possibly conditional)
reductive rewrite rules, or as usable conditional equations, both
oriented as \emph{background aware}.
In this case, the background equational theory is $(\Sigma^{=}(\overline{X}),
E^{=} \cup B_{0}^{=})$) and the usable equations are transformed
$\succ$-constrained rewrite rules in the manner
explained in Section \ref{OR-section}.  In both cases, if the rules
are conditional ---besides the possible variable matching condition 
$x := r$ and constraint condition $l \succ x$---
their remaining conditional part will be either a conjunction of equations
$\bigwedge_{i\in I} u_{i}=v_{i}$, or, for reductive or usable
conditional equations in the cases (i)--(ii) discussed towards the
end of Section \ref{OR-section}, their remaining conditional part will
be of the form\footnote{Note that, since 
$(\Sigma^{=}(\overline{X}), E^{=} \cup B_{0}^{=})$ already has an equality predicate $\_=\_$
we do not need to add a new sort $\mathit{Pred}$ or an
equality predicate $\_\equiv\_$ as done in Section \ref{OR-section}
to define such conditions:
we just use the sort $\mathit{NewBool}$ and the equality  predicate $\_=\_$
in  $(\Sigma^{=}(\overline{X}), E^{=} \cup B_{0}^{=})$ instead.}
$\bigwedge_{i\in I} (u_{i}=v_{i}) \rightarrow \top$.
In a rewrite theory $\mathcal{R}=(\Sigma,E \cup B,R)$,
a rewrite condition $u \rightarrow v$
is evaluated for a given matching substitution $\theta$
 by searching for a term $w$ such that: (i) either 
$u\theta =_{B} w$ or
$u\theta
 \rightarrow^{+}_{R\!:\!E,B} w$,
and (ii) there exists a substitution $\gamma$ such that
$w =_{B} (v \gamma)$.  In our case, this means that
each rewrite condition $(u_{i}=v_{i}) \rightarrow \top$ will
be evaluated by applying both the equational rules in 
$\vec{E}^{=}$, which include the rule $x = x \rightarrow \top$,
and the rules in $\vec{H}^{+}_{e_U}$.  This maximizes the chances of
success when  applying the rules in $\vec{H}_{e_U}$ to simplify a goal.

%
% \item $\vec{H}_{wu_U}$ are, by definition, the $U$-transformed
% versions of the hypotheses that are orientable
% as weakly usable conditional rewrite rules.  Their treatment
% as rewrite rules is entirely similar to that for
% the strongly usable rules in  $\vec{H}_{e_U}$, \emph{except}
% for one crucial difference: their equational
% condition $\bigwedge_{i\in I} u_{i}=v_{i}$ \emph{remains equational},
% and is therefore evaluated only  by means of equational rules.
% This is because, as pointed out in Section \ref{OR-section},
% the operational termination of weakly usable conditional
% rules cannot be guaranteed when such rules are themselves
% used recursively  to evaluate their condition.  However, termination
%of an equational condition instance $(\bigwedge_{i\in I} u_{i}=v_{i})\theta$
% \emph{is} guaranteed to terminate when it is only evaluated by means
% of ground \emph{convergent} equational rules.

\item $\vec{H}^{=}_{\vee,e_{U}}$  are rules obtained from 
  $U$-transformed  hypotheses
$\Gamma \rightarrow \Delta$
such that $\Delta$ contains two or more disjuncts and
$\mathit{vars}(\Gamma) \subseteq \mathit{vars}(\Delta)$
as follows: (i) if $\Delta > u_{i},v_{i}$ for each $u=v$ in $\Gamma$,
where $\succ$ denotes the extension of the RPO order on $\Sigma$
 to symbols in $\Sigma(X)^{=}$, then $\Gamma \rightarrow \Delta$
produces a rule of the form,
\[
\Delta \rightarrow \top \;\; \mathit{if} \;\; 
\bigwedge_{(u = v) \in \Gamma} (u=v) \rightarrow \top
\]
(ii) otherwise, $\Gamma \rightarrow \Delta$ produces a rule of the form,
\[
\Delta \rightarrow \top \;\; \mathit{if} \;\; 
\bigwedge_{(u = v) \in \Gamma} u=v.
\]
\end{enumerate}
A direct consequence of such a definition of 
$\mathcal{R}_{[\overline{X},\mathcal{E},H]}$ is that this theory
is \emph{operationally terminating}: we are guaranteed to never
loop when simplifying a goal's formula with its
associated rewrite theory
$\mathcal{R}_{[\overline{X},\mathcal{E},H]}$.
Operational termination of
$\mathcal{R}_{[\overline{X},\mathcal{E},H]}$
means that each QF formula $\varphi$ 
will have a finite \emph{set} of normal forms,
which by abuse of notation
will be denoted $\varphi!_{\vec{H}^{+}_{e_{U}}\!:\! \vec{E}^{=}_{U},B_{0}}$. 
In Maude, such a set can be computed using the search command:
\[
\mathit{search}\;\;\; \varphi \;\;\;
\Rightarrow!\;\;\;\; B\!:\! \mathit{NewBool}\;\; .
\]

Two additional notational
conventions will also be useful in what follows.  We can identify within the set
$ \vec{H}_{e_U}$ the subset $\vec{H}_{ge_U} \subseteq \vec{H}_{e_U}$
of its ground unconditional  $\Sigma(\overline{X})$-rewrite 
rules.\footnote{Note that \emph{all} unconditional ground
equations in $H$ are always orientable as
($U$-transformed) reductive ground rewrite rules 
in $\vec{H}_{ge_U}$.
This is because, as explained in Footnote
\ref{total-RPO}, $B_{0}$-equivalence
classes of $\Sigma$-terms are totally ordered
in a total RPO order modulo $B_{0}$, and we assume
that the total order on symbols in $\Sigma$  has been extended to
a total order on symbols in $\Sigma(\overline{X})$.}
The second notational convention is related to the previous one.  By
definition, $\vec{H}^{\oplus}_{e_U}=\vec{H}^{+}_{e_U}\setminus
\vec{H}_{ge_U}$.
In a similar way, if 
our hypotheses $H$ before the $U$-transformation were classified
as $H = H_e  \uplus H_{\vee,e} \uplus H_{ne}$, 
and $H_{ge} \subseteq H_{e}$ denotes the ground equations,
we define $H^{+}_{e} =  H_e \uplus H_{\vee,e}$,
and $H^{\oplus}_{e} = H^{+}_{e} \setminus H_{ge}$.

\vspace{2ex}

\noindent {\bf Simplifying Hypotheses}.  The above discussion has so
far focused on how to effectively use induction hypotheses to
simplify conjectures by rewriting, but \emph{who simplifies the simplifiers?}
There are, for example, two ways
in which induction hypotheses may be ineffective as rewrite rules.
To begin with, if the lefthand side of a hypothesis rule is not
normalized by the rules in 
$\vec{\mathcal{E}}$, it may fail even to
apply to a conjecture that has itself been simplified.
Furthermore, induction hypotheses
may \emph{preempt each other}, in the sense that
the application of one hypothesis might block the application of
another, also useful, hypothesis. For a simple example,
consider a ground hypothesis rule $\overline{x} \rightarrow \overline{y}$
and a non-ground hypothesis rule $z + \overline{x}  \rightarrow s(s(z))$.
Then, application of the ground hypothesis to an expression
$u + \overline{x}$ will preempt application of the non-ground hypothesis.

We can try to minimize
these problems by: (i)  ensuring that  lefthand sides of
hypotheses used as rewrite rules (or even as
non-executable hypotheses) are
$\vec{E},B$-normalized,
and (ii) trying to minimize in a practical
way the problem of executable  hypotheses preempting each other.
These aims can be advanced by a hypothesis
transformation process $H \mapsto H _{\mathit{simp}}$
that tries to arrive at simplified hypotheses $H _{\mathit{simp}}$
meeting goals (i)--(ii) as much as possible in a practical manner. 

There is, of course, no single such transformation possible.
What follows is a first proposal for such
a transformation, open to further modifications
and  improvements through future experimentation.
The transformation $H \mapsto H _{\mathit{simp}}$ is defined
as follows:

\begin{enumerate}
\item We first define $H' = \mathit{clauses}(H!_{\vec{\mathcal{E}}^{=}_{\overline{X}_{U}}})$
where $\mathit{clauses}$ turns a formula  into a set of clauses.
For example, if $[\_,\_]$ is a constructor for unordered pairs,
$\vec{E}^{=}$ will contain a rule $[x,y]=[x',y']
\rightarrow (x = x') \wedge (y = y')$.  Then, assuming that
the equation $[u,v] = [u'=v']$ belongs to $H$,
and that the $u,u',v,v'$ are already in $\vec{E},B$-normal form,
we will have $([u,v] =[u'=v'])!_{\vec{\mathcal{E}}^{=}_{\overline{X}_{U}}} =
\; (u=u' \wedge v = v')$, so that $\{u=u', v = v'\} \subseteq  H'$.
Note that the two equations $u=u'$  and $v = v'$ will be more
widely applicable as induction hypothesis than the
original equation $[u,v] = [u'=v']$.

\item As usual, we can classify $H'$ as the disjoint union
$H' = H'_{ge} \uplus H_{e}'^{\oplus} \uplus H'_{ne}$.  We then
define $H _{\mathit{simp}} = H''$, where,

\item $H''_{ge} = \mathit{orient_{\succ}}( (cc^{\succ}_{B_{0}}(H'_{ge}))!_{\vec{E},B})$.
That is, we first compute the congruence closure modulo $B_{0}$
of the ground equations $H'_{ge}$ (or an approximation to
it if $B_{0}$ contains axioms for associative but non-commutative
operators) to make them convergent; but to exclude the possibility 
that some resulting rules might not
be $\vec{E},B$-normalized, we further $\vec{E},B$-normalize
both sides of those ground equations and orient them
again as rules with the RPO order $\succ$.

\item ${H''_{e}}^{\oplus}$ is generated
as follows: (i) For each $u = v \;\; \mathit{if} \;\; \mathit{cond}$
in $ H_{e}'^{\oplus}$ (where $\mathit{cond}$ could be empty)
we add to ${H''_{e}}^{\oplus}$ the set
$(u = v)!_{\vec{H}''_{ge}\!:\! \vec{E},B} \;\; \mathit{if} \;\; \mathit{cond}$.
Please, recall the notation $\rightarrow_{R \!:\! \vec{E},B}$
introduced in Section \ref{RWL-BACKGR} for rewriting
with the rules $R$ and equations $E$ of a rewrite theory
modulo axioms $B$.  Since confluence
is not guaranteed, 
$(u = v)!_{\vec{H}''_{ge}\!:\! \vec{E},B} \;\; \mathit{if} \;\;\mathit{cond}$
may in general be a set of hypotheses. (ii)
 For each $\Gamma \rightarrow \Delta$ in $ H'^{\oplus}$
we add to ${H''_{e}}^{\oplus}$ the set
$\Gamma \rightarrow (\Delta)!_{\vec{H}''_{ge}\!:\! \vec{E},B}$.

\item $H''_{ne}$ is generated by adding for each clause
(which could be a conditional equation)
$\Gamma \rightarrow \Delta$ in $H_{ne}'$
the set $(\Gamma \rightarrow \Delta)!_{\vec{H}''_{ge}\!:\! \vec{E},B}$,
 to $H''_{ne}$.
\end{enumerate}
The simplified hypotheses $H''$ thus obtained, give rise
to corresponding rewrite rules $\vec{H}''_{ge}
\uplus {\vec{H''}}_{eU}^{\oplus}$ and to non-executable hypotheses
$H''_{ne}$ as usual.  So, what has been achieved
by the $H \mapsto H _{\mathit{simp}}$ transformation thus defined?
The main achievements are: (1) the ground rewrite rules 
$\vec{H}''_{ge}$ are $\vec{E},B$-normalized
and, with some luck, these rules may even be convergent; (2)
the rewrite rules in 
$\vec{H''}_{eU}^{\oplus}$
are both $\vec{E},B$-normalized and 
$\vec{H}''_{ge},B$-normalized; and
(3) the consequents of  clauses in $H''_{ne}$ are likewise
both $\vec{E},B$-normalized and 
$\vec{H}''_{ge},B$-normalized.  Although the clauses
in $H''_{ne}$ are not executable as rewrite rules,
as we shall see, they can nevertheless simplify goals
by means of
either the clause subsumption
$\textbf{CS}$ or the $\textbf{EQ}$ rule,
both discussed in Section \ref{ind-inf-rules}.
Since the $\textbf{CS}$ and $\textbf{EQ}$ rules
will be applied by matching
a goal (resp. a term in the goal)
 (that can both be safely assumed to already be 
both $\vec{E},B$-normalized and 
$\vec{H}''_{ge},B$-normalized)
to a pattern hypothesis in $H''_{ne}$
(resp. to a chosen side of the equation),
the chances for such a matching to succeed are maximized
by the above definition of $H''_{ne}$.

\section{Inductive Inference System} \label{ind-inf-rules}

The inductive inference system presented below transforms
\emph{inductive goals} of the form
$[\overline{X},\mathcal{E},H] \Vdash \Gamma \rightarrow \Lambda$,
where $[\overline{X},\mathcal{E},H]$ is an inductive theory
and $\Gamma \rightarrow \Lambda$ is
a $\Sigma^\square(\overline{X})$-multiclause,
 into sets
of goals.  The empty set of goals is denoted $\top$ to suggests that the
goal from which it was generated has been proved (is a \emph{closed}
goal).  However, in the special case of
goals of the form  
$[\emptyset,\mathcal{E},\emptyset]\Vdash \Gamma \rightarrow \Lambda$,
called \emph{initial goals},
we furthermore require that $\Gamma \rightarrow \Lambda$
is a $\Sigma$-multiclause; and
we also allow \emph{existential initial goals}
of the form 
$[\emptyset,\mathcal{E},\emptyset]\Vdash  (\exists \chi)(\Gamma\rightarrow\Lambda)$,
with $\chi$ a Skolem signature and $\Gamma \rightarrow \Lambda$ a
$\Sigma \cup \chi$-multiclause
(see Section \ref{Skolem-subsection}).

A \emph{proof tree} is a tree of goals, where at the root we have the original
goal that we want to prove and the children of each node in the tree have been
obtained by applying an inference rule in the usual bottom-up proof search fashion.
Goals in the leaves are called the \emph{pending goals}.  A proof tree
is \emph{closed} if it has no pending goals, i.e., if all its
leaves are marked $\top$.
Soundness of the inference system means that if the goal
$[\overline{X},\mathcal{E},H] \Vdash \phi$
is the root of a closed proof tree, then $\phi$ is valid in the inductive theory
$[\overline{X},\mathcal{E},H]$,
i.e., it is satisfied by all the models $(T_{\mathcal{E}^\square},[\overline{\alpha}])$ of
$[\overline{X},\mathcal{E},H]$ in the sense explained above.

The inductive inference system presented below consists of two sets of
inference rules:
(1) \emph{goal simplification rules}, which are easily amenable to automation, and
(2) \emph{inductive rules}, which are usually applied under user guidance,
although they could also be automated by tactics.

To increase its effectiveness,
this inference system maintains the invariant that the induction
hypotheses $H$ in all inductive theories will always be in simplified form.

\subsection{Goal Simplification Rules} \ \\

\noindent\textbf{Equality Predicate Simplification} (\textbf{EPS}).

\begin{prooftree}
\AxiomC{$\{[\overline{X},\mathcal{E},H] \Vdash \Gamma_{i}'
\rightarrow\Lambda_{i}'\}_{i \in I}$}
\UnaryInfC{$[\overline{X},\mathcal{E},H] \Vdash \Gamma\rightarrow\Lambda$}
\end{prooftree}

\noindent where $(\bigwedge_{i\in I} \Gamma_{i}'
\rightarrow\Lambda'_{i})\in
(\Gamma\rightarrow\Lambda)
!\,_{\vec{\mathcal{E}}_{\overline{X}_U}^=\cup\,\vec{H}^{+}_{e_U}}$.
Furthermore, if $\top \in
(\Gamma\rightarrow\Lambda)
!\,_{\vec{\mathcal{E}}_{\overline{X}_U}^=\cup\,\vec{H}^{+}_{e_U}}$,
then $\bigwedge_{i\in I} \Gamma_{i}'
\rightarrow\Lambda'_{i}$ is chosen to be $\top$;
if $\bot \in
(\Gamma\rightarrow\Lambda)
!\,_{\vec{\mathcal{E}}_{\overline{X}_U}^=\cup\,\vec{H}^{+}_{e_U}}$,
then $\bigwedge_{i\in I} \Gamma_{i}'
\rightarrow\Lambda'_{i}$ is chosen to be $\bot$
(i.e., the conjecture $\Gamma\rightarrow\Lambda$
is then shown to be \emph{false});  
otherwise, the choice of $\bigwedge_{i\in I} \Gamma_{i}'
\rightarrow\Lambda'_{i}$ is unspecified,
but it should be optimized according to some criteria that
make the resulting subgoals easier to prove.\footnote{For example,
that the number of the goals in the conjunction is smallest possible.}
The fact that an \textbf{EPS} simplification might result in
a \emph{conjunction} of goals is illustrated in 
Example \ref{EPS-example} below.
We assume that
the hypotheses $H$ have already been simplified using the
$H \mapsto H _{\mathit{simp}}$ transformation, since this
should be an invariant maintained throughout.
We also assume that constructors in the transformed
theory $\mathcal{E}_{U}$ are free
modulo axioms.\footnote{\label{non-free-ctors} This requirement will be relaxed
  in the future to allow equational programs whose constructors
  are not free modulo axioms.  The main reason for currently imposing this restriction
  is precisely to allow application of the \textbf{EPS} simplification
  rule, which simplifies formulas with the rules in
  $\vec{\mathcal{E}}_{U}^{=}$ (plus executable induction hypotheses), since 
  $\vec{\mathcal{E}}_{U}^{=}$ assumes that $\mathcal{E}_{U}$
has free constructors modulo $B_{0}$.}
This means that in 
$\mathcal{E}_\Omega=(\Omega,B_\Omega)$, the axioms $B_\Omega$
decompose as $B_\Omega=Q_\Omega\uplus U_\Omega$ with $U_\Omega$ the unit axioms
and $Q_\Omega$ the associative and/or commutative axioms such that
$T_{\mathcal{E}_\Omega}\cong T_{\Omega/Q_\Omega}$.
This can be arranged with relative ease in many cases by subsort overloading,
so that the rules in $\vec{U}_\Omega$ only apply to subsort-overloaded operators
that are \emph{not} constructors.

For example, consider sorts $\textit{Elt} < \textit{NeList} < \textit{List}$
and $\Omega$ with operators  $\mathit{nil}$ of sort \textit{List} and
$\_\,;\_\; : \textit{NeList}\ \textit{NeList} \rightarrow
\textit{NeList}$
and $B_{\Omega}$
associativity of $\_\,;\_$ with identity \textit{nil}, but where
$\_\,;\_\; : \textit{List}\ \textit{List} \rightarrow \textit{List}$,
declared with the same axioms, is in $\Sigma \setminus \Omega$.
Then $Q_\Omega$ is just the associativity axiom for $(\_\,;\_)$.

In summary, this inference rule simplifies a 
multiclause $\Gamma\rightarrow\Lambda$ with:
(i) the rules in $\vec{\mathcal{E}}_{U}$,
(ii) the equality predicate rewrite rules in 
$\vec{\mathcal{E}}_{U}^=$,
and
(iii) the hypothesis rewrite rules 
$\vec{H}^{+}_{e_U}$.\footnote{Recall 
that  $\Gamma$ is a conjunction 
and $\Lambda$ a conjunction of disjunctions. 
Therefore, the equality predicate rewrite rules 
together with $\vec{H}^{+}_{e_U}$
may have powerful ``cascade effects.''  For example, 
if either $\bot \in \Gamma 
!\,_{\vec{\mathcal{E}}_{\overline{X}_U}^=\cup\,\vec{H}^{+}_{e_U}}$
or
$\top \in \Lambda
!\,_{\vec{\mathcal{E}}_{\overline{X}_U}^=\cup\,\vec{H}^{+}_{e_U}}$,
then
$\top \in (\Gamma\rightarrow\Lambda)
!\,_{\vec{\mathcal{E}}_{\overline{X}_U}^=\cup\,\vec{H}^{+}_{e_U}}$ 
 and the goal has been proved.}

\begin{example} \label{EPS-example} Let $\mathcal{N}\mathcal{P}$ denote
  theory of the natural numbers in Peano notation with the usual
  equations defining addition, $+$,
  multiplication, $*$, and exponentiation, $(\_)^{(\_)}$, and with two
  additional sorts, $\mathit{Pair}$ and $\mathit{UPair}$, of,
  respectively, ordered and unordered pairs of numbers, built with the
constructor operators, 
$[\_,\_]: \mathit{Nat}\;\; \mathit{Nat} \rightarrow \mathit{Pair}$ and
$\{\_,\_\}: \mathit{Nat}\;\; \mathit{Nat} \rightarrow \mathit{UPair}$,
where $\{\_,\_\}$ satisfies the \emph{commutativity} axiom.
Now consider the  goal:
\[
[\emptyset, \mathcal{N}\mathcal{P},\emptyset] \Vdash
\{x^{s(s(0))},y\}=\{s(y),0\} \rightarrow [x + x^{s(s(0))},x * x]=[s(x+y),x]
\]
Note that the theory $\vec{\mathcal{N}\mathcal{P}}^{=}$ includes, among
others,  the rules:
\begin{itemize}
\item $[x,y]=[x',y'] \rightarrow x=x' \wedge y = y'$
\item $\{x,y\}=\{x',y'\} \rightarrow (x=x' \wedge y = y')
\vee (x=y' \wedge y = x')$
\item $s(x) = x \rightarrow \bot$.
\end{itemize}
Application of the first and second
 $\vec{\mathcal{N}\mathcal{P}}^{=}$-rules to the
above clause, plus Boolean simplification,
plus the equations for  $+$, $*$, $(\_)^{(\_)}$,  gives us a 
conjunction of two multiclauses:
\[
x * x =s(y),\; y =0 \rightarrow x + (x * x) = s(x+y) 
\wedge x * x = x 
\]
and
\[
x * x =0,\; y =s(y) \rightarrow x + (x * x) = s(x+y) 
\wedge x * x = x.
\]
But, since $y =s(y)$ simplifies to $\bot$ by the third rule, 
the second multiclause
simplifies to $\top$.  So we just get the first multiclause as the
result of applying the \textbf{EPS} rule to our original goal.
\end{example}

\begin{comment}

Finally, note that the hypotheses $H^{+}_e$ in the theory
$[\overline{X},\mathcal{E},H]$ are
transformed as rules $\vec{H}^{+}_{e_U}$ exactly as for $\vec{E}_U$
(but note the special definition in Section \ref{IndH-REW}
of the rule associated to a non-Horn hypothesis
clause $\Upsilon \rightarrow \Delta$), which are also
added as extra rules to the theory
$\vec{\mathcal{E}}_{\overline{X}_U}^=$.

\end{comment}

\vspace{2ex}
\noindent\textbf{Constructor Variant Unification Left} (\textbf{CVUL}).

\begin{prooftree}
\AxiomC{$\{[\overline{X}\uplus
    \overline{Y}_{\alpha},\; \mathcal{E},\;
(H \cup
\widetilde{\overline{\alpha}}|_{\overline{X}_{\Gamma}})_{\mathit{simp}}
] \Vdash (\Gamma'\rightarrow\Lambda)
\overline{\alpha}\}_{\alpha\in\textit{Unif}^{\hspace{1pt}\Omega}_{\mathcal{E}_1}\!(\Gamma^{\circ})}$}
\UnaryInfC{$[\overline{X},\mathcal{E},H] \Vdash \Gamma,\Gamma'\rightarrow\Lambda$}
\end{prooftree}

\noindent where: (i) $\Gamma$ is a conjunction of
$\mathcal{E}_1$-equalities and $\Gamma'$ does not
contain any $\mathcal{E}_1$-equalities;
(ii) $\overline{X}_{\Gamma}$ is the (possibly empty) set of constants from
$\overline{X}$ appearing in $\Gamma$, and
$X_{\Gamma}$ the set of variables obtained by replacing each
$\overline{x} \in \overline{X}_{\Gamma}$ by a fresh variable
$x$ of the same sort;
(iii) $\Gamma^{\circ}$ is just
$\Gamma \{\overline{x} \mapsto x\}_{\overline{x} \in\overline{X}_{\Gamma}}$, i.e.,
we replace each constant $\overline{x} \in\overline{X}_{\Gamma}$
by its corresponding variable $x \in X_{\Gamma}$
(this makes possible treating the constants
in $\overline{X}_{\Gamma}$ as variables for unification purposes);
(iv) $Y= \mathit{ran}(\alpha)$ is a set of fresh variables,
$Y_{\alpha} \subseteq Y$ is,
by definition,  the set of variables
$Y_{\alpha} = \mathit{vars}(\alpha(X_{\Gamma}))$, and
$\overline{Y}_{\alpha}$ is the corresponding set
of constants, where each $y \in Y_{\alpha}$
is replaced by a fresh constant $\overline{y}$ of the same sort;
 (v) $\overline{\alpha}$ is the
composed substitution
$\overline{\alpha}=
\{\overline{x} \mapsto x\}_{\overline{x} \in\overline{X}_{\Gamma}}
\alpha \{y \mapsto \overline{y}\}_{y \in Y_{\alpha}}$;
 (vi) $\widetilde{\overline{\alpha}}|_{\overline{X}_{\Gamma}}$
is the set of ground equations
$\widetilde{\overline{\alpha}}|_{\overline{X}_{\Gamma}}
= \{\overline{x}=
\overline{\alpha}(\overline{x})\}_{\overline{x}\in\overline{X}_{\Gamma}}$; 
and  (vii)
$\textit{Unif}^{\hspace{2pt}\Omega}_{\mathcal{E}_1}(\Gamma^{\circ})$
denotes the set of constructor $\mathcal{E}_1$-unifiers of $\Gamma^{\circ}$
\cite{var-sat-scp,skeirik-meseguer-var-sat-JLAMP}.

The \textbf{CVUL} simplification rule can be very powerful.  The most powerful
case  is when $\textit{Unif}^{\hspace{2pt}\Omega}_{\mathcal{E}_1}(\Gamma^{\circ}) =
\emptyset$, since then there are no subgoals, i.e., the rule's premise becomes $\top$, and we
 have  \emph{proved} the given  goal.

The main reason for  the extra technicalities (i)--(viii)  is that
they make this simplification rule \emph{more 
generally applicable},
so that it can also be applied 
when $\Gamma$ contains some extra constants in $\overline{X}$.
However, extra care is required to make sure
that the \textbf{CVUL} rule remains sound in this more general setting.
The main idea is to turn such extra constants into variables.  But in fact they
\emph{are} constants.   What can we
do?  That is what the technicalities (i)--(viii), and in particular
the new ground hypotheses 
$\widetilde{\overline{\alpha}}|_{\overline{X}_{\Gamma}}$,
answer.  Since a concrete example is
worth a thousand generic explanations, let us see the
technicalities  (i)--(viii) at work in a simple example.

\begin{example}
Consider an equational theory $\mathcal{E}$ containing
the multiset specification in Example \ref{multiset-ctors}
in Section \ref{GS-FubCall-Subsection} and perhaps many
more things.  For our purposes here, let us assume that it
also includes a Boolean predicate $p: \mathit{MSet} \rightarrow
\mathit{Bool}$ whose defining equations are irrelevant here.  
The theory  $\mathcal{E}$ 
certainly has an FVP subtheory  $\mathcal{E}_{1} \subseteq \mathcal{E}$ 
which contains the subsort-overloaded  $AC$ multiset union  operator
$\_\cup\_$, and a single equation, namely,
$S \cup \emptyset = S$, oriented as the rule
$S \cup \emptyset \rightarrow S$, with $S$  of sort $\mathit{MSet}$.
Now consider the inductive  goal:
\[
x \cup \overline{U} = y \cup V \rightarrow p(V) = \mathit{true}.
\]
where $x,y$ have sort $\mathit{Elt}$,  $V$ has sort  $\mathit{NeMSet}$,
and the constant $\overline{U}$ has also sort $\mathit{NeMSet}$.
So in this example our $\Gamma$ is $x \cup \overline{U} = y \cup V$,
$\overline{X}_{\Gamma} = \{\overline{U}\}$, $X_{\Gamma} = \{ U \}$, 
which we assume is fresh, i.e., it appears nowhere else,
and $\Gamma^{\circ} \equiv x \cup U = y \cup V$.
One of the variant constructor unifiers of $\Gamma^{\circ}$
is the unifier $\alpha
= \{U \mapsto y' \cup W,\;
x \mapsto x',\;
y \mapsto y',\;
V \mapsto x' \cup W\}$.
$Y_{\alpha} =
\mathit{vars}(\alpha(X_{\Gamma}))=\mathit{vars}(\alpha( U)) = \{y',W\}$, and
$\overline{Y}_{\alpha} =\{\overline{y'},\overline{W}\}$.
Therefore,
$\overline{\alpha}$  is the composed substitution:
\[
\{\overline{U} \mapsto U\}\;  \{U \mapsto y' \cup W,\;
x \mapsto x',\;
y \mapsto y',\;
V \mapsto x' \cup W\} \; \{y' \mapsto \overline{y'},\; W \mapsto \overline{W}\}.
\]
That is, the substitution
$\overline{\alpha} =  
\{\overline{U} \mapsto \overline{y'} \cup \overline{W},\;
x \mapsto x',\;
y \mapsto\overline{y'},\;
V \mapsto x' \cup 
 \overline{W}\}$, 
and 
$\widetilde{\overline{\alpha}}|_{\overline{X}_{\Gamma}}=\{\overline{U} = \overline{y'} \cup \overline{W}\}$. 
Therefore, if our original goal
had the form:
\[
[\overline{X},\mathcal{E},H] \Vdash x \cup \overline{U} = y \cup V \rightarrow p(V) = \mathit{true}
\]
then the subgoal associated to the constructor variant unifier
$\alpha$ is:
\[
[ \overline{X} \uplus \overline{Y}_{\alpha}
,\;\mathcal{E},
(H \cup \{\overline{U} = \overline{y'} \cup
\overline{W}\})_{\mathit{simp}}
] 
\Vdash  (p(V) =\mathit{true}) \overline{\alpha}. 
\]
\noindent That is, the subgoal:
\[
[ \overline{X} \uplus \overline{Y}_{\alpha}
,\;\mathcal{E},
(H \cup \{\overline{U} = \overline{y'} \cup
\overline{W}\})_{\mathit{simp}}] 
\Vdash  p(x' \cup \overline{W}) =\mathit{true}.
\]
\noindent Note the very useful, but slight poetic license of disregarding  any
differences between  (universal) Skolem constants and variables
in this extended notation for substitutions.
\end{example}

\begin{comment}

\vspace{2ex}
\noindent\textbf{Constructor Variant Unification Failure Left} (\textbf{CVUFL}).

\begin{prooftree}
\AxiomC{$\top$}
\UnaryInfC{$[\overline{X},\mathcal{E},H] \Vdash \Gamma,\Gamma'\rightarrow\Lambda$}
\end{prooftree}

\noindent where $\Gamma$ is a conjunction of 
$\mathcal{E}_{1_{\overline{X}}}$-equalities, $\Gamma'$
contains no extra such $\mathcal{E}_{1_{\overline{X}}}$-equalities, 
$\Gamma^\degree$ is the conjunction of 
$\mathcal{E}_1$-equalities obtained by replacing
the constants $\overline{x}\in\overline{X}$ by corresponding variables $x\in X$, and
$\textit{Unif}^{\hspace{2pt}\Omega}_{\mathcal{E}_1}(\Gamma^{\degree})=\emptyset$.

\end{comment}

\vspace{2ex}
\noindent\textbf{Constructor Variant Unification Failure Right} (\textbf{CVUFR}).

\begin{prooftree}
\AxiomC{$[\overline{X},\mathcal{E},H] \Vdash \Gamma\rightarrow\Lambda\wedge\Delta $}
\UnaryInfC{$[\overline{X},\mathcal{E},H] \Vdash \Gamma\rightarrow\Lambda \wedge (u=v,\Delta)$}
\end{prooftree}

\noindent where $u=v$ is a $\mathcal{E}_{1_{\overline{X}}}$-equality and
$\textit{Unif}^{\hspace{2pt}\Omega}_{\mathcal{E}_1}((u=v)^{\circ})=\emptyset$.

\vspace{2ex}
\noindent\textbf{Substitution Left} (\textbf{SUBL}).

\noindent We consider two cases: when $x$ is a variable, and 
when $\overline{x}$ is a fresh constant.

\begin{prooftree}
\AxiomC{$[\overline{X},\mathcal{E},H] \Vdash (\Gamma\rightarrow\Lambda)\{x\mapsto u\}$}
\UnaryInfC{$[\overline{X},\mathcal{E},H] \Vdash x=u,\; \Gamma\rightarrow\Lambda$}
\end{prooftree}

\noindent where: (i) $x$ is a variable of sort $s$, $ls(u)\leq s$, and
$x\not\in vars(u)$; and
(ii) $u$ is not a $\Sigma_1$-term and $\Gamma$ contains no
other $\Sigma_1$-equations.
Note that $\_=\_$ is assumed
commutative, so cases $x=u$ and $u=x$ are both covered.

\vspace{2ex}

\noindent When $\overline{x}$ is a fresh constant the \textbf{SUBL}
rule has the form:

\begin{prooftree}
\AxiomC{$[\overline{X} \uplus
  \overline{Y}_{u} ,\mathcal{E},(H \cup \{\overline{x} = \overline{u}\})_{\mathit{simp}}] 
\Vdash (\Gamma\rightarrow\Lambda) (\{y \mapsto \overline{y}\}_{y \in  Y_{u}} \uplus
\{\overline{x}\mapsto \overline{u}\})$}
\UnaryInfC{$[\overline{X},\mathcal{E},H] \Vdash 
\overline{x}=u,\; \Gamma\rightarrow\Lambda$}
\end{prooftree}

\noindent where $\overline{x} \in \overline{X}$ has sort $s$, $ls(u)\leq s$, 
$\overline{x}$ does not appear in $u$, which 
is not a $\Sigma_1$-term, and $\Gamma$ contains no 
$\Sigma_1$-equations; and where $Y_{u}=\mathit{vars}(u)$,
$\overline{Y}_{u}$ are fresh new constants $\overline{y}$ of same sort
for each $y \in Y_{u}$, and $\overline{u} \equiv u \{y \mapsto
\overline{y}\}_{y \in Y_{u}}$. Note again that $\_=\_$ is assumed commutative.

\vspace{2ex}
\noindent\textbf{Substitution Right} (\textbf{SUBR}).

\noindent We again consider two cases: when $x$ is a variable, and 
when $\overline{x}$ is a fresh constant.

\begin{prooftree}
\AxiomC{$[\overline{X},\mathcal{E},H] \Vdash \Gamma\rightarrow x=u$}
\AxiomC{$[\overline{X},\mathcal{E},H] \Vdash (\Gamma\rightarrow\Lambda)\{x\mapsto u\}$}
\BinaryInfC{$[\overline{X},\mathcal{E},H] \Vdash \Gamma\rightarrow\Lambda\wedge x=u$}
\end{prooftree}

\noindent provided (i) $\Lambda \not= \top$,  (ii)
variable $x$ has sort $s$ and  $ls(u)\leq s$, 
and (iii)  $x\not\in vars(u)$.  (i) avoids looping, and (ii) makes
$\{x\mapsto u\}$ an order-sorted substitution.
Cases $x=u$ and $u=x$ are both covered.

\vspace{2ex}

\noindent When $\overline{x}$ is a fresh constant the \textbf{SUBR}
rule has the form:

\begin{prooftree}
\AxiomC{$[\overline{X},\mathcal{E},H] \Vdash \Gamma\rightarrow \overline{x}=u$}
\AxiomC{$[\overline{X} \uplus
  \overline{Y}_{u},\mathcal{E},
(H \cup \{ \overline{x} =  \overline{u}\}) _{\mathit{simp}}] \Vdash
  (\Gamma\rightarrow\Lambda)(\{y \mapsto \overline{y}\}_{y \in Y_{u}} \uplus
\{ \overline{x}\mapsto \overline{u}\})$}
\BinaryInfC{$[\overline{X},\mathcal{E},H] \Vdash
  \Gamma\rightarrow\Lambda\wedge 
\overline{x}=u$}
\end{prooftree}

\noindent provided (i)  $\Lambda \not= \top$,  (ii)
$\overline{x} \in \overline{X}$ has sort $s$ and  $ls(u)\leq s$, and (iii)
$\overline{x}$ does not appear in $u$;
and where $Y_{u}=\mathit{vars}(u)$,
$\overline{Y}_{u}$ are fresh new constants $\overline{y}$ of same sort
for each $y \in Y_{u}$, and $\overline{u} \equiv u \{y \mapsto
\overline{y}\}_{y \in Y_{u}}$. 
Cases $\overline{x}=u$ and $u=\overline{x}$ are both covered.

\vspace{2ex}

\noindent\textbf{Narrowing Simplification} (\textbf{NS}).

\begin{prooftree}
\AxiomC{$\{[\overline{X}
\uplus\overline{Y}_{i,j},\mathcal{E},(H \uplus 
\widetilde{\overline{\alpha}}_{i,j}|_{\overline{X}_{f(\vec{v})}})_{\mathit{simp}}
]
 \Vdash (\Gamma_{i},(\Gamma\rightarrow\Lambda)[r_{i}=u]_{p})
\overline{\alpha}_{i,j}
\}_{ i \in I_{0}}^{j \in J_{i}}$}
\UnaryInfC{$[\overline{X},\mathcal{E},H] \Vdash (\Gamma\rightarrow
\Lambda) [f(\vec{v})=u]_{p}$}
\end{prooftree}

\noindent where:
\begin{enumerate}
\item $f(\vec{v})$ is called the \emph{narrowex of goal}
$\Gamma\rightarrow\Lambda$
\emph{for the equation at position} $p$, $\overline{X}_{f(\vec{v})}$ is
the set of constants from $\overline{X}$ appearing in 
$f(\vec{v})$, $X_{f(\vec{v})}$ are the corresponding fresh
variables $x$ having the same sort as $\overline{x}$
for each $\overline{x} \in \overline{X}_{f(\vec{v})}$,
and $f(\vec{v})^{\circ}$ is the term obtained
by simultaneously replacing each $\overline{x} \in
\overline{X}_{f(\vec{v})}$ by its corresponding
$x \in X_{f(\vec{v})}$.

\item $f$ is a non-constructor symbol in $\Sigma\setminus \Omega$,
the terms $\vec{v}$ are $\Omega$-terms, and $f$ is defined in the
transformed ground convergent theory $\vec{\mathcal{E}}_{U}$
by a family of (possibly conditional) rewrite rules (with constructor
argument subcalls), of the form:
$\{[i]: f(\vec{u_{i}}) \rightarrow r_{i} \;\; \mathit{if} \;\;
\Gamma_{i} \}_{i \in I}$ such that: (i) 
are renamed with \emph{fresh variables} disjoint from those in $X_{f(\vec{v})}$
 and
in $\Gamma\rightarrow\Lambda$; (ii)
as assumed throughout,  for each $i \in I$,
$\mathit{vars}(f(\vec{u_{i}})) 
\supseteq
\mathit{vars}(r_{i}) \cup
\mathit{vars}(\Gamma_{i})$ and the rules are \emph{sufficiently complete},
i.e., can rewrite modulo the axioms $B_{0}$ of $\vec{\mathcal{E}}_{U}$
any $\Sigma$-ground term of the form $f(\vec{w})$, with the $\vec{w}$ ground
$\Omega$-terms; and (iii) are the transformed rules by the
$\vec{\mathcal{E}} \mapsto \vec{\mathcal{E}}_{U}$ transformation
of corresponding rules defining $f$ in $\vec{\mathcal{E}}$
and are such that,
as explained in Footnote  \ref{collapse-foot}, they
never lose their lefthandside's top symbol $f$ in the
$\vec{\mathcal{E}} \mapsto \vec{\mathcal{E}}_{U}$ transformation
due to a $U$-collapse.

\item For each $i \in I$,
  $\mathit{Unif}_{B_{0}}(f(\vec{v})^{\circ}=
f(\vec{u_{i}}))$ 
is a family of most general
$B_{0}$-unifiers\footnote{The fact that  $B_{0}$ may involve axioms
of associativity without commutativity ---for which the number of
unifiers may be infinite--- may be a problem.
In this case, Maude's $B_{0}$-unification algorithm will
either find a complete finite set of unifiers, or will return
a finite set of such unifiers with an incompleteness warning.
In this second case, the ${\bf NS}$ rule application would have to be undone due to
incompleteness.  But  this second case seems unlikely
thanks to two favorable reasons: (1) the unifiers in 
$\mathit{Unif}_{B_{0}}(f(\vec{v}) ^{\circ} =f(\vec{u_{i}}))$ are
\emph{disjoint} unifiers, i.e., no variables are shared
between the two unificands; and (2)  we may safely assume, without loss of generality,
that all the lefthand sides $f(\vec{u_{i}})$ of the rules defining $f$
are \emph{linear} terms (have no repeated variables); 
because if they are not, we can easily  linearize
them and add the associated equalities between the linearized
variables to their rule's condition $\Gamma_{i}$.}
$\mathit{Unif}_{B_{0}}(f(\vec{v})^{\circ}=f(\vec{u_{i}}))=
\{\alpha_{i,j}\}_{j \in J_{i}}$, with $I_{0} = \{i \in I \mid
  \mathit{Unif}_{B_{0}}(f(\vec{v}) ^{\circ}=f(\vec{u_{i}}))
\not= \emptyset\}$, and with $\mathit{dom}(\alpha_{i,j})=
\mathit{vars}(f(\vec{v}))^{\circ} \uplus \mathit{vars}(f(\vec{u_{i}}))$,
 $\mathit{ran}(\alpha_{i,j})$ is a set of  \emph{fresh} variables 
not appearing anywhere, 
$Y_{i,j} = \mathit{vars}(\alpha_{i,j}(X_{f(\vec{v})}))$, and
$\overline{Y}_{i,j}$
the set of \emph{fresh} constants 
$\overline{y}$ of same sort for each variable $y \in Y_{i,j}$, and
$\overline{\alpha}_{i,j}$ 
the composed substitution
$\overline{\alpha}_{i,j} = \{\overline{x} \mapsto x 
\}_{\overline{x} \in \overline{X}_{f(\vec{v})}}
\alpha_{i,j}\{y \mapsto \overline{y}\}_{y\in Y_{i,j}}$.
Finally,  $\widetilde{\overline{\alpha}}_{i,j}|_{\overline{X}_{f(\vec{v})}}$
is the set of ground equations
$\widetilde{\overline{\alpha}}_{i,j}|_{\overline{X}_{f(\vec{v})}}
= \{\overline{x}=
\overline{\alpha}_{i,j}(\overline{x})\}_{\overline{x}\in \overline{X}_{f(\vec{v})}}$.

\item We furthermore assume that: (i) $u$ is a $\Sigma_{1}$-term; and
(ii) all $r_{i}$ in
the set of equations $\{[i]: f(\vec{u_{i}}) \rightarrow r_{i} \;\; \mathit{if} \;\;
\Gamma_{i} \}_{i \in I}$ defining $f$ are also $\Sigma_{1}$-terms,
and  $f$ is not a $\Sigma_{1}$-symbol.
This will ensure that the equations  $(r_{i}=u)\overline{\alpha}_{i,j}$
are all $\Sigma_{1}$-equations.  As earlier, we assume
throughout that the symbol $\_=\_$ is commutative.
\end{enumerate}

\vspace{1ex}

Thanks to condition (4), the effect of this rule if $p$ is a position in the
premise $\Gamma$ of the goal multiclause is that the {\bf CVUL} 
simplification rule will become enabled in all the resulting
subgoals, thus generating a healthy ``chain reaction''
of simplifications.  Given the more restrictive nature of the {\bf CVUFR} rule,
the effect will be more limited
if $p$ is a position in the multiclause's conclusion $\Lambda$.
The intention to ensure some definite progress in the
simplification process is the only reason for adding condition (4)
as a reasonable restriction for automation.  However, the ${\bf NS}$ 
rule makes perfect sense dropping such a restriction, i.e., as a 
simplification rule of the form:

\begin{prooftree}
\AxiomC{$\{[\overline{X}
\uplus\overline{Y}_{i,j},\mathcal{E},(H \uplus 
\widetilde{\overline{\alpha}}_{i,j}|_{\overline{X}_{f(\vec{v})}})_{\mathit{simp}}
]
 \Vdash (\Gamma_{i},(\Gamma\rightarrow\Lambda)[r_{i}]_{p})
\overline{\alpha}_{i,j}
\}_{ i \in I_{0}}^{j \in J_{i}}$}
\UnaryInfC{$[\overline{X},\mathcal{E},H] \Vdash (\Gamma\rightarrow
\Lambda) [f(\vec{v})]_{p}$}
\end{prooftree}

\noindent  keeping the above side conditions (1)--(3),
but now allowing the  narrowex $f(\vec{v})$ to apear
\emph{anywhere} in $\Gamma\rightarrow
\Lambda$.  In particular, $f(\vec{v})$ may well be a proper subterm
of another $\Sigma$-term, and therefore neither the lefthand nor the 
righthand side of an equation.  For purposes of proving soundness,
we will use this more general form of the ${\bf NS}$
inference rule.  In practice, we shall regard ${\bf NS}$ as a 
simplification
rule with a \emph{dual use}: (i) automated with the additional 
restrictions in
(4), and (ii)  applied under user control
with the more general rule, assuming (1)--(3).

\vspace{1ex}

\noindent The ${\bf NS}$ simplification
rule is very closely related, and complements,
the narrowing induction (${\bf NI}$)  inference 
rule to be discussed later.  The only differences
are that: (i)  ${\bf NI}$ adds additional induction
hypotheses to its subgoals, based on the subcalls
appearing in the rules defining $f$; and (ii) in the
${\bf NI}$ rule, the narrowex $f(\vec{v})$ does not
have any constants in $\overline{X}$.  Therefore, these rules
have different restrictions and can cover different situations.
Since  several detailed examples illustrating the use of the ${\bf NI}$ 
rule will be given later, and the notation introduced in both rules
is quite similar, I do not give an example here.

\vspace{2ex}

\noindent\textbf{Clause Subsumption} ($\textbf{CS}$).

\begin{prooftree}
\AxiomC{$[\overline{X},\mathcal{E},H\cup\{\Gamma\rightarrow\Delta\}]
  \Vdash \Gamma'\rightarrow (\mathit{ESC})\theta$}
\UnaryInfC{$[\overline{X},\mathcal{E},H\cup\{\Gamma\rightarrow\Delta\}]
  \Vdash \Gamma'\rightarrow\Lambda'$}
\end{prooftree}

\noindent where $\theta$ is a matching (modulo $\mathit{ACU}$)  substitution
such that $\Gamma'\rightarrow\Lambda' =_{B^{=}} 
(\Gamma,\mathit{ES} \rightarrow \Delta, \mathit{ES}' \wedge \mathit{ESC})
\theta$.  That is, the hypothesis $\Gamma\rightarrow\Delta$
is extended to the pattern $\Gamma,\mathit{ES} \rightarrow \Delta,
\mathit{ES}' \wedge \mathit{ESC}$ for pattern matching multiclauses,
where the variables $\mathit{ES}$ and $\mathit{ES}'$ range over \emph{equation
sets} (with $\mathit{ES}$, appearing on the left side, understood
as a conjunction, and $\mathit{ES}'$, appearing on the right side,
understood as a disjunction), and the variable $\mathit{ESC}$ 
ranges over \emph{equation set conjunctions} (with such equation
sets understood as disjunctions).

This rule
can be fully automated: if a substitution $\theta$
is found such that the goal matches the extended pattern of
 some hypothesis   $\Gamma\rightarrow\Delta$ in the
goal's inductive theory, then the goal can be automatically
simplified.

Let us illustrate the $\textbf{CS}$ rule by means of
 a simple example.  Suppose that
the inductive theory involves the Peano natural numbers
with $0$, $s$, $+$, $*$ and $>$, 
the goal formula is:
\[
\Gamma' \rightarrow\Delta' \; \equiv \;
\mathit{true} = s(n) > m, \; \mathit{true} = m > \overline{z},\; m * m = m 
\rightarrow s(n) > \overline{z} = \mathit{true} \; \wedge \; 
 (n*n =n,\; n > m = \mathit{true}) 
\]
and we have the (non-executable in this case) hypothesis:
\[
\Gamma \rightarrow\Delta \; \equiv \;
x> y = \mathit{true}, \; y > \overline{z} = \mathit{true} \rightarrow 
x > \overline{z} = \mathit{true}. 
\]
Then, the $\textbf{CS}$ rule can be applied with
matching substitution $\theta = \{x \mapsto s(n),\;
y \mapsto m,\; \mathit{ES} \mapsto ( m * m = m),\;
\mathit{ES}' \mapsto \emptyset ,\; \mathit{ESC} \mapsto
(n*n =n,\; n > m = \mathit{true})\}$, and yields the resulting
subgoal:
\[
\Gamma' \rightarrow \theta(\mathit{ESC}) \; \equiv \;
\mathit{true} = s(n) > m, \; \mathit{true} = m > \overline{z},\; m * m = m 
\rightarrow n*n =n,\; n > m = \mathit{true}.
\]
%
% For an example showing how the $\textbf{CS}$ rule is
% applied, together with other inference rules,
% to prove the cancellation law for multiplication
% of non-zero natural numbers see Section \ref{inf-ex}.
%

\vspace{1ex}

\noindent The above $\textbf{CS}$ rule complements the inductive congruence
closure rule (\textbf{ICC}) presented later in at least two ways.
First, it makes it possible to use those induction hypotheses
not usable as rewrite rules, namely, those in the
set  $H_{ne}$, for clause simplification purposes.  
In a sense,  the $\textbf{CS}$ rule  manages to use a clause $\Gamma\rightarrow\Delta$
 in $H_{ne}$ ``as if it were a rewrite rule,''  since we can think of
 such a clause
---in a way entirely similar as how clauses $\Upsilon \rightarrow \Delta'$
in $H^{=}_{\vee,e}$ were oriented as rewrite rules 
$\Upsilon \rightarrow (\Delta' \rightarrow
\top)$--- as a ``conditional rewrite rule'' of the form 
$\Gamma\rightarrow(\Delta \rightarrow \top)$ that
the  $\textbf{CS}$ rule ``applies'' to the head of a conjecture when
part of the conjecture's  head matches the lefthand side $\Delta$
and the condition $\Gamma$ is ``satisfied'' by part of
the conjecture's premise matching it.  Second, since the $\textbf{CS}$
rule is computationally less expensive that \textbf{ICC} \emph{and}
all hypotheses in $H$ can be used for subsumption, the $\textbf{CS}$
rule also complements the heavier artillery of \textbf{ICC}, because
it may blow away some conjectures without having to bring in
\textbf{ICC} itself.  In practice, therefore, $\textbf{CS}$ should be
tried before \textbf{ICC} to pick any low-hanging fruit that may
be around.

\vspace{2ex}
\noindent\textbf{Equation Rewriting (Left and Right)} (\textbf{ERL}\ and\ \textbf{ERR}).

\begin{prooftree}
\AxiomC{$[\overline{X},\mathcal{E},H] \Vdash (u' = v')\theta,\;\Gamma\rightarrow\Lambda$}
\AxiomC{$[\emptyset,\mathcal{E},\emptyset] \Vdash u=v \Leftrightarrow u'=v'$}
\LeftLabel{\textbf{(ERL)}}
\BinaryInfC{$[\overline{X},\mathcal{E},H] \Vdash w = w',\;\Gamma\rightarrow\Lambda$}
\end{prooftree}

\begin{prooftree}
\AxiomC{$[\overline{X},\mathcal{E},H] \Vdash \Gamma\rightarrow\Lambda
\wedge((u'=v')\theta,\;\Delta)$}
\AxiomC{$[\emptyset,\mathcal{E},\emptyset] \Vdash u=v \Leftrightarrow u'=v'$}
\LeftLabel{\textbf{(ERR)}}
\BinaryInfC{$[\overline{X},\mathcal{E},H] \Vdash
\Gamma\rightarrow\Lambda
\wedge(w=w',\; \Delta)$}
\end{prooftree}

\noindent where there is a substitution $\theta$ such that
 $(w = w') =_{B^{=}} (u = v)\theta$,
$\mathit{vars}(u' = v') \subseteq \mathit{vars}(u =
v)$, $(u = v) \succ (u' = v')$, and
$[\emptyset,\mathcal{E},\emptyset] \Vdash u=v \Leftrightarrow u'=v'$
abbreviates \emph{two} different goals:
$[\emptyset,\mathcal{E},\emptyset] \Vdash u=v \rightarrow u'=v'$
and
$[\emptyset,\mathcal{E},\emptyset] \Vdash u'=v' \rightarrow u=v$.
 The equivalence
 $(u=v)\Leftrightarrow (u'=v')$ should be verified ahead of time as a separate proof
 obligation, so that it can be used automatically to simplify many
goals containing any equations $w_{1} = w_{2}$
that are instances of $u=v$ modulo $B^{=}$,
without requiring reproving $(u=v)\Leftrightarrow (u'=v')$ each time.
This provides a general method
to fully \emph{automate} rules \textbf{ERL}
and \textbf{ERR} so that they are 
\emph{subsumed}\footnote{The net effect is not only
that \textbf{EPS} both subsumes \textbf{ERL}
and \textbf{ERR} and becomes more powerful:
by adding such extra rules to $\vec{\mathcal{E}}_{U}^=$,
the \textbf{ICC} simplification rule
discussed next, which also performs simplification
with equality predicates, also becomes more powerful.}
by the \textbf{EPS} simplification
rule when the equality predicate theory $\vec{\mathcal{E}}_{U}^=$
(and therefore the given inductive theory)
has been extended with reductive rules of the form $(u=v)\rightarrow (u'=v')$
proved as lemmas,\footnote{More generally,
the equality predicate theory $\vec{\mathcal{E}}_{U}^=$
can be extended by adding to it \emph{conditional rewrite rules}
that orient inductive theorems of $\mathcal{E}$
or $\mathcal{E}_{U}^=$,
are executable, and keep $\vec{\mathcal{E}}_{U}^=$ operationally
terminating.  For example,
if $c$ and $c'$ are \emph{different} constructors
whose sorts belong to the same connected component having
a top sort, say, $s$, then the
conditional rewrite rule
$x=c(x_{1},\ldots, x_{n})\wedge x=c'(y_{1},\ldots, y_{m})\rightarrow
\bot$, where $x$ has sort $s$ orients
an inductively valid lemma, clearly terminates,
 and  can thus be added to $\vec{\mathcal{E}}_{U}^=$.
In particular, if $p$ is a Boolean-valued predicate,
$p(u_{1},\ldots, u_{n})= \mathit{true}\wedge p(u_{1},\ldots, u_{n})= 
\mathit{false}$ rewrites to $\bot$.}
that is, after we have proved the lemma
$(u=v) \Leftrightarrow (u'=v')$.
Note that what we are exploiting in the \textbf{ERL}
and \textbf{ERR} rules
is the fundamental equivalence:
\[T_{\mathcal{E}} \models \varphi \;\; \Leftrightarrow \;\; 
  \psi
\;\;\; \;\;\; \mathit{iff} \;\;\; \;\;\;
T_{\mathcal{E}^{=}} \models \varphi = \psi
\]
from  Section \ref{EQ-PREDS}, together
with the property $(u = v) \succ (u' = v')$, which
allows us to orient the equation
$(u = v) = (u' = v')$ between Boolean terms as a reductive 
rewrite rule $(u = v) \rightarrow (u' = v')$.

For an example of a useful
rewrite rule of this kind, one can prove for 
natural number addition in, say, Peano notation,
the equivalence $x+z = y +z \Leftrightarrow x = y$
and then use it as the reductive $\Sigma^{=}$-rewrite rule
$x + z = y + z \rightarrow x = y$ to simplify, by means of  the \textbf{ERL}
and \textbf{ERR} rules,
natural arithmetic expressions.

\vspace{2ex}

\noindent {\bf Inductive Congruence Closure}  (\textbf{ICC}).  

\vspace{2ex}

\noindent The next simplification
rule is the  \textbf{ICC} rule.  Since
this is one of the
most complex rules in the inference system, to help the reader I first
explain a simpler 
\emph{congruence closure}  rule, which has
 not been made part of the inference system because
\textbf{ICC} is a more powerful extension of it.
As shown below, \textbf{ICC} is more powerful because
it can draw useful \emph{inductive} consequences (thus the
``\textbf{I}''  in \textbf{ICC}) that cannot be drawn  with congruence
closure alone.

\vspace{1ex}

Furthermore, even though congruence
closure and  \textbf{ICC} can be viewed as  ``modus ponens'' steps
internal to the given goal $\Gamma\rightarrow\Lambda$, where
we assume a convergent equivalent
version of the ground condition $\overline{\Gamma}$
to try to prove the ground conclusion $\overline{\Lambda}$ true, we shall see that, even
when  $\overline{\Lambda}$ cannot be proved true,
\textbf{ICC} can still simplify $\Gamma\rightarrow\Lambda$ to $\top$
by showing $\Gamma$ unsatisfiable.

\begin{example} \label{ICC-EX1}
  Assume a specification $\mathcal{E}$  of the natural numbers in
  Peano notation with the usual definition of $+$ and with predicates
  $\mathit{even}$ and $\mathit{odd}$ defined
  by equations  $\mathit{even}(0) =\mathit{true}$, $\mathit{even}(s(0)) =\mathit{false}$,
   $\mathit{even}(s(s(x)))= \mathit{even}(x)$,
 $\mathit{odd}(0) =\mathit{false}$, and
   $\mathit{odd}(s(x))= \mathit{even}(x)$.  Now consider that, while
   proving a bigger theorem, we encounter a subgoal of the form:
\[
  [\overline{X},\mathcal{E},H] \Vdash \mathit{even}(n+m)=
  \mathit{true},\;
  \mathit{odd}(n+m)= \mathit{true}
  \rightarrow \bot
\]
where $H$ contains the already proved lemma
 $\mathit{even}(x) =\mathit{true} \rightarrow \mathit{odd}(x)
 =\mathit{false}$,
 which is an executable hypothesis with the RPO order
 generated by the symbol order $+ \succ \mathit{odd} \succ \mathit{even} \succ \mathit{not}
 \succ s \succ 0 \succ \mathit{false} \succ \mathit{true}$.
 I shall use this as a running example to show the difference between
 congruence closure and \textbf{ICC}.
\end{example}

\noindent First of all, note  that the  \textbf{EPS} rule will not change the
goal in Example \ref{ICC-EX1}, since the goal's condition cannot be further
simplified by  \textbf{EPS}.  An attempt to simplify this goal by
applying to it a \emph{congruence closure} rule of the form:

\begin{prooftree}
\AxiomC{$[\overline{X},\mathcal{E},H]
  \Vdash \Gamma^{\bullet}\rightarrow \Lambda^{\bullet}$}
\UnaryInfC{$[\overline{X},\mathcal{E},H]
  \Vdash \Gamma \rightarrow\Lambda $}
\end{prooftree}

\noindent where, if  $Y=vars(\Gamma\rightarrow\Lambda)$, then
$\overline{\Gamma}$, resp. $\overline{\Lambda}$, denotes the
ground formula  obtained by instantiating each $y \in Y$
in $\Gamma$, resp. in $\Lambda$, by a fresh constant $\overline{y}$
with same sort as $y$;  and, assuming  that
$\mathcal{E}$ has axioms $B=B_{0} \uplus U$ and
recalling the definition of the congruence closure transformation
$cc^\succ_{B_0}$ in Section \ref{OS-CC},
$\overline{\Gamma}^{\bullet}$
is defined as the set of ground equations
associated to the ground convergent set of rules
$cc^\succ_{B_0} (\overline{\Gamma})$,
and $\Gamma^{\bullet}$
is obtained from $\overline{\Gamma}^{\bullet}$
by mapping back each such  $\overline{y}$
to its corresponding  $y \in Y$.  Likewise,
$\Lambda^{\bullet}$ is obtained in the same manner from
$\overline{\Lambda}^{\bullet}$, where, by definition,
$\overline{\Lambda}^{\bullet}  \in \overline{\Lambda}
!\,_{\vec{\mathcal{E}}_{\overline{X \cup
      Y}_U}^{=}\cup\,\vec{H}^{+}_{e_U} \cup cc^\succ_{B_0}(\overline{\Gamma})}$.  That is,
we simplify  $\overline{\Lambda}$
with the combined power or \textbf{EPS} and the convergent rules
$cc^\succ_{B_0} (\overline{\Gamma})$ to get
$\overline{\Lambda}^{\bullet}$.  Of course, if 
$\overline{\Lambda}^{\bullet} = \top$, the entire
multiclause $\Gamma \rightarrow\Lambda$ is simplified to $\top$.
Note, however, that this congruence closure simplification method leaves
the goal in Example \ref{ICC-EX1} untouched.  This is because:
(i) the rules $\mathit{even}(\overline{n}+\overline{m}) \rightarrow
  \mathit{true},\;
  \mathit{odd}(\overline{n}+ \overline{m}) \rightarrow \mathit{true}$
  are already ground convergent, and (ii) $\bot$ cannot be simplified.

  \vspace{1ex}

  A further observation about the limitations of the above congruence
  closure simplification rule is that  the rules in 
$\vec{\mathcal{E}}_{\overline{X \cup
      Y}_U}^{=}\cup\,\vec{H}^{+}_{e_U}
  \cup cc^\succ_{B_0}(\overline{\Gamma})$ that
  we are using for simplifying $\overline{\Lambda}$
can \emph{preempt each other}, thus becoming less effective.
For example, the lefthand side $l$ of a ground rule $l \rightarrow r$ in 
$cc^\succ_{B_0}(\overline{\Gamma})$ may be reducible
by some rule in $\vec{\mathcal{E}}_{\overline{X \cup
    Y}_U}^{=}\cup\,\vec{H}^{+}_{e_U}$ and may for this reason fail to
be applied to some ground subterm of $\overline{\Lambda}$, simply
because such a subterm was already simplified by the same rule
in
$\vec{\mathcal{E}}_{\overline{X \cup Y}_U}^{=} \cup\,\vec{H}^{+}_{e_U}$,
which can also simplify $l$.
This suggests defining the formula:
\[
 \overline{\Gamma}^{\flat} =_{\mathit{def}}
  \bigwedge_{(l \rightarrow r) \in cc^\succ_{B_0}(\overline{\Gamma})}
(l = r)!_{\vec{\mathcal{E}}_{\overline{X \cup Y}_U}^=\cup\,\vec{H}^{+}_{e_U}}
\]
Note that, as shown in Example \ref{EPS-example}, since we are
performing  \textbf{EPS} simplification,
$(l = r)!_{\vec{\mathcal{E}}_{\overline{X \cup Y}_U}^{=}\cup\,\vec{H}^{+}_{e_U}}$
need not be a single equation: it can be
so; but it can also be a disjunction or conjunction of equations,
or $\top$ or $\bot$.  Therefore, by putting 
$\overline{\Gamma}^{\flat}$ in disjunctive normal form, it can be
either  $\top$, or $\bot$, or a disjunction of conjunctions of the form
$\overline{\Gamma}^{\flat} = \bigvee_{i \in I}\overline{\Gamma}^{\flat}_{i}$.
We can then define
$\vec{\overline{\Gamma}}^{\flat} = \bigvee_{i \in I}
\vec{\overline{\Gamma}}^{\flat}_{i}$,
where
$\vec{\overline{\Gamma}}^{\flat}_{i} =_{\mathit{def}}
\mathit{orient}_{\succ}(\overline{\Gamma}^{\flat}_{i})$,
where $\mathit{orient}_{\succ}(\overline{\Gamma}^{\flat}_{i})$
denotes the set of ground rewrite rules
associated to the ground equations in $\overline{\Gamma}^{\flat}_{i}$
according to the RPO order $\succ$.  This has
several important advantages over the above congruence closure
simplification method.  First, we may have $\overline{\Gamma}^{\flat}
= \bot$, showing $\Gamma$ unsatisfiable and thus proving the given
goal $\Gamma \rightarrow\Lambda$ true.  For
example, one of the rules in $cc^\succ_{B_0}(\overline{\Gamma})$
may
be of  the form $s(\overline{v}) \rightarrow 0$, with $0$ and $s$ the natural
zero and successor symbols, so that
the ground equation $s(\overline{v}) = 0$ will be
\textbf{EPS}-simplified to $\bot$, making $\overline{\Gamma}^{\flat}= \bot$.
Second, a rule in $cc^\succ_{B_0}(\overline{\Gamma})$ may become
much more widely applicable when transformed into a rule in some
$\vec{\overline{\Gamma}}^{\flat}_{i}$.  For example, a ground rule of the
form $s(\overline{u}) \rightarrow s(\overline{v})$ in
$cc^\succ_{B_0}(\overline{\Gamma})$
may become in some $\vec{\overline{\Gamma}}^{\flat}_{i}$
either a rule  $\overline{u}' \rightarrow \overline{v}'$,
or   $\overline{v}' \rightarrow \overline{u}'$, depending
on the RPO-ordering of the simplified forms $\overline{u}'$ and $\overline{v}'$
of $\overline{u}$ and $\overline{v}$,
making such a rule much more widely applicable.
However, even with all these improvements, for our goal
in Example \ref{ICC-EX1} we still get
$\vec{\overline{\Gamma}}^{\flat}=
\{\mathit{even}(\overline{n}+\overline{m}) \rightarrow
  \mathit{true},\;
  \mathit{odd}(\overline{n}+ \overline{m}) \rightarrow \mathit{true}\}$
  because: (i) \textbf{EPS} simplification leaves these rules untouched,
  and (ii) the inductive lemma  $\mathit{even}(x) =\mathit{true} \rightarrow \mathit{odd}(x)
 =\mathit{false}$ cannot be applied to simplify 
$\mathit{odd}(\overline{n}+ \overline{m})$ to $\mathit{false}$, since
the needed assumption $\mathit{even}(x) =\mathit{true}$ cannot be proved
with the rules in  $\vec{\mathcal{E}}_{\overline{X \cup
    Y}_U}^{=}\cup\,\vec{H}^{+}_{e_U}$.  So, even this more powerful
form of congruence closure still leaves the goal in Example \ref{ICC-EX1}
untouched.

\vspace{1ex}

This suggests one more turning of the screw, finally bringing us
to our desired  \textbf{ICC} rule, namely, to
\emph{inter-reduce} the rules in $\vec{\overline{\Gamma}}^{\flat}$
by defining:
\[
 \overline{\Gamma}^{\sharp} =_{\mathit{def}}
  \bigwedge_{(l \rightarrow r) \in cc^\succ_{B_0}(\overline{\Gamma})}
(l = r)!_{\vec{\mathcal{E}}_{\overline{X \cup
      Y}_U}^=\cup\,\vec{H}^{+}_{e_U} \cup
  cc^\succ_{B_0}(\overline{\Gamma})\setminus \{l \rightarrow r\}}
\]
and then defining $\overline{\Gamma}^{\sharp} = \bigvee_{i \in I}
\overline{\Gamma}^{\sharp}_{i}$ and 
$\vec{\overline{\Gamma}}^{\sharp} = \bigvee_{i \in I}
\vec{\overline{\Gamma}}^{\sharp}_{i}$
in the same way as we defined $\overline{\Gamma}^{\flat}$,
$\vec{\overline{\Gamma}}^{\flat}$,
and the $\overline{\Gamma}^{\flat}_{i}$ and
$\vec{\overline{\Gamma}}^{\flat}_{i}$.
This is indeed a powerful enough method to prove our
goal in  Example \ref{ICC-EX1},
since now we can inter-reduce the equation
$\mathit{odd}(\overline{n}+ \overline{m}) = \mathit{true}$
with the rules
 $\vec{\mathcal{E}}_{\overline{X \cup Y}_U}^{=}\cup\,\vec{H}^{+}_{e_U}
 \cup \, \{\mathit{even}(\overline{n}+\overline{m}) \rightarrow
  \mathit{true}\}$, which allow the inductive lemma
$\mathit{even}(x) =\mathit{true} \rightarrow \mathit{odd}(x)
=\mathit{false}$ to be appied to simplify the above equation
to $\mathit{false} = \mathit{true}$, which is then
 \textbf{EPS}-simplified to $\bot$.

\vspace{1ex}

Note that, since we have the Boolean equivalence
$(A \vee B) \Rightarrow C \equiv (A \Rightarrow C) \wedge (B
\Rightarrow C)$, when performing  \textbf{ICC} simplification
of $\overline{\Lambda}$
by means of $\vec{\overline{\Gamma}}^{\sharp}$,
we in general may get not a single subgoal but several of them,
each corresponding to a different set of ground rules $\vec{\overline{\Gamma}}^{\sharp}_{i}$.
By bluring the distinction between the $\_\wedge\_$
and the $\_ , \_$ notation for conjunction, we then get our
desired \textbf{ICC} inference rule:

\begin{prooftree}
\AxiomC{$\{[\overline{X},\mathcal{E},H]
  \Vdash
  \Gamma^{\sharp}_{i} \rightarrow \Lambda^{\sharp}_{i} \}_{i \in I}$}
\LeftLabel{\textbf{(ICC)}}
\UnaryInfC{$[\overline{X},\mathcal{E},H]
  \Vdash \Gamma \rightarrow\Lambda $}
\end{prooftree}

\noindent where: (i) by definition,
$\overline{\Lambda}^{\sharp}_{i}  
 \in \overline{\Lambda}
!\,_{\vec{\mathcal{E}}_{\overline{X \cup
Y}_U}^{=}\cup\,\vec{H}^{+}_{e_U} \cup
\vec{\overline{\Gamma}}^{\sharp}_{i}}$,
and we always pick $\overline{\Lambda}^{\sharp}_{i}= \top$
if $\top \in  \overline{\Lambda}
!\,_{\vec{\mathcal{E}}_{\overline{X \cup
Y}_U}^{=}\cup\,\vec{H}^{+}_{e_U} \cup
\vec{\overline{\Gamma}}^{\sharp}_{i}}$;
(ii) $\Gamma^{\sharp}_{i} \rightarrow \Lambda^{\sharp}_{i}$
is obtained from
$\overline{\Gamma}^{\sharp}_{i} \rightarrow \overline{\Lambda}^{\sharp}_{i}$
by converting back the Skolem constants associated to the variables of
$\Gamma \rightarrow\Lambda$ into those same variables,
and
(iii) the case $\overline{\Gamma}^{\sharp}= \bot$
is the case when there are no conjunctions
in the disjunctive normal form
$\overline{\Gamma}^{\sharp}$,  i.e., $I = \emptyset$,
so that, by convention, the notation
$\{[\overline{X},\mathcal{E},H]
  \Vdash
  \Gamma^{\sharp}_{i} \rightarrow \Lambda^{\sharp}_{i} \}_{i \in \emptyset}$
  then denotes $\top$, because, since $\Gamma$ has been shown inductively
  unsatisfiable, the goal is proved.
  Note that when $\overline{\Gamma}^{\sharp}= \top$, $|I|=1$, and what we
  get is the subgoal $[\overline{X},\mathcal{E},H]
  \Vdash \Lambda'$,
  with $\overline{\Lambda'} \in    \overline{\Lambda}
!\,_{\vec{\mathcal{E}}_{\overline{X \cup
Y}_U}^{=}\cup\,\vec{H}^{+}_{e_U}}$.

\vspace{1ex}

\noindent Note the
close connection between the \textbf{ICC} and
\textbf{EPS} rules:  when the premise
$\Gamma$ of the goal $\Gamma \rightarrow \Lambda$
 is $\top$, then \textbf{ICC}  coincides
 with \textbf{EPS}.
This means that it is useless to apply \textbf{ICC}
when $\Gamma$ is $\top$.  But 
since \textbf{EPS} is a simpler and computationally
 less expensive rule than \textbf{ICC},
when $\Gamma$ is a non-trivial premise
it is a good strategy to always apply \textbf{EPS} 
before applying \textbf{ICC}.

\begin{example} \label{ICC-example}
Consider the \textbf{EPS}-simplified goal  
\[
[\emptyset, \mathcal{N}\mathcal{P},\emptyset] \Vdash
x * x =s(y),\; y =0 \rightarrow x + (x * x) = s(x+y) 
\wedge x * x = x 
\]
from Example  \ref{EPS-example}.  We want to further
simplify it using \textbf{ICC}.
Suppose
that, after adding new constants $\overline{x}$
and $\overline{y}$ to the signature of $\mathcal{N}\mathcal{P}$,
we linearly order its symbols in the chain:
\[
(\_)^{(\_)} \; \succ\; * \;\succ\; + \;\succ\; \overline{x} \;\succ\; \overline{y}
\;\succ\; s \;\succ\; 0.
\]
Then, the congruence closure of the ground equations
$\overline{\Gamma}=\{\overline{x} * \overline{x}
=s(\overline{y}), \overline{y} =0\}$
using the associated RPO order is the set of rules
$\{\overline{x} * \overline{x}  \rightarrow s(0),\overline{y}
\rightarrow 0\}$, which, since it cannot be further simplified by
the rules in $\vec{\mathcal{N}\mathcal{P}}^{=}$, 
 there are
 no induction hypotheses,
and is already inter-reduced,
 is just $\vec{\overline{\Gamma}}^{\sharp}$.
In this case, $\overline{\Lambda} =\overline{x} +
(\overline{x}  * \overline{x})  = s(\overline{x}+\overline{y}) 
\wedge \overline{x} * \overline{x} = \overline{x}$,
which is simplified by 
$\vec{\mathcal{N}\mathcal{P}}^{=} 
\cup \vec{\overline{\Gamma}}^{\sharp}$
(where $\vec{\mathcal{N}\mathcal{P}}^{=}$ includes
the rule $x = x \rightarrow \top$)
to $s(0) = \overline{x}$.
Therefore, 
\textbf{ICC} has simplified our goal to the considerably simpler goal,
\[
[\emptyset, \mathcal{N}\mathcal{P},\emptyset] \Vdash
x* x=s(0),\; y =0 \rightarrow s(0) = x. 
\]
\end{example}

\vspace{2ex}
\noindent\textbf{Variant Satisfiability} ($\textbf{VARSAT}$).

\begin{prooftree}
\AxiomC{$\top$}
\UnaryInfC{$[\overline{X},\mathcal{E},H] \Vdash \Gamma\rightarrow\Lambda$}
\end{prooftree}

\noindent if $\Gamma\rightarrow\Lambda$ is an $\mathcal{E}_1$-formula and
$\neg(\Gamma^\degree\rightarrow\Lambda^\degree)$ is unsatisfiable in $T_{\mathcal{E}_1}$,
where $\Gamma^\degree\rightarrow\Lambda^\degree$ is obtained from
$\Gamma\rightarrow\Lambda$ by replacing constants in $\overline{X}$ by
corresponding variables in $X$.  This rule
can be automated because being a $\mathcal{E}_1$-formula is a
syntactic condition.  Furthermore, since $\mathcal{E}_1$ is FVP and
sufficiently complete on free constructors modulo $A \vee C$ axioms,
satisfiability in $T_{\mathcal{E}_1}$
(and therefore in $T_{\mathcal{E}}$, since
$T_{\mathcal{E}}|_{\Sigma_{1}} \cong T_{\mathcal{E}_1}$)
is decidable \cite{var-sat-scp}.

% The rule could be more effective by adding the possible
% existing  constraints between Skolem constants to $\Gamma$
% but this would require careful thought.

\vspace{2ex}

\noindent One could politely ask: why {\bf VARSAT} and not
a more general rule allowing any combination of decision procedures?
Why not indeed? The purely pragmatic reason why  {\bf VARSAT} 
only supports variant satisfiability is that
\emph{it is already there}.  It has been
used effectively together with the other simplification
rules in discharging many verification conditions in
a large proof effort such as the proof of the IBOS browser
security discussed in \cite{ind-ctxt-rew}.  
Variant satisfiability
makes SMT solving \emph{user-extensible} to an infinite class
of user-definable data types \cite{var-sat-scp}, and therefore
\emph{complements}  the domain-specific decision
procedures supported by current SMT solvers.  One of course wants
\emph{both}: the efficiency of domain-specific
procedures and  {\bf VARSAT}'s  user-definable extensibility.  Therefore,
{\bf VARSAT} should in the future be generalized to
a ${\bf VARSAT}+{\bf SMT}$ simplification rule that combines domain-specific
procedures 
and variant satisfiability.

\subsection{Inductive Rules}\ \\

\noindent\textbf{Generator Set Induction} (\textbf{GSI}).

\vspace{1ex}

\noindent The generator set induction rule is a substantial
generalization of structural
induction on constructors. In both structural induction and its
\textbf{GSI} generalization one inducts on a variable $z$ in the 
conjecture. 

\begin{comment}

 The first
modality of the \textbf{GSI} rule below covers this, most common case.  
However, to widen the scope
of the \textbf{GSI} rule,  in a second modality of the rule, also
described below, we also support inducting on a fresh
constant $\overline{z}$. To induct on a variable $z$ we use the first
rule modality:

\end{comment}

{\small
\begin{prooftree}
\AxiomC{$\{[\overline{X}\uplus\overline{Y}^{\bullet}_{i},\mathcal{E},(H\uplus H_{i})_{\mathit{simp}}] \Vdash 
(\Gamma\rightarrow\bigwedge_{j\in
    J}\Delta_j)\{z\mapsto\overline{u}^{\bullet}_i\}
\}_{1 \leq i \leq n}$}
\UnaryInfC{$[\overline{X},\mathcal{E},H] \Vdash \Gamma\rightarrow\bigwedge_{j\in J}\Delta_j$}
\end{prooftree}
}
\noindent where $z\in vars(\Gamma\rightarrow\bigwedge_{j\in J}\Delta_j)$ has sort $s$,
$\{u_1\cdots u_n\}$ is a $B_{0}$-generator set
 for $s$, with $B_{0}$ the axioms of the
transformed ground convergent theory $\vec{\mathcal{E}}_{U}$,
and with  $Y_i$ fresh variables $Y_i=vars(u_i)$
for $1\leq i\leq n$.
The set of induction hypotheses $H_{i}$ is then, by definition,
the set
\[
H_{i}=\{((\Gamma\rightarrow\Delta_j)\{z\rightarrow\overline{v}\}) 
!\,_{\vec{\mathcal{E}}_{{\overline{X} \uplus\overline{Y}_{i}}_U}^{=}}
\mid v \in \mathit{PST}_{B_{0}, \leq s}(u_{i})\; \wedge \; j \in J \}
\]
where, by definition, $\mathit{PST}_{B_{0}, \leq s}(u_{i})=\{
v \in \mathit{PST}_{B_{0}}(u_{i})\mid ls(v) \leq s\}$
(recall the notion of proper $B_{0}$-subterm
in Definition \ref{B0-subterm-def}, and
Remark \ref{B0-subterm-remark} there),  
and $\overline{v}$ denotes the
instantiation of $v$ by the substitution
$\{y \mapsto \overline{y}\}_{y \in Y_{i}}$.
By \emph{notational convention}: (1)
when $\mathit{PST}_{B_{0}, \leq s}(u_{i}) \not= \emptyset$ and 
the set of induction hypotheses 
$H_{i}$ for constructor term $u_{i}$ are \emph{nontrivial}, i.e.,
at least for some $j\in J$ we have
$((\Gamma\rightarrow\Delta_j)\{z\rightarrow\overline{v}\}) 
!\,_{\vec{\mathcal{E}}_{{\overline{X} \uplus\overline{Y}_{i}}_U}^{=}}
\not= \top$ (trivial induction hypotheses are always omitted),
then
$\overline{Y}_{i}^{\bullet} \equiv \overline{Y}_{i}$ 
are fresh constants $\overline{y}$, with the same sorts, for
each $y \in Y_{i}$,
and $\overline{u}^{\bullet}_i \equiv \overline{u}_i$
denotes the instantiation
 of $u_i$ by the substitution
$\{y \mapsto \overline{y}\}_{y \in Y_{i}}$;
(2) otherwise, i.e., when either $Y_{i} = \emptyset$,
or all the simplified induction hypotheses $H_{i}$
are \emph{trivial}, then,
$\overline{Y}_{i}^{\bullet} \equiv \emptyset$, and
$\overline{u}^{\bullet}_i \equiv u_{i}$.  That is, in case (2)
the  inductive subgoal for  constructor term $u_{i}$
 has no (nontrivial) induction hypotheses, and then
we neither introduce new constants $\overline{Y}_{i}$ nor instantiate
$u_{i}$ with them.  Note, finally, that the new sets of hypotheses
$H \uplus H_{i}$ are simplified for each $1 \leq i \leq n$
in the way described in Section \ref{IndH-REW}.

\begin{comment}

\vspace{2ex}

\noindent To induct on a fresh constant $\overline{z}\in \overline{X}$ 
of sort $s$, we use the
second rule modality:

{\small
\begin{prooftree}
\AxiomC{$\{[\overline{X}
  \uplus\overline{Y}_{i},\mathcal{E},
(H  \uplus \{\overline{z} = \overline{u}_i\} \uplus
  \{((\Gamma\rightarrow\Delta_j)\{\overline{z}\rightarrow\overline{y}\}) 
!\,_{\vec{\mathcal{E}}_{{\overline{X} \uplus\overline{Y}_{i}}_U}^{=}}\}_{\overline{y}\in\overline{Y}_{\!\!i\leq
      s}}^{j\in J})_{\mathit{simp}}] \Vdash 
((\Gamma\rightarrow\bigwedge_{j\in
    J}\Delta_j)\{\overline{z}\mapsto\overline{u}_i\}) !\,_{\vec{\mathcal{E}}_{{\overline{X} \uplus\overline{Y}_{i}}_U}^{=}}
\}_{1 \leq i \leq n}$}
\UnaryInfC{$[\overline{X},\mathcal{E},H] \Vdash \Gamma\rightarrow\bigwedge_{j\in J}\Delta_j$}
\end{prooftree}
}

\noindent where $\overline{z}$ occurs in
$\Gamma\rightarrow\bigwedge_{j\in J}\Delta_j$, and
the meaning of the sets $\overline{Y}_{i}$ and $\overline{Y}_{i \leq s}$
 is exactly as defined for the first rule modality.
The only notational difference is that here $\overline{u}_i$ denotes
the instantiation of $u_i$ by the substitution
$\{y \mapsto \overline{y}\}_{y \in Y_{i}}$, and therefore $\overline{u}_i$
is always a ground term.

\end{comment}

\vspace{1ex}

Let me illustrate the \textbf{GSI} rule with two
examples.  The first example should be familiar to many readers,
but may still be helpful to illustrate the notation and
use of \textbf{GSI}.  The second shows both the
versatility and the proving power added by \textbf{GSI} to
standard structural induction.

\begin{comment}

Although the \textbf{GSI} rule generalizes standard structural
induction on constructors and should be familiar to many readers,
it may be worthwhile to illustrate its use through a couple of simple
examples for two reasons: (i) the fact that in the first modality,
as explained in Footnote \ref{u-i-bar-nonground},
the term $\overline{u}_{i}$ need not be ground; and (ii) since
we are not aware of uses of the second modality in the
literature, it may be of interest to also illustrate its use to better
appreciate its advantages.

\end{comment}

\begin{example}  \label{cons-ind-example}
Suppose we are building lists of
natural numbers with two constructors: $\mathit{nil}$, and
a ``cons'' operator
$\_\cdot \_: \mathit{Nat}\;\; \mathit{List}\rightarrow \mathit{List}$,
and we want to prove the associativity of the list append
operator $\_@ \_: \mathit{List}\;\; \mathit{List}\rightarrow
\mathit{List}$, defined with the usual recursive equations
$\mathit{nil}@L=L$ and $(n \cdot L)@Q= n \cdot (L@Q)$,
where $n,m$ have sort $\mathit{Nat}$, and $L,P,Q,R$ have sort $\mathit{List}$.
Let us call this theory $\mathcal{L}_{\mathit{cons}}$.  That is,
we want to prove the goal:
\[
[\emptyset , \mathcal{L}_{\mathit{cons}},\emptyset] \Vdash 
(L @ P)@Q = L@(P@Q). 
\]
We can do so by applying the \textbf{GSI} rule on variable $L$
with generator set $\{\mathit{nil},m \cdot R\}$, i.e., by structural
list induction.  We get two goals.  Goal (1) for $\mathit{nil}$
easily simplifies to $\top$.  Goal (2) for $m \cdot R$ also
easily simplifies to $\top$. The point, however, is to 
illustrate how this second
subgoal is generated according to rule \textbf{GSI}
and is simplified to $\top$ by the {\bf EPS} simplification rule.
$Y_{2}=\{m,R\}$, and $\mathit{PST}_{B_{0}, \leq \mathit{List} }(m \cdot R )
=\{R\}$. Therefore, we get the subgoal:
\[
[\{ \overline{m},\overline{R}\} ,
\mathcal{L}_{\mathit{cons}},\{(\overline{R} @ P)@Q = 
\overline{R}@(P@Q)\}
] \Vdash 
(\overline{m} \cdot \overline{R} @ P)@Q = \overline{m} \cdot \overline{R} @(P@Q). 
\]
and $H_{2}$ is the induction hypothesis
$(\overline{R} @ P)@Q = \overline{R}@(P@Q)$.
Since in the standard RPO order we have
$(\overline{R} @ P)@Q \succ \overline{R}@(P@Q)$,
$H_{2}$ can be
used as the rewrite rule
$(\overline{R} @ P)@Q \rightarrow \overline{R}@(P@Q)$.
The  {\bf EPS} rule
 can simplify the subgoal
$(\overline{m} \cdot \overline{R} @ P)@Q = \overline{m} \cdot \overline{R} @(P@Q)$
to $\top$ because, using the second recursive equation,
it simplifies to
$\overline{m} \cdot (\overline{R} @ P)@Q = \overline{m} \cdot
\overline{R} @(P@Q)$, which using the induction hypothesis is
rewritten to
$\overline{m} \cdot \overline{R}@(P@Q) = \overline{m} \cdot
\overline{R} @(P@Q)$, which using the equationally-defined
equality predicate rules
in $ \mathcal{L}_{\mathit{cons}}^{=}$ is finally reduced to $\top$.
\end{example}

\begin{example}  \label{+-comm-lemma}
This example similar to
Example \ref{comm-ord-rew} in Section \ref{OR-section}
on proving commutativity of natural number addition
for the Peano naturals, i.e., the goal is
$x+y=y+x$.  But
 here we prove commutativity for the ``twice as fast''
definition of addition: $n+0=n,\; n+s(0)=s(n), \; n+s(s(m))=s(s(n+m))$.
We induct on $x$ applying the ${\bf GSI}$ rule
with generator set $\{0,s(0),s(s(k))\}$.
Subgoal (1), the subgoal  for $0$,
 is  $0+y=y+0$, which simplifies to $0+y=y$ and 
cannot be further simplified.  
We apply again the \textbf{GSI} rule
inducting on $y$.  The sub-subgoals for $0$ and $s(0)$ are
 trivially simplified.
The  simplified sub-subgoal for $s(s(\overline{k}))$
is $s(s(0+\overline{k}))=s(s(\overline{k}))$.
Since $\mathit{PST}_{B_{0}, \leq \mathit{Nat} }(s(s(k)) )
=\{k,s(k)\}$, we get the
induction hypotheses:
$0+\overline{k}=\overline{k}$ and
$0+s(\overline{k})=s(\overline{k})$,
which are both reductive.  The {\bf EPS} rule then first
simplifies this sub-subgoal to
$s(s(\overline{k}))=s(s(\overline{k}))$ using the first hypothesis,
which is then simplified to $\top$ by the equality predicate
rules.

The simplified versions of  the 
other two subgoals of the original goal are:
(2) $s(0)+y=s(y)$, and (3) $s(s(\overline{k'}))+y =
s(s(y+\overline{k'}))$.  The reader can check that,
for both (2) and (3), a further induction
on $y$ followed by {\bf EPS} simplification also discharges them.
It is also interesting to see that, as already happened
in Example \ref{comm-ord-rew}, it is only by 
using  the ``background-theory-aware'' semantics
of ordered rewriting that the entire commutative law
can be proved without any need for extra lemmas.

\vspace{1ex}

\noindent This example illustrates two main points.
First, the versatility of the ${\bf GSI}$ rule
in being able to use the ``right'' generator set 
for each function, based on its recursive equations.
This often substantially helps the formula simplification
process.  The second point is the added power of
generating induction hypotheses,
not just for the \emph{variables} of a
term $u_{i}$ in the generator set whose sort
is smaller than or equal to the induction variable's sort $s$,
but also for all other \emph{proper subterms} of $u_{i}$
whose least sort is smaller than or equal to $s$.
Here this is illustrated by the
fact that $\mathit{PST}_{B_{0}, \leq \mathit{Nat} }(s(s(k)) )
=\{k,s(k)\}$, so that \emph{two} induction hypotheses are 
generated for $s(s(k))$.
In general, of course, this
provides a \emph{stronger} induction principle 
than just generating induction hypotheses
for the variables of $u_{i}$.  This increases the
chances of success in simplifying conjectures by
applying the additional induction hypotheses
to subcalls of the function (or functions)
whose recursive equations suggested the choice of the generator set.
This is analogous to the way in which strong induction on the
natural numbers is stronger than standard induction.
\end{example}

\begin{remark}  The {\bf GSI} rule could be further
generalized by allowing the notion of
$B_{0}$-generator set for a sort $s$ to be, not just a
set $\{u_1\cdots u_n\}$, but a set of
$\Sigma$-\emph{constrained terms}
$\{u_1 | \varphi_{1}\;\cdots\; u_n|\varphi_{n}\}$
such that $T_{\Omega/B_{0},s}=
\bigcup_{1 \leq i \leq n} \llbracket  u_{i} | \varphi_{i}
\rrbracket$, where, by definition, 
$\llbracket  u | \varphi
\rrbracket = \{[v] \in T_{\Omega/B_{0}}
\mid \exists \rho \in[X \rightarrow T_{\Omega}]\;\; s.t. \;\;
[v]  = [u\; \rho] \; \wedge \; \mathcal{E} \vdash \varphi \rho \}$,
where $X = \mathit{vars}( u)$ and
$\varphi$ is a QF $\Sigma$-formula
with  $\mathit{vars}(\varphi) \subseteq \mathit{vars}( u)$.
Why could this be useful?
Because it would allow what one might call
``bespoke generator sets,'' tailor made for
a specific function or set of functions $f$
involving \emph{conditional} equations
in their definition.  The reason why the
${\bf GSI}$ rule was not specified
in this more general form is that,
 the ``bespoke'' nature
of such a generalization is implicitly
provided by another inference rule, namely,
the \textbf{NI} rule to be presented next.
However, if future experimentation
suggests that it would be useful to have
this more general form of the ${\bf GSI}$ rule
explicitly available, it could be added to the present
inference system as a natural generalization.
\end{remark}

\vspace{2ex}

\noindent\textbf{Narrowing Induction} (\textbf{NI}).

\begin{prooftree}
\AxiomC{$\{[\overline{X}
\uplus\overline{Y}^{\bullet}_{i,j},\mathcal{E},(H \uplus 
(H_{i,j}) !\,_{\vec{\mathcal{E}}_{{\overline{X} \uplus \overline{Y}_{i,j}}_U}^{=}})_{\mathit{simp}}
]
 \Vdash (\Gamma_{i},(\Gamma\rightarrow\bigwedge_{l\in L}\Delta_l)[r_{i}]_{p})
\overline{\alpha}^{\bullet}_{i,j}
\}_{ i \in I_{0}}^{j \in J_{i}}$}
\UnaryInfC{$[\overline{X},\mathcal{E},H] \Vdash (\Gamma\rightarrow\bigwedge_{l\in L}\Delta_l) [f(\vec{v})]_{p}$}
\end{prooftree}

\noindent where:
\begin{enumerate}
\item $f(\vec{v})$ does not contain any constants in $\overline{X}$.
 We call $f(\vec{v})$ the \emph{focus narrowex of goal}
$\Gamma\rightarrow\bigwedge_{l\in L}\Delta_l$
\emph{at position} $p$.

\item $f$ is a non-constructor symbol in $\Sigma\setminus \Omega$,
the terms $\vec{v}$ are $\Omega$-terms, and $f$ is defined in the
transformed ground convergent theory $\vec{\mathcal{E}}_{U}$
by a family of (possibly conditional) rewrite rules (with constructor
argument subcalls), of the form:
$\{[i]: f(\vec{u_{i}}) \rightarrow r_{i} \;\; \mathit{if} \;\;
\Gamma_{i} \}_{i \in I}$ such that: (i)
are renamed with \emph{fresh variables} disjoint from those in $X$ and
in $\Gamma\rightarrow\bigwedge_{l\in L}\Delta_l$;
(ii) as assumed throughout, for each $i \in I$,
$\mathit{vars}(f(\vec{u_{i}})) 
\supseteq
\mathit{vars}(r_{i}) \cup
\mathit{vars}(\Gamma_{i})$ and the rules are \emph{sufficiently complete},
i.e., can rewrite modulo the axioms $B_{0}$ of $\vec{\mathcal{E}}_{U}$
any $\Sigma$-ground term of the form $f(\vec{w})$, with the $\vec{w}$ ground
$\Omega$-terms; and (iii) are the transformed rules by the
$\vec{\mathcal{E}} \mapsto \vec{\mathcal{E}}_{U}$ transformation
of corresponding rules defining $f$ in $\vec{\mathcal{E}}$ and are such that,
as explained in Footnote  \ref{collapse-foot},
never lose their lefthandside's top symbol $f$ in the
$\vec{\mathcal{E}} \mapsto \vec{\mathcal{E}}_{U}$ transformation
due to a $U$-collapse.

\item For each $i \in I$,
  $\mathit{Unif}_{B_{0}}(f(\vec{v})=f(\vec{u_{i}}))$ 
is a family of $B_{0}$-unifiers
$\mathit{Unif}_{B_{0}}(f(\vec{v})=f(\vec{u_{i}}))=
\{\alpha_{i,j}\}_{j \in J_{i}}$, with $I_{0} = \{i \in I \mid
  \mathit{Unif}_{B_{0}}(f(\vec{v})=f(\vec{u_{i}}))
\not= \emptyset\}$, and with $\mathit{dom}(\alpha_{i,j})=
\mathit{vars}(f(\vec{v})) \uplus \mathit{vars}(f(\vec{u_{i}}))$,
 $Y_{i,j} = \mathit{ran}(\alpha_{i,j})$ \emph{fresh} variables 
not appearing anywhere, 
$\overline{Y}_{i,j}$
 denotes the set of \emph{fresh} constants 
$\overline{y}$ of same sort for each variable $y \in Y_{i,j}$, and
$\overline{\alpha}_{i,j}$ 
denotes the composed substitution
$\overline{\alpha}_{i,j} = 
\alpha_{i,j}\{y \mapsto \overline{y}\}_{y\in Y_{i,j}}$.

\item $H_{i,j}$ is defined by cases, depending on where position
$p$ occurs in $\Gamma\rightarrow\bigwedge_{l\in L}\Delta_l$.
{\bf Case} (1): If $p$ occurs in $\Gamma$, then
$H_{i,j} =  \{(\Gamma\rightarrow\Delta_l) 
\overline{\gamma} \mid \;\; l \in L \; \wedge \; 
 f(\vec{w})  \in  \mathit{SSC}([i])\; \wedge \; 
f(\vec{w}) \alpha_{i,j} =_{B_{0}} f(\vec{v}) \gamma\}$.
{\bf Case} (2): If  $p$ occurs in some $\Delta_{k}$ in $\bigwedge_{l\in
  L}\Delta_l$, then
$H_{i,j} =  \{(\Gamma\rightarrow\Delta_k) 
\overline{\gamma} \mid \;\; 
 f(\vec{w})  \in  \mathit{SSC}([i])\; \wedge \; 
f(\vec{w}) \alpha_{i,j} =_{B_{0}} f(\vec{v}) \gamma\}$.
In both cases $\overline{\gamma}$ is, by definition, the
composed substitution
$\overline{\gamma} = 
\gamma \{y \mapsto \overline{y}\}_{y \in Y_{i,j}}$.

\item By \emph{notational convention}, (1) 
when the simplified induction hypotheses
$(H_{i,j}) !\,_{\vec{\mathcal{E}}_{{\overline{X} \uplus \overline{Y}_{i,j}}_U}^{=}}$
are \emph{nontrivial} (i.e., non-empty and at least one
 does not simplify to  $\top$),
then $\overline{Y}^{\bullet}_{i,j}=
\overline{Y}_{i,j}$, and
$\overline{\alpha}^{\bullet}_{i,j} = \overline{\alpha}_{i,j}$;
(2) otherwise, i.e., when the simplified induction hypotheses
of the $i,j$-th inductive subgoal are \emph{trivial},
then, $Y^{\bullet}_{i,j} = \emptyset$,
and $\overline{\alpha}^{\bullet}_{i,j} = \alpha_{i,j}$.
\end{enumerate}

\vspace{1ex}

\noindent Note that, by the way the induction hypotheses
$H_{i,j}$ are defined in (4) above, if $\bigwedge_{l\in L}\Delta_l$
contains more than one disjunction $\Delta_l$, i.e., if the
goal is properly a multiclause and not just a clause, there
are definite advantages in choosing a positon $p$ for the focus
narrowex in the goal's premise $\Gamma$, since this will have
a ``gunshot'' effect of producing as many induction hypotheses
as conjuncts $\Delta_l$ in $\bigwedge_{l\in L}\Delta_l$,
which may simultaneously help in proving all such conjuncts
(see Section \ref{inf-ex} for an example).

\vspace{1ex}

\noindent The above, somewhat long,
 list of technical details specifying the \textbf{NI} rule
is of course essential for its precise definition, proof of
soundness, and correct implementation.  But it may obscure
 its intuitive meaning, which can be clarified and illustrated 
by means of an example.

\begin{example}  \label{reverse-example}
Consider a signature of constructors $\Omega$ for
non-empty lists of elements (the constructors for such elements are
irrelevant), with sorts $\mathit{Elt}$ and $\mathit{List}$,
a subsort inclusion $\mathit{Elt} < \mathit{List}$,
an associative ($A$) list concatenation operator
$\_\cdot\_ : \mathit{List} \;\; \mathit{List} \rightarrow
\mathit{List}$,
and a convergent and sufficiently complete list reverse function 
$\mathit{rev}: \mathit{List} \rightarrow
\mathit{List}$, defined by the equations oriented as rules,
$[1]: \mathit{rev}(x) \rightarrow x$ (which has no subcalls),
 and $[2]: \mathit{rev}(x\; \cdot \; L)\rightarrow \mathit{rev}(L) \;
\cdot \;x$, with $\mathit{SSC}([2]) = \{\mathit{rev}(L)\}$,
where $x,x',y,y'$ have sort $\mathit{Elt}$, and $L,P,Q$ have sort
$\mathit{List}$.   Note that in this example $U = \emptyset$, and therefore
$\vec{\mathcal{E}}=\vec{\mathcal{E}}_{U}$.

We would like to prove the inductive lemma:
\[\mathit{rev}(Q \; \cdot \; y)=y \;
\cdot \; \mathit{rev}(Q).
\]
We can do so by applying the \textbf{NI} rule to the conjecture's
lefthand side as our focus narrowex, i.e., by trying to 
\emph{narrow}\footnote{For unconditional rules,
``narrowing'' is meant here in the sense (generalized to the
order-sorted case)
of narrowing a term with rewrite rules modulo axioms $B$ \cite{JKK83}.  For
conditional rules, it is meant
 in the sense of order-sorted
\emph{constrained narrowing} modulo axioms $B$
proposed in \cite{conditional-narrowing-SCP}.  The condition
$\Gamma_{i}$ of rule [i] is then carried along as a constraint, which
in the  ${\bf NI}$  inference rule is then added to the premise
of subgoal $i,j$ ---after instantiating  it by the associated
substitution--- for
each $j \in J_{i}$.}
 it with the two
rules defining $\mathit{rev}$.  Rule $[1]$ has no $A$-unifiers.
Rule $[2]$ has two most general $A$-unifiers, namely,
$\alpha_{2,1}= \{x \mapsto x',\; Q \mapsto x',\; y \mapsto y',\; L
\mapsto y'\}$, with $Y_{2,1}=\{x',\;y'\}$,
and
$\alpha_{2,2}= \{x \mapsto x',\; y \mapsto y', \;
L \mapsto P\; \cdot \;  y', \; Q \mapsto x' \; \cdot \; P\}$,
with $Y_{2,1}=\{x',\;y',P\}$.  

\noindent With $\alpha_{2,1}$,
since for the subcall $\mathit{rev}(L)$ we get the instance
$\mathit{rev}(L) \alpha_{2,1}= \mathit{rev}(y')$, which does
\emph{not} match the focus narrowex
$\mathit{rev}(Q \; \cdot \; y)$ 
as an instance modulo $A$,
$H_{2,1} = \emptyset$, and we just get the subgoal
$\mathit{rev}(y')\; \cdot \; x' =y' \;
\cdot \; \mathit{rev}(x')$,
which simplifies to $\top$.

\noindent With $\alpha_{2,2}$,
 the subcall $\mathit{rev}(L)$ yields the instance
$\mathit{rev}(L) \alpha_{2,2}= \mathit{rev}(P \; \cdot \; y')$, which \emph{does}
match the focus narrowex
$\mathit{rev}(Q \; \cdot \; y)$ 
as an instance modulo $A$
with substitution $\gamma = \{Q \mapsto P,\; y \mapsto y' \}$.
Therefore, $\overline{Y}^{\bullet}_{2,2}
= \overline{Y}_{2,2}=\{\overline{x'},\; \overline{y'},\; 
\overline{P}\}$, 
$\overline{\alpha}^{\bullet}_{2,2}=\widehat{\overline{\alpha}}_{2,2}$,
and
$\overline{\gamma}=
\{Q \mapsto \overline{P},\; y \mapsto\overline{y'} \}$.
And we get
$H_{2,2} = \{ 
\mathit{rev}(\overline{P} \; \cdot
\;\overline{y'})= \overline{y'}
 \; \cdot
\; \mathit{rev}(\overline{P})\}$ (which is already simplified),
and the subgoal 
$\mathit{rev}(\overline{P} \; \cdot
\;\overline{y'}) \; \cdot \overline{x'} = \overline{y'}
 \; \cdot
\; \mathit{rev}(\overline{x'} \; \cdot \;\overline{P})$,
which simplifies to
$\mathit{rev}(\overline{P} \; \cdot
\; \overline{y'}) \; \cdot\overline{x'} = \overline{y'}
 \; \cdot \; \mathit{rev}(\overline{P}) \; \cdot \; \overline{x'}$,
and can, using $H_{2,2}$,
 be further simplified to $\top$ by means of the {\bf EPS}
simplification rule, thus finishing the proof of the lemma
by a single application of  {\bf NI} followed by {\bf EPS} simplification.
\end{example}

\begin{remark}  The ${\bf NI}$ rule can be further generalized to take
into account the possibility that a defined binary function symbol $f$
may enjoy axioms such as associativity or associativity-commutativity.
For example, we might have an $AC$ defined function symbol $\_*\_$
for multiplication of natural numbers.  The issue this raises
is that we may wish to narrow on a focus narrowex that is
not syntactically present in the given  multiclause, but \emph{is}
present modulo the given axioms.  For example, the multiclause
may contain the subexpression $(x*s(x))*y$ and we may
wish to narrow on $x*y$, which is not a proper subterm of it,
but \emph{is} a proper subterm modulo $AC$ 
(see Section \ref{GS-FubCall-Subsection}).  The generalization of
rule  ${\bf NI}$ allowing us to do this is straightforward.  We
replace the conclusion part of the inference rule by the more general
conclusion:
\[
[\overline{X},\mathcal{E},H] \Vdash \Gamma' \rightarrow\bigwedge_{l\in
  L}\Delta'_l
\]
adding to the side conditions (1)--(6) for the ${\bf NI}$
rule, just as originally stated, the extra side
condition (0):
\[
 (\Gamma'\rightarrow\bigwedge_{l\in L}\Delta'_l) =_{B}
 (\Gamma \rightarrow\bigwedge_{l\in L}\Delta_l) [f(\vec{v})]_{p}
\]
and keeping the ${\bf NI}$ rule's premises the same.
In what follows, I will assume without further ado that the  ${\bf NI}$ rule
can be applied in this generalized from, which frees it
from what might be called the ``slavery of syntax.''
\end{remark}

\vspace{2ex}
\noindent\textbf{Existential} ($\mathbf{\exists}$).

\begin{prooftree}
\AxiomC{$[\emptyset,\mathcal{E},\emptyset] \Vdash I(\Gamma\rightarrow\Lambda)$}
\UnaryInfC{$[\emptyset,\mathcal{E},\emptyset] \Vdash (\exists \chi)(\Gamma\rightarrow\Lambda)$}
\end{prooftree}

\noindent where $\chi$ is a Skolem signature,
$\Gamma\rightarrow\Lambda$
is a $\Sigma \cup \chi$-multiclause, 
and $I: \chi \rightarrow
\mathcal{E}$ is a theory interpretation (see Section \ref{Skolem-subsection}).
Note that the $(\exists)$ rule only applies when the inductive theory is
$[\emptyset,\mathcal{E},\emptyset]$, that is,
at the beginning of the inductive reasoning process,
and that $I$ must be provided by the user as a witness
(for example, as a view in Maude \cite{maude-book}).
See Section  \ref{Skolem-subsection} for a simple
example of an arithmetic formula of the
form $(\exists \chi)(\Gamma\rightarrow\Lambda)$
and its associated theory interpretation
$I: \chi \rightarrow
\mathcal{E}$.

\vspace{2ex}
\noindent\textbf{Lemma Enrichment} ($\textbf{LE}$).

\begin{prooftree}
\AxiomC{$[\overline{X}_{0},\mathcal{E},H_{0}] \Vdash
  \Gamma'\rightarrow\bigwedge_{j\in J}\Delta'_j
\;\;\;\;\;\; \;\;\;\;\;\; \;\;\;\;\;\;
 [\overline{X},\mathcal{E},H] \Vdash H_{0}$}
\AxiomC{$[\overline{X},\mathcal{E},(H\uplus
\{\Gamma'\rightarrow\Delta'_j\}_{j\in J})_{\mathit{simp}}] \Vdash \Gamma\rightarrow\Lambda$}
\BinaryInfC{$[\overline{X},\mathcal{E},H] \Vdash \Gamma\rightarrow\Lambda$}
\end{prooftree}

\noindent where 
$\emptyset \subseteq \overline{X}_{0} \subseteq \overline{X}$.
Common cases of use include
either: (a) $\overline{X}_{0} =  H_{0} = \emptyset$,
or (b) $\overline{X}_{0} =  \overline{X}$ and 
$H_{0} = H$.  The use of this inference rule is illustrated in 
Section \ref{inf-ex}.

\vspace{2ex}
\noindent\textbf{Split} ($\textbf{SP}$).

\begin{prooftree}
\AxiomC{$\{ [\overline{X},\mathcal{E},H] \Vdash
  \Gamma_i\theta,\Gamma\rightarrow\Lambda \}_{i\in I}
\;\;\;\;\;\; \;\;\;\;\;\; \;\;\;\;\;\;
 [\overline{X},\mathcal{E},H] \Vdash H_{0}$}
\AxiomC{$[\overline{X}_{0},\mathcal{E}, H_{0} ] \Vdash
\mathit{cnf} (\bigvee_{i \in I} \Gamma_{i})$}
\BinaryInfC{$[\overline{X},\mathcal{E},H] \Vdash \Gamma\rightarrow\Lambda$}
\end{prooftree}

\noindent where $\emptyset \subseteq  \overline{X}_{0} \subseteq 
\overline{X}$,
$vars((\bigvee_{i \in I} \Gamma_{i} )\theta)\subseteq
vars(\Gamma\rightarrow\Lambda)$, and 
$\mathit{cnf} (\bigvee_{i \in I} \Gamma_{i})$
denotes the conjunctive normal form of
$\bigvee_{i \in I} \Gamma_{i}$, which is a multiclause
of the general form $\Lambda'$.

\vspace{1ex}

\noindent This rule is, on purpose,
very general. In many cases we may have
either: (a) $\overline{X}_{0} =  H_{0} = \emptyset$,
or (b) $\overline{X}_{0} =  \overline{X}$ and 
$H_{0} = H$.  In either case (a) or (b),
two very common special cases
are: (i)  $\Gamma_{i}=(u_{i}=v_{i})$, $i \in I$, and (ii) the
even simpler subcase of a disjunction
$u = \mathit{true}\; \vee \;   u = \mathit{false}$,
where $u$ is a term of sort $\mathit{Bool}$ and
$\theta$ is the identity substitution.

\vspace{2ex}
\noindent\textbf{Case} ($\textbf{CAS}$).

\noindent This rule has two modalities: one for $z$ a variable,
and another for $\overline{z}$ a (universal) Skolem constant.
For $z$ a variable the rule is:

\begin{prooftree}
\AxiomC{$\left\{[\overline{X},\mathcal{E},H] \Vdash (\Gamma\rightarrow\Lambda)\{z\mapsto u_i\}\right\}_{1\leq i\leq n}$}
\UnaryInfC{$[\overline{X},\mathcal{E},H] \Vdash \Gamma\rightarrow\Lambda$}
\end{prooftree}

\noindent where $z\in vars(\Gamma\rightarrow\Lambda)$ has sort $s$ and $\{u_1,\cdots,u_n\}$
is a $B_{0}$-generator
set for sort $s$ with the $u_i$ for $1\leq i\leq n$ having fresh variables.

\vspace{2ex} 

\noindent For a fresh constant $\overline{z}\in \overline{X}$ 
of sort $s$ the rule is:

\begin{prooftree}
\AxiomC{$\{[\overline{X}
  \uplus\overline{Y}_{i},\mathcal{E},
(H  \uplus \{\overline{z} = \overline{u}_i\})_{\mathit{simp}}] \Vdash 
(\Gamma\rightarrow\Lambda)\{\overline{z}\mapsto\overline{u}_i\}
\}_{1 \leq i \leq n}$}
\UnaryInfC{$[\overline{X},\mathcal{E},H] \Vdash \Gamma\rightarrow
\Lambda$}
\end{prooftree}

\noindent where $\overline{z}$ occurs in
$\Gamma\rightarrow\Lambda$, 
$Y_{i} = \mathit{vars}(u_i)$, $\overline{Y}_{i}$
are the corresponding new fresh constants, and
 $\overline{u}_i \equiv u_{i} \{y \mapsto \overline{y}\}_{y \in
   Y_{i}}$.

\noindent The use of this inference rule is illustrated in 
Section \ref{inf-ex}.

\vspace{2ex}

\vspace{2ex}
\noindent\textbf{Variable Abstraction} ($\textbf{VA}$).

\begin{prooftree}
\AxiomC{$[\overline{X},\mathcal{E},H] \Vdash u=v[x]_{p},\; z =w, \;\Gamma\rightarrow\Lambda$}
\UnaryInfC{$[\overline{X},\mathcal{E},H] \Vdash u=v[w]_{p},\;\Gamma\rightarrow\Lambda$}
\end{prooftree}

\noindent where $z$ is a fresh  variable whose sort is the least
sort of subterm $w$ of $v$ at position $p$.  $\textbf{VA}$ is particularly
effective when the  equation $ u=v[z]_{p}$ resulting from the abstraction is a
$\Sigma_{1}$-equation (this may require several applications of $\textbf{VA}$),
since then it can be unified away by
$\textbf{CVUL}$. For several examples illustrating the application and 
usefulness of the $\textbf{VA}$ rule, see Section \ref{inf-ex}.

\begin{comment}

\vspace{2ex}
\noindent\textbf{Variable Abstraction} ($\textbf{VA}$).

\begin{prooftree}
\AxiomC{$[\overline{X},\mathcal{E},H] \Vdash u=w,\; x_1=v_1,\cdots,\;x_n=v_n,\;\Gamma\rightarrow\Lambda$}
\UnaryInfC{$[\overline{X},\mathcal{E},H] \Vdash u=v,\;\Gamma\rightarrow\Lambda$}
\end{prooftree}

\noindent where $u$ and $w$ are $\Sigma_{1}$-terms, $v$ is not so, and
$x_{1},\ldots, x_{n}$ are fresh variables
whose sorts are respectively the least sorts
of the $v_{1},\ldots,v_{n}$, which are subterms of $v$
such that their top symbols are not in $\Sigma_{1}$,
and $v = w \{x_{1} \mapsto v_{1},\ldots, x_{n} \mapsto v_{n}\}$.

\noindent For some examples illustrating the application and 
usefulness 
 of the $\textbf{VA}$ rule, 
see Section \ref{inf-ex}.

\end{comment}

\begin{comment}
\vspace{2ex}
\noindent\textbf{Lemma Strengthening} ($\textbf{LS}$).

\begin{prooftree}
\AxiomC{$[\overline{X},\mathcal{E},H\cup\{\Gamma'\rightarrow x=u\}] \Vdash (\Gamma'\theta,\Gamma\rightarrow\Lambda)\{y\mapsto u\theta\}$}
\UnaryInfC{$[\overline{X},\mathcal{E},H\cup\{\Gamma'\rightarrow x=u\}] \Vdash \Gamma'\theta,\Gamma\rightarrow\Lambda$}
\end{prooftree}

\noindent where $u$ is a ground term, $y \equiv \theta(x)$ is a
variable appearing in $\Gamma'\theta,\Gamma\rightarrow\Lambda$,
and $ls(y) \geq ls(u\theta)$.
\end{comment}

\vspace{2ex}
\noindent\textbf{Equality} ($\textbf{EQ}$).

\begin{prooftree}
\AxiomC{$[\overline{X},\mathcal{E},H] 
\Vdash (\Gamma\rightarrow\Lambda)[v \theta\gamma]_{p}$}
\UnaryInfC{$[\overline{X},\mathcal{E},H]
  \Vdash 
(\Gamma\rightarrow\Lambda)[w]_{p}$}
\end{prooftree}

\noindent where: (i) $w=_{B_{0}}u\theta\gamma$;
(ii) $\Gamma'\rightarrow u=v$ is
a conditional equation in either $\mathcal{E}$  or $H$; (iii) $\theta$
is a user-provided,  possibly partial substitution;
and (iv) $\theta \gamma$ satisfies $\Gamma'$, i.e.,
 for
each $t=t'$ in the hypothesis' condition $\Gamma'$,
$t\theta \gamma !\,_{\vec{\mathcal{E}}_{\overline{X}_U} \cup\,\vec{H}^{+}_{e_U}}
=_{B_{0}} t'\theta \gamma!\,_{\vec{\mathcal{E}}_{\overline{X}_U}
  \cup\,\vec{H}^{+}_{e_U}}$.  Since $\_=\_$ is assumed
    commutative, the user
    \emph{chooses}\footnote{In general, there may exist more than
      one position $p$ (and more than one $\gamma$) at which
      $w=_{B_{0}} u\theta
      \gamma$ occurs
      in the goal multiclause.  In such a case, an implementation
      of this rule should present the various matches and positions found to
    let the user choose the desired ones.}
 the term $u$ in the conditional equation ($u$ is displayed on
 the left of the equation, but could syntactically appear on the right)
 such  that $u \theta$ will be
$B_{0}$-matched by $w$ at position $p$.  Since 
$\Gamma'\rightarrow u=v$ may be a non-executable hypothesis, there may be extra 
variables in $v$ and/or $\Gamma'$ not appearing in $u$, so that
specifying the partial substitution $\theta$ may be necessary.
In simple cases $\theta$ may not be needed and $\gamma$ may be
found just by $B_{0}$-matching the goal's subterm  $w$ with the pattern $u$.
Rule $\textbf{EQ}$ allows  performing
one step of equational inference with a 
possibly conditional equational hypothesis
in a user-guided manner.

\vspace{2ex}
\noindent\textbf{Cut}.

\begin{prooftree}
\AxiomC{$[\overline{X},\mathcal{E},H] \Vdash
  \Gamma\rightarrow \Gamma'$}
\AxiomC{$[\overline{X},\mathcal{E},H] \Vdash \Gamma,\Gamma' \rightarrow\Lambda$}
\BinaryInfC{$[\overline{X},\mathcal{E},H] \Vdash \Gamma\rightarrow\Lambda$}
\end{prooftree}

\noindent where $\Gamma'$ is a conjunction of equalities and
$\mathit{vars}(\Gamma') \subseteq \mathit{vars}(\Gamma
\rightarrow \Lambda)$.  \textbf{Cut} can be viewed as a generalized
\emph{modus ponens} rule, since when $\Gamma \equiv \top$ it actually
becomes modus ponens.  In relation to \textbf{ICC}, we can regard
\textbf{ICC} as a modus ponens rule \emph{internal} to the
goal $\Gamma\rightarrow\Lambda$ viewed as an implication;
whereas \textbf{Cut} is a modus ponens rule \emph{external} to
the goal  $\Gamma\rightarrow\Lambda$   viewed as a formula, where
the user has complete freedom to choose a suitable  $\Gamma'$ 
that will help in proving the original goal.

\begin{example}
Consider an equational theory $\mathcal{E}$ containing
the multiset specification in Example \ref{multiset-ctors}
in Section \ref{GS-FubCall-Subsection}.  For concreteness,
let us identify the sort $\mathit{Elt}$ with the sort
$\mathit{Nat}$ for the Peano natural numbers, to which we have added
an equality predicate $\_\cdot\! =\! \cdot\_$ on naturals.
The  membership predicate $\_\in\_$ is defined by the equations:
$n \in \emptyset = false$,
$n \in m = n \cdot\! =\! \cdot m$, and
$n \in (m \cup U) = (n \cdot\! =\! \cdot m) \;\; \mathit{or} \;\;  n \in U$,
with $n,m$ of sort $\mathit{Nat}$ and $U,V$ of sort $\mathit{NeMSet}$.
Consider now the goal:
\[n \in V = \mathit{false},\; n \cup U = V \rightarrow \bot
\]
We can prove this goal by applying  \textbf{Cut}
  with $\Gamma' = n \in (n \cup U) = n \in V$.  That is,
  we need to prove the subgoals:
(1) $n \in V = \mathit{false},\; n \cup U = V \rightarrow n \in (n \cup U) = n \in V$,
and (2) $n \in V = \mathit{false},\; n \cup U = V , \; n \in (n \cup
U) = n \in V \rightarrow \bot$.  Subgoal (1) can be discharged by
\textbf{ICC} followed by \textbf{EPS}.  Appliying \textbf{EPS} to
subgoal (2) we get $(2')$ $n \in V = \mathit{false},\; n \cup U = V ,
\; true = n \in V \rightarrow \bot$, which can also be discharged by \textbf{ICC} followed by \textbf{EPS}.
\end{example}

\vspace{2ex}

The main property about the above inference system is the following
Soundness Theorem, whose proof is given in Appendix \ref{Proof-Sound-Theo}:

\begin{theorem}[Soundness Theorem] \label{Sound-Theo}
If a closed proof tree can be built from a goal of the form
$[\overline{X},\mathcal{E},H]\Vdash\Gamma\rightarrow\Lambda$,
then
$[\overline{X},\mathcal{E},H]\models \Gamma\rightarrow\Lambda$.
\end{theorem}

Although the Soundness Theorem is stated in full generality, in
practice, of course, its main application will be
to initial goals of the form 
$[\emptyset,\mathcal{E},\emptyset]\Vdash \Gamma \rightarrow \Lambda$
or $[\emptyset,\mathcal{E},\emptyset]\Vdash  (\exists \chi)(\Gamma 
\rightarrow \Lambda)$.

\vspace{2ex}

\noindent {\bf The Automated vs. Interactive Tradeoff}. 
This section has presented eleven goal
simplification inference  rules that exploit the induction 
hypotheses and the various symbolic techniques explained in Section \ref{PRELIM-SECT}
to simplify goals as much as possible.  It has also presented
nine inductive inference rules to be used primarily in intercative
mode, although some could be automated.
There is in fact
a tradeoff between automation and interaction.  A key 
goal of the inference system is to allow ``seven league boots'' proof steps.
%
% But a user may sometimes
% find it  difficult to understand
% how formulas have been simplified.  This is an important pragmatic question 
%beyond  the scope of this paper.
% As more experience is gained
% with Maude's {\bf NuITP}, helpful explanations can be given to users 
% to increase  understandability.
%
The {\bf NuITP}  allows the user to fine
tune the tradeoff between automation and interaction
by choosing between  applying one inference rule at a time
or combining an inference rule step with subsequent 
goal simplification steps.  Furthermore, the  {\bf NuITP}
already provides basic support for user-defined \emph{proof strategies}
through a simple strategy language that, as more experience is gained,
will be extended in the future.
 Finally,  a \emph{fully
  automated} use of the {\bf NuITP}, in which the tool
with a chosen  proof strategy  is used as a backend
by other tools is illustrated by the {\bf DM-Check} tool   \cite{DM-Check,DM-Check-WRLA24}
and is envisioned for other tools.

\section{Other Inductive Proof Examples}
\label{inf-ex}

Six additional examples  further illustrate the use of the 
simplification and inductive rules and show  the advantages of the multiclause
representation.

\vspace{1ex} 

\noindent {\bf 1. Using Lemmas}.  We illustrate how
the $\textbf{LE}$ rule can help in 
proving a conjecture.  The conjecture in question 
is the original formula in Example \ref{EPS-example},
which applying the $\textbf{EPS}$ and $\textbf{ICC}$ simplification
rules in, respectively, Examples \ref{EPS-example}
and \ref{ICC-example}, became the considerably simpler goal:
\[
[\emptyset, \mathcal{N}\mathcal{P},\emptyset] \Vdash 
x* x=s(0),\; y =0 \rightarrow s(0) = x. 
\]
We can apply the $\textbf{CAS}$ rule to variable $x$
with generator set $\{0,s(0),s(s(z))\}$ to get the
following  three subgoals:
\[
[\emptyset, \mathcal{N}\mathcal{P},\emptyset] \Vdash 
0* 0=s(0),\; y =0 \rightarrow s(0) = 0 
\]

\[
[\emptyset, \mathcal{N}\mathcal{P},\emptyset] \Vdash 
s(0) * s(0) =s(0),\; y =0 \rightarrow s(0) = s(0)
\]

\[
[\emptyset, \mathcal{N}\mathcal{P},\emptyset] \Vdash 
s(s(z))* s(s(z)) =s(0),\; y =0 \rightarrow s(0) = s(s(z))
\]

\noindent The $\textbf{EPS}$ rule 
automatically discharges the first two.  To prove the third, we assume
as already proved the
commutativity of natural number addition in Peano notation
---a property that was proved in Example
\ref{+-comm-lemma} for the ``twice as fast''
definition of addition, and whose detailed proof for 
the standard definition of addition
we leave to the reader--- 
and use the  $\textbf{LE}$ rule to
get the goal:
\[
[\emptyset, \mathcal{N}\mathcal{P},\{n+m=m+n\}] \Vdash 
s(s(z))* s(s(z)) =s(0),\; y =0 \rightarrow s(0) = s(s(z)). 
\]
Thanks to the fact that the commutativity equation $n+m=m+n$
can be applied using ordered rewriting, $\textbf{EPS}$
simplification of this goal can use the lemma plus the equations for
* and + to  first transform the equation
$s(s(z))* s(s(z)) =s(0)$ into the equation
$s(s(s(s((s(s(z))*z)+z)))) =s(0)$, which is further
transformed to $\bot$ by means of the $\mathcal{N}\mathcal{P}^{=}$-rules
$s(n)=s(m) \rightarrow n = m$ and $s(n)=0 \rightarrow \bot$, thus
discharging the goal and finishing the proof.

An intriguing thought about the $n+m=m+n$ lemma
is that, instead of applying it by ordered rewriting,
it could have been ``internalized'' after the fact as a commutativity
\emph{axiom} which is
added to the theory $\mathcal{N}\mathcal{P}$.
More generally, any associativity and/or commutativity
properties that have been proved as lemmas for some defined
binary  function symbol $f$ could be so internalized,
provided the module's  oriented equations $\vec{\mathcal{E}}$ 
remain RPO-terminating. This internalization process is already
supported by the {\bf NuITP} prover for both orientable
equations and associativity or commutativity axioms proved as lemmas
\cite{NuITP} (for a broader view of internalization
and its semantic foundations see \cite{meseguer-WRLA24}).

\vspace{2ex}

\noindent {\bf 2. Multiplicative Cancellation}.
This example is borrowed from \cite{ind-ctxt-rew}.
We wish to prove the cancellation law for natural number multiplication

\begin{center}
$x * z' = y * z' \Rightarrow x = y$
\end{center}

\noindent where $z'$ ranges over non-zero natural number while $x$ and
$y$ range over natural numbers.
We specify natural number addition and multiplication as
associative-commutative operators, as well as the $>$ predicate, in a theory 
$\mathcal{N}$ having a subsort relation $\mathit{NzNat} <
\mathit{Nat}$ of non-zero numbers as subset of all naturals
(see Appendix \ref{natural-example} for a detailed specification of $\mathcal{N}$).
Of course, since the proof of the reverse implication
$x = y \Rightarrow x * z' = y * z'$ follows trivially
by simplification with the \textbf{(ICC)} rule, what we are really
proving is the equivalence $x * z' = y * z' \Leftrightarrow x = y$.
Therefore,  using  the 
 \textbf{(ERL)}
and \textbf{(ERR)} rules, once the above cancellation
rule has been proved, the rewrite rule
$x * z' = y * z' \rightarrow x = y$
can be added to
the equality predicate theory $\vec{\mathcal{N}}^=$
($\mathcal{N}$ has no $U$ axioms, i.e.,
$\vec{\mathcal{N}}_{U}^= = \vec{\mathcal{N}}^=$)
to obtain a more powerful version of the \textbf{(EPS)} simplification
rule.

\newcommand{\natgoali}[3]{\goali{\mathcal{N}}{#1}{#2}{#3}}
\newcommand{\natgoal}[1]{\natgoali{\emptyset}{#1}{\emptyset}}

\newcommand{\barx}{\ensuremath{\overline{x}_1}}
\newcommand{\natgoalimc}[1]{\natgoali{\barx * z' = y * z'\! \rightarrow \barx = y}{#1}{\{\barx\}}}

\noindent We begin with the goal:

\begin{center}
$G : \natgoal{x * z' = y * z' \rightarrow x = y}$
\end{center}

\noindent After applying the rule \GSI\ to the variable $x$ with the 
 $B_{0}$-generator
 set $\{0, 1 + x_1\}$
and simplifying by \EPS\ we obtain:

\begin{center}
$G_1 : \natgoal{0 = y * z' \rightarrow 0 = y}$\\[2pt]
$G_2 : \natgoalimc{z' + (\barx * z') = y * z' \rightarrow \barx + 1 = y}$
\end{center}

\noindent We first prove $G_1$ by: (a) applying the \CAS\ rule to
variable $y$ with the  $B_{0}$-generator
set $\{0, y'\}$, where $y'$ has the non-zero natural
 sort \textit{NzNat}; and (b) applying the \EPS\ rule to obtain: 

\begin{center}
$G_{1.1} : \natgoal{0 = 0 \rightarrow 0 = 0}$\\[2pt]
$G_{1.2} : \natgoal{0 = y' * z' \rightarrow 0 = y'}$
\end{center}

\noindent To solve $G_{1.1}$, apply \EPS\ to obtain $\top$.
To solve $G_{1.2}$, apply \VA\ to the term $y' * z'$ which has least sort \textit{NzNat}
to obtain:

\begin{center}
$G_{1.2.1} : \natgoal{0 = z'', z'' = y' * z' \rightarrow 0 = y'}$
\end{center}

\noindent where $z''$ also has sort \textit{NzNat}. Finally apply \CVUL\ to obtain $\top$,
since the equation $0 = z''$ has no unifiers. This finishes the proof of $G_1$.
We now prove $G_2$ by: (a) applying the \CAS\ rule to variable $y$
with  $B_{0}$-generator
set
$\{0, y_1 + 1\}$; and (b) applying the \EPS\ rule to obtain:

\begin{center}
$G_{2.1} : \natgoalimc{z' + (\barx * z') = 0 \rightarrow \barx + 1 = 0}$\\[4pt]
$G_{2.2} : \natgoalimc{\minibox[c]{$z' + (\barx * z') = (y_1 * z') + z'\!$\\ $\rightarrow \barx + 1 = y_1 + 1$}}$
\end{center}

\noindent To solve $G_{2.1}$, apply \VA\ to the term $z' + (\barx * z')$ which has least sort \textit{NzNat}
to obtain:

\begin{center}
$G_{2.1.1} : \natgoalimc{\minibox[c]{$z' + (\barx * z') = z'', z'' = 0$\\$\rightarrow \barx + 1 = 0$}}$
\end{center}

\noindent where $z''$ also has sort \textit{NzNat}. As in $G_{1.2.1}$, apply \CVUL\ to obtain $\top$.
Finally, to solve $G_{2.2}$, we apply \ERL\ and \ERR\ with the equivalence
$z_1 + z_2 = z_1 + z_3 \Leftrightarrow z_2 = z_3$
(which can be proved by variant satisfiability) to obtain:

\begin{center}
$G_{2.2.1} : \natgoalimc{\barx * z' = y_1 * z' \rightarrow \barx = y_1}$
\end{center}

\noindent But note that a proof of  $G_{2.2.1}$ immediately follows by \CS.
In summary, we completed the proof after 14 applications of our
inference rules.

\vspace{2ex}

\noindent {\bf 3. Proving Disequalities}.
Consider again the theory $\mathcal{N}$ used in the above
{\bf Multiplicative Cancellation} example and 
described in detail in Appendix \ref{natural-example}.  We now wish to
prove the inductive validity of the implication:

\begin{center}
$n > 1 = \mathit{true} \Rightarrow (n + n \not= n \; \wedge \; n * n 
\not= n)$.
\end{center}

\noindent That is, the inductive validity of the clauses:

\begin{center}
$n > 1 = \mathit{true} \, , \;  n + n = n   \rightarrow \bot$
\end{center}

\noindent and

\begin{center}
$n > 1 = \mathit{true}  \, ,  \;  n * n = n   \rightarrow \bot$. 
\end{center}

Since, as pointed out in Appendix \ref{natural-example}, the
subtheory $\mathcal{N}_{1}$ with constants $0$, $1$,
$\mathit{true}$ and $\mathit{false}$, $>$, and all typings for
operator $+$, and the equations
for $+$ and $>$ is FVP, and
has a subsignature of constructors with $0$, $1$ ,
$\mathit{true}$, $\mathit{false}$, and the smallest typing for the
$AC$ operator $+$,  the first clause can be automatically simplified 
to $\top$ in two different ways: either  by application of 
the $\textbf{VARSAT}$ rule, or ---since 
$n > 1 = \mathit{true}  \; \wedge \;  n + n = n$ has no 
variant unifiers--- by applying instead the \textbf{CVUL} rule.

\vspace{1ex}

\noindent To prove the second clause,
we can apply to the equation $n > 1 = \mathit{true}$ in its
condition the \textbf{CVUL} rule.
This equation has the single constructor variant unifier
$\{n \mapsto 1 + n'\}$, where $n'$ has sort $\mathit{NzNat}$,
so we get the goal:

\begin{center}
$(n' +1) * (n' + 1) = n' + 1   \rightarrow \bot$
\end{center}

\noindent which is simplified by the \textbf{EPS} rule to the goal:

\begin{center}
$n' + n' + (n'*n')  = n' \rightarrow \bot$
\end{center}

\noindent which applying the \textbf{VA} rule yields the goal:

\begin{center}
$n' + n' + m' = n'   \,  ,  \; m' = n'*n'   \rightarrow \bot$. 
\end{center}

\noindent with $m'$ of sort $\mathit{NzNat}$,
which automatically simplifies to $\top$ by applying the
\textbf{CVUL} rule to the equation $n' + n' + m' = n'$,
since that equation has no variant unifiers.

\vspace{2ex}

\noindent {\bf 4. Reversing Palindromes}.  This example extends
the theory $\mathcal{L}$ in
Example \ref{reverse-example} on reversing (non-empty) lists of
elements.  It comes with the added bonus of illustrating
the inductive congruence closure  (${\bf ICC}$) simplification rule.
We extend $\mathcal{L}$
by: (i) automatically extending it to $\mathcal{L}^{=}$,
where equationally-defined equality predicates have been 
added,\footnote{\label{eq-pred-foot} For the sake of
simplicity, we here use $\mathit{true}$ and $\mathit{false}$
instead of $\top$ and $\bot$ for the truth values
of  $\_=\_$,
i.e., we do not keep a renamed copy of the Booleans
with $\top$ and $\bot$ as its truth values.}
and (ii) adding also a palindrome Boolean predicate  on (non-empty)
lists $\mathit{pal}: \mathit{List} \rightarrow
\mathit{Bool}$, defined (with the same typing for variables as in
Example \ref{reverse-example})
 by the following equations oriented as 
rules:
$[3]: \mathit{pal}(x) = \mathit{true}$,
$[4]: \mathit{pal}(x\; \cdot  \; x) = \mathit{true}$,
$[5]: \mathit{pal}(x\; \cdot  \;  Q \; \cdot  \;x) = \mathit{pal}(Q)$,
$[6]: \mathit{pal}(x\; \cdot  \; y) = \mathit{false}\;\; \mathit{if}\;\; (x =
y) = \mathit{false}$,
$[7]: \mathit{pal}(x\;  \cdot \;  Q \;  \cdot  \;y) = 
\mathit{false}\;\; \mathit{if}\;\; (x = y) = \mathit{false}$.
The goal we want to prove is:
\[
\mathit{pal}(L) = \mathit{true} \rightarrow \mathit{rev}(L)=L.
\]
First, of all, using the ${\bf LE}$ rule, we can add to this goal  the
 lemma $\mathit{rev}(Q \; \cdot \; y)=y \;
\cdot \; \mathit{rev}(Q)$ already proved in Example
\ref{reverse-example},
which, with the order on symbols
 $(\mathit{rev}) \succ (\_\cdot \_)$, is reductive and can therefore
be listed as rule $[8]: \mathit{rev}(Q \; \cdot \; y) \rightarrow y \;
\cdot \; \mathit{rev}(Q)$.  Let us apply rule ${\bf NI}$, choosing 
$\mathit{pal}(L)$ as our focus narrowex.  Note that the narrowings
with rules $[6]$--$[7]$ generate $\mathit{true}=\mathit{false}$
the corresponding
  premises of their subgoals,
so those subgoals simplify to $\top$.  The subgoals obtained
by  narrowing with rules $[3]$ (resp. $[4]$) simplify to $\top$
thanks to rule $[1]$ (resp. rules $[1]$ and $[2]$) for $\mathit{rev}$.
The only non-trivial case is the subgoal obtained by narrowing 
with rule $[5]$ and  substitution
$\alpha = \{L \mapsto y\;  \cdot \;  P \;  \cdot \;y,\;
Q \mapsto P ,\;  x  \mapsto y \}$, yielding subgoal
\[
\mathit{pal}(\overline{P}) = \mathit{true} \rightarrow
\mathit{rev}(\overline{y}\;  \cdot \;  
\overline{P} \; \cdot  \;\overline{y} )= 
\overline{y}\;  \cdot \;  \overline{P} \; \cdot  \;\overline{y}.
\]
which, using the $\mathit{rev}$ equations and
Lemma $[8]$, simplifies to
\[
\mathit{pal}(\overline{P}) = \mathit{true} \rightarrow
\overline{y}\;  \cdot \;  \mathit{rev}(\overline{P}) \;  \cdot \;\overline{y}
= \overline{y}\;  \cdot \;  \overline{P} \;  \cdot  \;\overline{y}.
\]
and has induction hypothesis
\[
\mathit{pal}(\overline{P}) = \mathit{true} \rightarrow
\mathit{rev}(\overline{P})=
\overline{P}.
\]

\vspace{2ex}

\noindent We can  further simplify this remaining subgoal
by ${\bf ICC}$ simplification. 
The congruence closure
of $\mathit{pal}(\overline{P}) = \mathit{true}$
is just the rule $\mathit{pal}(\overline{P}) \rightarrow \mathit{true}$,
which cannot be further simplified. 
The induction hypothesis 
$\mathit{pal}(\overline{P}) = \mathit{true} \rightarrow
\mathit{rev}(\overline{P})=
\overline{P}$, being ground
and having $\mathit{rev}(\overline{P}) \succ
\overline{P}$, is orientable as a rewrite
rule $\mathit{pal}(\overline{P}) = \mathit{true} \rightarrow
(\mathit{rev}(\overline{P}) \rightarrow
\overline{P})$ in $\vec{H}^{\oplus}_{e_U}$.  Furthermore: (i) its lefthand side
does match the subterm $\mathit{rev}(\overline{P})$ in the
subgoal's conclusion; and (ii)  its condition
is satisfied with the congruence closure 
$\mathit{pal}(\overline{P}) \rightarrow \mathit{true}$.
Therefore, thanks to the ${\bf ICC}$ rule, the entire goal
simplifies to $\top$,
finishing the proof of our original goal by a single
application of ${\bf NI}$ followed by simplification.

\vspace{2ex}

\noindent {\bf 5. Simplifying Conjectures by Rewriting
with  Non-Horn Hypotheses}.
As explained in Section \ref{IndH-REW}, we are using a substantial set
of induction hypothesis as rewrite rules $\vec{H}^{+}_{e_U}$ to
simplify conjectures.  Some of the hypotheses so
used can  be non-Horn, that is, clauses $\Upsilon \rightarrow \Delta$
where $\Delta$ has more than one disjunct.  But we have not yet seen any
examples illustrating the power of  simplification with
non-Horn hypotheses.  A simple example
illustrating such power is provided by the well-known 
unsorted theory $\mathcal{N}_{>}$ defining
the order relation on the
Peano natural numbers, defined by the equations, oriented as rules,
$[1]: 0>n \rightarrow \mathit{false}$,
$[2]: s(n) >0 \rightarrow \mathit{true}$, and
$[3]: s(n) >s(m) \rightarrow n > m$ (which has the subcall $n>m$).
We want to use  $\mathcal{N}_{>}$
to prove the trichotomy
law:
\[
x > y  = \mathit{true} \; \vee \; x =y \; \vee \; y > x = \mathit{true}.
\]
We can apply the ${\bf NI}$ rule to the narrowex $x>y$.  Narrowing
with rule $[2]$, the leftmost equation in the disjunction becomes  $\mathit{true}  =
\mathit{true}$, which allows easily discharging the generated
subgoal by ${\bf EPS}$ simplification.  
Narrowing with rule $[1]$ and substitution
$\alpha_{1}=\{ x \mapsto 0,\; y\mapsto y',\; n\mapsto y'\}$,
we generate the subgoal  $\mathit{false}  =
\mathit{true}    \; \vee \; 0 =y' \; \vee \; y' > 0 = \mathit{true}$,
which by $\vec{\mathcal{E}}^{=}$-simplification becomes the subgoal
$0 =y' \; \vee \; y' > 0 = \mathit{true}$, which 
can be easily discharged by applying again the ${\bf NI}$ rule to the
narrowex $y'>0$, to which only rule $[2]$ can be applied,
making the second disjunct $\mathit{true}  =
\mathit{true}$, which allows easily discharging of the generated
subgoal by $\vec{\mathcal{E}}^{=}$-simplification.  Finally, narrowing with rule $[3]$
and substitution
$\alpha_{3}=\{ x \mapsto s(x'),\; y\mapsto s(y'),\; n \mapsto x',\;  m\mapsto y'\}$,
we get subgoal (3):
$\overline{x'} > \overline{y'}  = \mathit{true} \; \vee \;
s(\overline{x'}) = s(\overline{y'}) \; \vee \;  s(\overline{y'}) >
s(\overline{x'}) = \mathit{true}$, which 
$\vec{\mathcal{E}}^{=}_{\{\overline{x'}, \overline{y'} \}}$-simplifies to
$\overline{x'} > \overline{y'}  = \mathit{true} \; \vee \;
\overline{x'} = \overline{y'} \; \vee \;  \overline{y'} >
\overline{x'} = \mathit{true}$, and we also get
the induction hypothesis $H_{3}\equiv
\overline{x'} > \overline{y'}  = \mathit{true} \; \vee \;
\overline{x'} = \overline{y'} \; \vee \;  \overline{y'} >
\overline{x'} = \mathit{true}$, which in $\vec{H}^{\oplus}_{e_U}$
becomes the ground rewrite rule
$(\overline{x'} > \overline{y'}  = \mathit{true} \; \vee \;
\overline{x'} = \overline{y'} \; \vee \;  \overline{y'} >
\overline{x'} = \mathit{true}) \rightarrow \top$.  But then, the simplified
subgoal (3) is immediately discharged by ${\bf  EPS}$ simplification,
thanks to this rule in $\vec{H}^{+\oplus}_{e_U}$.

\vspace{2ex}

\noindent {\bf 6. Reasoning with Multiclauses}.  Up to now most 
examples have involved clauses.  The reader may reasonably wonder
whether the extra generality of supporting multiclauses is worth the
trouble.  The purpose of this simple example is to dispel any
such qualms: multiclauses
 can afford a substantial economy of thought and support
shorter proofs.  The example is unsorted and well known: 
a slight extension of the  theory $\mathcal{N}_{>}$ just
used for illustrating the simplification of conjectures with
non-Horn hypothesis.
 It comes with the added bonuses of illustrating
equationally-defined equality predicates (see Section \ref{EQ-PREDS})
and the inductive congruence closure  (${\bf ICC}$) simplification rule.

First of all, we extend the theory  $\mathcal{N}_{>}$
to its ---automatically generated--- protecting extension 
$\mathcal{N}_{>}^{=}$, which adds a commutative,  equationally defined
Boolean equality predicate $\_=\_$. In
 $\mathcal{N}_{>}^{=}$, the equality predicate  $\_=\_$
for sort the $\mathit{Nat}$ of naturals
is defined by three rules,\footnote{For the sake of
simplicity, we here use $\mathit{true}$ and $\mathit{false}$
instead of $\top$ and $\bot$ for the truth values
of  $\_=\_$ (see Footnote \ref{eq-pred-foot}).}
namely, $n = n \rightarrow \mathit{true}$,
$0 = s(n) \rightarrow \mathit{false}$,
and $s(n) = s(m) \rightarrow n = m$.  We then further extend
$\mathcal{N}_{>}^{=}$ by declaring the $\geq$ predicate, defined
by the single rule: $[4]: n \geq m = (n > m \;\; \mathit{or}\; \; n = m)$.
Two basic properties about $>$ and $\geq$ 
that we wish to prove as lemmas are:
\[
x>y= \mathit{true} \rightarrow s(x) > y = \mathit{true} 
\;\;\; \;\;\; \;\;\; \; \;\;\; \mathit{and} \;\;\; \;\; \; \;\;\; \; \;\;\;
x>y= \mathit{true} \rightarrow y \geq x = \mathit{false}.
\]
We can of course prove them as separate lemmas.  But we can 
bundle them together and prove instead the single multiclause:
\[
x>y= \mathit{true} \rightarrow (s(x) > y = \mathit{true} \;\wedge \;
 y \geq x = \mathit{false}).
\]
We can, for example, apply ${\bf NI}$ to the focus narrowex
$x>y$.  The subgoal obtained by narrowing with rule $[1]$ is discharged by
simplification, since we get $\mathit{false} = \mathit{true}$ in the
premise. Narrowing with rule $[2]$ and  substitution
$\alpha_{2} = \{x \mapsto s(n'),\; n \mapsto n',\;
y \mapsto 0\}$ yields the subgoal
$\mathit{true} = \mathit{true} \rightarrow
(s(s(n')) > 0 = \mathit{true} \; \wedge \; 0 \geq s(n') = \mathit{false}$,
which also simplifies to $\top$ and is likewise discharged.
The interesting goal is the one obtained by narrowing with
rule $[3]$ and substitution
$\alpha_{2} = \{x \mapsto s(n'),\; n \mapsto n',\;
y \mapsto s(m'),\; m \mapsto m'\}$, namely,
\[
\overline{n'} >\overline{m'} = \mathit{true} 
\rightarrow (s(s(\overline{n'})) > s(\overline{m'}) = \mathit{true} \;\wedge \; 
 s(\overline{m'}) \geq s(\overline{n'}) = \mathit{false}). 
\]
which simplifies to:
\[
\overline{n'} >\overline{m'} = \mathit{true} 
\rightarrow (s(\overline{n'}) > \overline{m'} = \mathit{true} \;\wedge \; 
 (\overline{m'} > \overline{n'} \; \mathit{or}\; \overline{m'} =
 \overline{n'})  = 
\mathit{false})
\]
and has the following two simplified induction hypotheses:
\[
\overline{n'} >\overline{m'} = \mathit{true} \rightarrow s(\overline{n'}) > \overline{m'} = \mathit{true} 
\;\;\; \;\;\; \;\;\; \; \;\;\; \mathit{and} \;\;\; \;\; \; \;\;\; \;
\;\;\;
\overline{n'} >\overline{m'} = \mathit{true} \rightarrow 
(\overline{m'} > \overline{n'} \; \mathit{or}\; \overline{m'} =
 \overline{n'})  = 
\mathit{false}
\]

\vspace{1ex}

\noindent The congruence closure of the goal's  premise is
quite immediate, namely, the rewrite rule 
$\overline{n'} >\overline{m'} \rightarrow \mathit{true}$, which cannot
be further simplified.
 Since the two hypotheses are reductive, they are
orientable as rewrite rules
$\overline{n'} >\overline{m'} = \mathit{true} \rightarrow (s(\overline{n'}) > \overline{m'} \rightarrow \mathit{true})$
and $
\overline{n'} >\overline{m'} = \mathit{true} \rightarrow
((\overline{m'} > \overline{n'} \; \mathit{or}\; \overline{m'} =
 \overline{n'})  \rightarrow 
\mathit{false})$, which
 belong to $\vec{H}_{e_U}$ and, a fortiori,
to $\vec{H}^{\oplus}_{e_U}$.  But since: (i)  their lefhand sides 
match respective  subterms of the subgoal's conclusion,
and (ii) their condition is satisfied by the
congruence closure $\overline{n'} >\overline{m'} \rightarrow \mathit{true}$,
the ${\bf ICC}$ rule simplifies the subgoal to $\top$,
thus finishing the (joint) proof
of both lemmas by a single application of ${\bf NI}$ to their
multiclause bundling, followed by simplification.

\vspace{1ex}

An even simpler and very common opportunity
of bundling clauses into multiclauses  arises when trying to prove
several conjectures that are themselves equations,
in which case a multiclause is just a conjunction of equations.
Let us focus for simplicity on two such equations,
$e_{1}(x,y)$ and $e_{2}(x',y')$, involving variables 
 $x,x'$ of sort $s_{1}$ and  $y,y'$ of sort $s_{2}$.
Since we know that
$T_{\mathcal{E}} \models e_{1}(x,y)$
iff  $T_{\mathcal{E}} \models e_{1}(x',y')$, then,
$T_{\mathcal{E}} \models e_{1}(x,y)$ and  $T_{\mathcal{E}} \models
    e_{2}(x',y')$ hold, i.e., both conjectures are valid, iff
$T_{\mathcal{E}} \models e_{1}(x',y') \; \wedge \;  e_{2}(x',y')$ does.
The moral of this little \emph{Gedankenexperiment} is that, to take full advantage 
of bundling several clauses into a multiclause,
we should first  \emph{rename some of their variables},
so that the different conjuncts \emph{share as many variables as
  possible}. In this way, we may achieve the proverbial ---yet,
\emph{not} politically correct--- objective of killing as many birds
as possible with a single stone.

\section{Related Work and Conclusions}

As already mentioned, this work combines features from
automated, e.g., 
\cite{DBLP:conf/popl/Musser80,ind,DBLP:journals/jcss/HuetH82,DBLP:journals/ai/KapurM87,DBLP:conf/lics/Bachmair88,DBLP:journals/jar/BouhoulaR95,DBLP:journals/iandc/ComonN00}
and explicit, e.g.,
\cite{DBLP:conf/rta/KapurZ89,DBLP:conf/alp/Goguen90,DBLP:series/mcs/GuttagHGJMW93,cafe-tools-paper,DBLP:conf/birthday/GAinALOF14,itp-manual,hendrix-thesis,itp/HendrixKM10,DBLP:journals/corr/abs-2101-02690} equational
inductive theorem proving, as well as some features from first-order
superposition theorem proving
\cite{DBLP:journals/logcom/BachmairG94,DBLP:journals/aicom/Schulz02,weidenbach-SPASS},
 in a novel way.  In the explicit induction area,
the well-known ACL2 prover \cite{acl2-book} should also be mentioned.
ACL2 does not directly support inductive reasoning about general algebraic
specifications.  It does instead support very powerful inductive reasoning
about LISP-style data structures.  One way to relate ACL2 to the
above-mentioned explicit induction equational inductive provers is to view it as
a very powerful domain-specific explicit induction equational theorem
prover for recursive functions defined over LISP-style 
data structures.\footnote{ACL2 is of course a \emph{general purpose}
inductive theorem prover.  The main difference with equational theorem
provers in the broader sense is that they support any user-defined algebraic signatures.
Therefore, they can  \emph{directly} represent any algebraic data types,
whereas in ACL2 such data types are represented \emph{indirectly}, by 
encoding them as LISP data structures.}

Some of the
automatable techniques presented here have been used in some fashion in earlier
work, but, to the best of my knowledge, others have not.  For example, 
congruence closure is used in many 
provers, but congruence 
closure modulo is considerably less used, 
and order-sorted congruence closure modulo
is here used for the first time. 
Contextual
rewriting goes back to the Boyer-More prover
\cite{boyer-moore80}, later extended to ACL2 \cite{acl2-book},
and has also been used, for example, in RRL
\cite{DBLP:journals/iandc/ComonN00} and in Spike 
\cite{DBLP:journals/jar/BouhoulaR95};
and clause subsumption and ordered rewriting are used in most
automated theorem provers,
including inductive ones.
Equational simplification is used 
by most provers, and ordered rewriting is used by most
first-order and inductive automatic provers
and by some explicit induction ones;
but to the best of my knowledge simplification with
equationally-defined equality predicates \emph{modulo} axioms $B_{0}$
was only previously used in \cite{rocha-meseguer-calco11},
although in the much easier free case
equality predicates have been used
to specify ``consistency'' properties of data types in, e.g.,
\cite{DBLP:journals/jcss/HuetH82,DBLP:journals/iandc/ComonN00}.
To the best of my knowledge, neither 
constructor variant unification nor
variant satisfiability have been used in other 
general-purpose provers,
although variant unification is used
in various cryptographic protocol verification tools, e.g.,
 \cite{maude-npa-tutorial,DBLP:conf/cav/MeierSCB13}.
Combining \emph{all} these techniques, and doing so
in the very general setting of conditional order-sorted equational theories
---which subsume unsorted and many-sorted ones as
special cases--- and modulo any associativity and/or commutativity
and/or identity axioms appears to be new.

The combination of features from
automated and explicit-induction theorem proving 
offers the short-term possibility of
an inference subsystem that can be automated as a practical
oracle for inductive validity of VCs generated by other tools.
This automation, by including most of the formula simplification rules, would
allow users to focus on applying just the 9 inductive inference rules.
Of these, the ${\bf NI}$ rule and (a special case of) the ${\bf VA}$ rule
offer the prospect
of being easily automatable,\footnote{The {\bf GSI} rule
could also be automated; but this will probably
 require more complex heuristics.} bringing
us closer to the goal of achieving a practical synthesis between
interactive and automated inductive theorem proving. 
As the experience of using a subset of the formula simplification
rules to discharge VCs generated by the reachability logic
theorem prover reported in \cite{ind-ctxt-rew}
as well as the more recent experience of using the {\bf NuITP}
as a backend to discharge inductive VCs in the {\bf DM-Check} tool
\cite{DM-Check,DM-Check-WRLA24} suggests, such a synthesis could provide an effective  way
for a wide variety of other tools to use an
inductive theorem prover as an automatic ``backend'' VC verifier.
\emph{Strategies} will play a key role in achieving this goal.
The {\bf NuITP} already provides some support for defining and
using strategies, but this is an area that should be further developed.

Another area that needs further development is that of \emph{proof
  certification}.  The current version of the  {\bf NuITP} supports the
saving of proof scripts and the display of proof trees in  \LaTeX{}
notation.  Full proof certification is possible, but it will require very substantial
efforts.  A key challenge is the large body
of symbolic algorithms involved that need to be
certified.  For example, only very recently has certification for unsorted
associative-commutative unification become possible after a very
large effort formalizing and verifying Stickel's algorithm in the PVS
prover \cite{DBLP:conf/fscd/Ayala-RinconFSS22}.  Several 
inference rules use either order-sorted $B$-unification for any combination
of associative and/or commutative and/or unit axioms, or
variant $E \cup B$-unification, for which no machine-assisted formalizations such as
that in \cite{DBLP:conf/fscd/Ayala-RinconFSS22} currently exist
to the best of my knowledge.  Some \emph{partial} certification
of $B$- and  $E \cup B$-unifiers is certainly possible and achievable
in the near future, namely, certification  that  a $B$- or $E \cup
B$-unifier is correct.  The challenging part is the certification  of
\emph{completeness}: that the set of unifiers provided by the
unification algorithm covers as instances all other unifiers for
the given unification problem.  In the near future, since a good
number of inference rules are based on rewriting modulo axioms $B$
in an order-sorted equational theory, the correctness of those
rewriting steps
as equality steps could be certified using the certification method
developed for that purpose in \cite{DBLP:conf/fm/RosuELM03}.
This would also allow the ``easy'' part of unifier certification.
Included also in the need for certification are  the requirements
made on the equational theory $\mathcal{E}$, such as ground
convergence and sufficient completeness.  Here partial
certification is already available thanks to other exisiting Maude formal
tools such as Maude's Church-Rosser Checker \cite{crc-alp},
Termination Tool \cite{MTT-ijcar08,DBLP:journals/entcs/DuranLM09}
and Sufficient Completeness Checker
\cite{hendrix-meseguer-ohsaki-ijcar06}.  But, again, full certification
would require substantial new efforts.  For example, computation
of critical pairs in \cite{crc-alp} requires $B$-unification, which
itself would have to be certified.

In summary, what this paper reports on is a novel combination of inductive
theorem proving techniques to prove properties of equational programs
under very general assumptions: the equational programs can use
conditional equations, can execute modulo structural axioms $B$ such
as associativity and/or communtativity and/or unit element axioms,
and can have types and subtypes.  The main goal is to combine as
much as possible features of automatic and interactive inductive
theorem proving to make proofs shorter, while still giving the user
complete freedom to guide the proof effort.  The experience already
gained with the  {\bf NuITP} is quite encouraging; but there is much
work ahead.  First, as mentioned in Footnote
\ref{non-free-ctors}, the requirement of constructors being
free modulo $B_{0}$ should be relaxed.  Second, strategies more powerful
than those currently supported by the {\bf NuITP} should be developed
and illustrated with examples.  Third, a richer collection
of examples and challenging case studies as well as libraries
of already verified equational programs should be developed.
Fourth, the use of the  {\bf NuITP} in automated mode as a backed
should be applied to a variety of other formal tools.
Fifth, as mentioned above, substantial work is needed
in proof certification, which realistically should begin with
various kinds of partial certification as steps towards full certification.
In these and other ways, additional evidence for the usefulness
of the current inference system will become available,
and useful extensions and improvements of the inference system itself
are likely to be found.  
%
% As with pudding, one can say ---and even with more reason---
%that, ``the proof of the prover is in the proving.''

\vspace{1.5 ex}

\noindent {\bf Acknowledgements}. I warmly  thank Stephen Skeirik
for his contributions to the conference paper \cite{ind-ctxt-rew},
which, as explained in Section \ref{INTRO-SECT}, has been
further developed in substantial ways in this paper.  Since, due to other
professional obligations, Dr. Skeirik was not able to participate in
these new developments, he has expressed his
agreement on my  being the sole author of the present
paper. I cordially thank the reviewers of this paper
as well as Francisco Dur\'{a}n, Santiago Escobar, Ugo Montanari and Julia Sapi\~{n}a for
their excellent suggestions for improvement, which have led to 
a better and clearer exposition.  Furthermore, the collaboration on
the {\bf NuITP} with Drs. Dur\'{a}n,  Escobar and Sapi\~{n}a has,
as already mentioned, both further advanced the inference system and
demonstrated its effectiveness on substantial examples.
This work has been partially supported by
NRL under contracts N00173-17-1-G002 and N0017323C2002.

\bibliographystyle{splncs04}
\bibliography{ref,tex}

\newpage

\appendix

\section{The Natural Numbers Theory
  $\mathcal{N}$} \label{natural-example}.

\noindent The natural number theory used in the multiplicative cancellation
example of Section \ref{inf-ex} is borrowed from  \cite{ind-ctxt-rew} and
has the following Maude specification:

\lstset{numbers=left, numberstyle=\tiny, basicstyle=\ttfamily, columns=fullflexible, keepspaces=true, emph={fmod,endfm,mod,endm,sort,subsort,op,var,eq,rl,protecting,is,ctor,assoc,comm,id,ditto,variant}, emphstyle=\bfseries}

\begin{figure}
\hspace{20pt}
\begin{minipage}{\textwidth}
\begin{lstlisting}
fmod NATURAL is protecting TRUTH-VALUE .
  sorts Zero NzNat Nat .
  subsorts Zero NzNat < Nat .

  op 0 : -> Zero  [ctor] .
  op 1 : -> NzNat [ctor] .
  op _+_ : NzNat NzNat -> NzNat [ctor assoc comm] .
  op _+_ : NzNat   Nat -> NzNat [     assoc comm] .
  op _+_ : Nat   NzNat -> NzNat [     assoc comm] .
  op _+_ : Nat   Nat   -> Nat   [     assoc comm] .
  op _*_ : NzNat NzNat -> NzNat [     assoc comm] .
  op _*_ : Nat   Nat   -> Nat   [     assoc comm] .
  op _>_ : Nat   Nat   -> Bool .

  vars X Y Z : Nat .  var X' : NzNat .

  eq X +  0      =  X [variant] .
  eq X *  0      =  0 .
  eq X *  1      =  X .
  eq X * (Y + Z) = (X * Y) + (X * Z) .
  eq X + X' > X = true [variant] .
  eq X > X + Y = false [variant] .
endfm
\end{lstlisting}
\end{minipage}
\caption{Natural Number Theory Specification.}
\label{natural-thy}
\end{figure}

\noindent  Note that we have a ``sandwich'' of theories
$\mathcal{N}_{\Omega} \subseteq \mathcal{N}_{1} \subseteq
\mathcal{N}$, where $\mathcal{N}_{\Omega}$ is given by the
constants \texttt{true} and
\texttt{false} in \texttt{TRUTH-VALUE} plus the
 operators
marked as \texttt{ctor}, including the first typing for $+$,
and that $\mathcal{N}_{1}$ is the FVP theory extending 
$\mathcal{N}_{\Omega}$ with the remaining typings for $+$, the $>$ predicate,
the equation for \texttt{0} as identity element for $+$, and the two
equations for $>$.

\section{Proof of the Soundness Theorem} \label{Proof-Sound-Theo}

We need to prove that, under the theorem's assumptions
on $\mathcal{E}$,
 if $[\overline{X},\mathcal{E},H] \Vdash \Gamma
\rightarrow \Lambda$ has a closed proof tree, then
$[\overline{X},\mathcal{E},H] \models \Gamma
\rightarrow \Lambda$.  We reason by contradiction,
and assume that such an implication does not hold.
This means that there is a goal
$[\overline{X},\mathcal{E},H] \Vdash \Gamma
\rightarrow \Lambda$ 
having a closed proof tree of smallest depth $d$  possible and
such that $[\overline{X},\mathcal{E},H] 
\not\models \Gamma
\rightarrow \Lambda$.
That is, any closed proof tree of any goal
having depth less than $d$ proves a goal that is valid in its
associated theory.  We then reach a contradiction by
considering the inference rule applied at the root of the
tree.  Before reasoning by cases considering each inference rule,
we prove three lemmas that will be useful in what follows.
The statement of the first lemma might be deceptive without
some explanation of its purpose.
One might easily assume that $(\Sigma,E)$ will be used
in practice as an order-sorted equational
theory $\mathcal{E}$ where we want to prove
inductive theorems about its initial model $T_{\mathcal{E}}$.  
This is \emph{a} possible
use of the lemma, but the intention is to use it in 
the following, more general sense.  Recall the following definition
of executable inductive hypotheses:
\[
\vec{H}^{+}_{e_U} =  \vec{H}_{e_U} \cup \vec{H}_{wu_U} \cup \vec{H}^{=}_{\vee,e}
\]
and note that $\vec{H}_{e_U} \cup \vec{H}_{wu_U}$ denotes a set of
$\Sigma(\overline{X})$-rewrite rules orienting conditional
equations, whereas $\vec{H}^{=}_{\vee,e}$
denotes a set of $\Sigma(\overline{X})^{=}$-rewrite rules
orienting non-Horn hypotheses.  Furthermore, let $H^{eq}_{ne}$
denote the subset of the set $H_{ne}$ of non-executable
hypotheses that are equations or conditional equations.
In our intended use,
 $(\Sigma,E)$  will actually stand for an equational theory of the
form: $(\Sigma(\overline{X}),E \cup B \cup H_{e} \cup H_{wu} \cup H^{eq}_{ne})$,
where $(\Sigma,E \cup B)$ is the original theory $\mathcal{E}$
on which we are doing inductive reasoning about theorems valid
 in its initial algebra  $T_{\mathcal{E}}$.

\begin{lemma}  \label{formula-E-equiv}
Let $(\Sigma,E)$ be an order-sorted equational
theory, and let $(\Sigma^{=},E)$ be the extension of $(\Sigma,E)$
where $\Sigma$ is extended to $\Sigma^{=}$ as
described in Section \ref{EQ-PREDS},\footnote{The equations defining
the equality predicates are not needed here: we only need the extend
signature to represent formulas as terms.}
so that QF $\Sigma$-formulas
are represented as terms of the new Boolean sort added to
the sorts of $\Sigma$.  Let $\varphi$ and $\psi$ be any two 
QF formulas such that $\varphi =_{E} \psi$.  Then, these formulas
are $E$-equivalent, i.e., for any $(\Sigma,E)$-algebra $A$
and any assignment $a \in [X \sra A]$, where $X$ contains the
variables of $\varphi$ and $\psi$, we have the equivalence:
\[A, a \models \varphi \;\;\; \Leftrightarrow \;\;\; A, a \models
  \psi.\]
That is, we have $E \models \varphi \Leftrightarrow \psi$.
\end{lemma}

\begin{proof}  Since the equality relation is 
reflexive and transitive, and so is logical equivalence,
it is enough to prove the lemma when the equality
$\varphi =_{E} \psi$ is obtained by a single step of $E$-equality.
That is, there is position $p$ in $\varphi$
such that $\varphi|_{p}$ is an equation $t = t'$,
and there is a \emph{term}  position $i.q$ in such an equation, 
$1 \leq i \leq 2$, such that, taking w.l.o.g. $i=1$,
there is an order-sorted substitution $\theta$
and a (possibly conditional) equation $(u = v\;\; \mathit{if} \;C) \in E$ (or $(v = u\;\; \mathit{if} \;C) \in E$)
such that $E \vdash C \theta$ and
 $\psi = \varphi[v \theta]_{p.1.q}$.
In other words, $\psi$ only differs from $\varphi$ in that
at position $p$ the equation now has the form
$t[v\theta]_{q} = t'$.  But, since $A \models E$, we must have
$t a = t[v\theta]_{q} a$, and therefore we also must have
$A, a \models t = t' \; \Leftrightarrow A, a 
\models t[v\theta]_{q} = t'$.
But, by the inductive definition
of the satisfaction relation $A, a \models \varphi$
in terms of the Boolean structure of $\varphi$, a simple in
induction on $|p|$, the length of $p$,
 forces
$A, a \models \varphi \; \Leftrightarrow \;\; A, a \models
  \psi$, as desired.  $\Box$
\end{proof}

\begin{lemma} \label{formula-E=-equiv}
  Let $[\overline{X},\mathcal{E},H]$ be an
inductive theory, and consider again the
signature extension $\Sigma(X) \subseteq \Sigma(X)^{=}$
allowing the representation of QF $\Sigma(X)$-formulas 
as terms of the new Boolean sort added to
the sorts of $\Sigma(X)$. Let $\varphi$ and $\psi$ be any two 
QF  $\Sigma(X)$-formulas such that $\varphi
\rightarrow^{*}_{\vec{\mathcal{E}}^{=}_{\overline{X}_{U}} \cup \vec{H}^{+}_{e_U}} \psi$. 
Then, these formulas
are $[\overline{X},\mathcal{E},H]$-equivalent,
i.e., for any $[\overline{X},\mathcal{E},H]$-model
$(T_{\mathcal{E}^{\square}},[\overline{\alpha}])$
and any constructor ground substitution
$\beta$ whose domain contains the variables of $\varphi$ and $\psi$,
we have the equivalence:
\[T_{\mathcal{E}^{\square}} \models 
\varphi^{\circ}
(\alpha \uplus \beta)
\;\;\; \Leftrightarrow \;\;\; 
T_{\mathcal{E}^{\square}} \models 
\psi^{\circ}
(\alpha \uplus \beta)
\]
That is, we have $[\overline{X},\mathcal{E},H] \models \varphi \Leftrightarrow \psi$.
\end{lemma}

\begin{proof} Since the rewrite relation 
$\rightarrow^{*}_{\vec{\mathcal{E}}^{=}_{\overline{X}_{U}} \cup \vec{H}^{+}_{e_U}}$ is 
reflexive and transitive, and so is logical equivalence,
it is enough to prove the lemma for a single rewrite step
$\varphi
\rightarrow_{\vec{\mathcal{E}}^{=}_{\overline{X}_{U}} \cup
  \vec{H}^{+}_{e_U}} \psi$.
We give separate proofs of the lemma
for the three possible cases in which
a rewrite $\varphi
\rightarrow_{\vec{\mathcal{E}}^{=}_{\overline{X}_{U}} \cup
  \vec{H}^{+}_{e_U}} \psi$ can happen.

\vspace{1ex}

\noindent {\bf Case} (1).  The rewrite is performed with a
rule in 
$\vec{\mathcal{E}}_{U} \cup  \vec{H}_{e_U} \cup \vec{H}_{wu_U}$. 
Then the result follows from 
$[\overline{X},\mathcal{E},H] \models
E \cup B \cup H_{e} \cup H_{wu}$
and Lemma
\ref{formula-E-equiv}
applied to the equational
theory $(\Sigma(\overline{X}),E \cup B \cup H_{e} \cup H_{wu})$.
This  finishes the proof for {\bf Case} (1).

\vspace{1ex}

\noindent {\bf Case} (2). The rewrite is the
the application of a rule in 
$\vec{\mathcal{E}}^{=}_{\overline{X}_{U}}$
not in $\vec{\mathcal{E}}_{U}$; or
(3) it is the application of a rule
in $\vec{H}^{=}_{\vee,e}$.
In case (2), all such rules are of the form
$(u = v) \rightarrow \phi \;\; \mathit{if} \; C$
with $\phi$ a $\Sigma$-formula.
For example, if $[\_,\_]$ is a pairing constructor
in $\Sigma$ satisfying no axioms, then
there is a rule in $\vec{\mathcal{E}}^{=}_{\overline{X}_{U}}$
of the form $([x,y] = [x'=y']) \rightarrow
x = x' \wedge y = y'$.  Therefore,
$\varphi
\rightarrow_{\vec{\mathcal{E}}^{=}_{\overline{X}_{U}}} \psi$
exactly means that there is position $p$ in $\varphi$
such that $\varphi|_{p}$ is an equation $t = t'$,
and there is an order-sorted substitution $\theta$ and
a rule $(u = v) \rightarrow \phi \;\; \mathit{if} \; C$
in $\vec{\mathcal{E}}^{=}_{\overline{X}_{U}}$
such that: (a) $(u = v)\theta =_{B_{0}} (t = t')$,
(b) $\mathcal{E}^{=}_{\overline{X}_{U}} \vdash C \theta$,
and (c) $\psi = \varphi[\phi\theta]_{p}$.  We now have to prove
that for any
any $[\overline{X},\mathcal{E},H]$-model
$(T_{\mathcal{E}^{\square}},[\overline{\alpha}])$
and any constructor ground substitution
$\beta$ whose domain contains the variables of 
$\varphi$ and $\varphi[\phi\theta]_{p}$,
we have the equivalence:
\[T_{\mathcal{E}^{\square}} \models 
\varphi^{\circ}
(\alpha \uplus \beta)
\;\;\; \Leftrightarrow \;\;\; 
T_{\mathcal{E}^{\square}} \models 
\varphi[\phi\theta]_{p}^{\circ}
(\alpha \uplus \beta).
\]
But note that
if we have a rewrite 
$\varphi
\rightarrow_{\vec{\mathcal{E}}^{=}_{\overline{X}_{U}}}
\varphi[\phi\theta]_{p}$ with substitution $\theta$, we also have
a rewrite
$\varphi^{\circ}
\rightarrow_{{\vec{\mathcal{E}}_{U}}^{=}}
\varphi[\phi \theta]_{p}^{\circ}$ with 
substitution $\theta^{\circ}$,
so that $\varphi[\phi \theta]_{p}^{\circ}
= \varphi^{\circ}[\phi \theta^{\circ}]_{p}$.
But a simple induction on $|p|$ using the inductive definition
of the satisfaction relation $T_{\mathcal{E}^{\square}} \models 
\varphi^{\circ}$
in terms of the Boolean structure of $\varphi$, together with
the fact that for this rewrite at position $p$ to happen
$\varphi|_{p} \equiv t = t'$ must be a $\Sigma(X)$-equation,
forcing $T_{\mathcal{E}^{\square}} \models (t = t')^{\circ}$
iff $T_{\mathcal{E}} \models (t = t')^{\circ}$
iff (by the properties of $\mathcal{E}_{U}^{=}$)
 $T_{\mathcal{E}} \models \phi \theta^{\circ}$
iff  $T_{\mathcal{E}^{\square}} \models \phi \theta^{\circ}$,
 gives us 
$T_{\mathcal{E}^{\square}} \models 
\varphi^{\circ}
(\alpha \uplus \beta)
\; \Leftrightarrow \;
T_{\mathcal{E}^{\square}} \models 
\varphi[\phi\theta]_{p}^{\circ}
(\alpha \uplus \beta)$,
as desired.  This finishes the proof of  {\bf Case} (2).

\vspace{1ex}

\noindent  {\bf Case} (3).   There is a non-Horn clause
$\Upsilon \rightarrow \Delta$ in $H^{=}_{e}$,
oriented as a rewrite rule
$\Upsilon \rightarrow (\Delta \rightarrow \top)$,
 a position $p$ in $\varphi$, and a substitution
$\theta$ such that: (i) $\varphi|_{p} =_{B^{=}_{0}} \Delta \theta$
(where $B^{=}_{0}$ are the axioms in $\vec{\mathcal{E}}^=_{{\overline{X}}_U}$)
and (ii) $\top \in \Upsilon \alpha !_{\vec{\mathcal{E}}^=_{{\overline{X}}_U} \cup\ \vec{H}_{e_U}}$.
But using the (reflexive transitive closure of) the
\emph{already proved} cases (1)--(2) above,
$\top \in \Upsilon \theta !_{\vec{\mathcal{E}}^=_{{\overline{X}}_U}
  \cup\ \vec{H}_{e_U}}$
implies that for any $[\overline{X},\mathcal{E},H]$-model
$(T_{\mathcal{E}^{\square}},[\overline{\alpha}])$
and any constructor ground substitution
$\beta$ whose domain contains the variables of $\varphi$ and $\psi$,
we have the equivalence:
\[T_{\mathcal{E}^{\square}} \models 
(\Upsilon \theta)^{\circ}
(\alpha \uplus \beta)
\;\;\; \Leftrightarrow \;\;\; 
T_{\mathcal{E}^{\square}} \models 
\top
(\alpha \uplus \beta)
\]
that is, we have $T_{\mathcal{E}^{\square}} \models 
(\Upsilon \theta)^{\circ}
(\alpha \uplus \beta)$  But we also
have $[\overline{X},\mathcal{E},H] \models (\Upsilon \rightarrow \Delta)\theta$,
which forces
$T_{\mathcal{E}^{\square}} \models 
(\Delta \theta)^{\circ}
(\alpha \uplus \beta)$, and therefore the
equivalence
\[T_{\mathcal{E}^{\square}} \models 
(\Delta \theta)^{\circ}
(\alpha \uplus \beta)
\;\;\; \Leftrightarrow \;\;\; 
T_{\mathcal{E}^{\square}} \models 
\top
(\alpha \uplus \beta)
\]
which by the Tarskian semantics for QF formulas forces the
equivalence
\[T_{\mathcal{E}^{\square}} \models 
(\varphi[\Delta \theta]_{p})^{\circ}
(\alpha \uplus \beta)
\;\;\; \Leftrightarrow \;\;\; 
T_{\mathcal{E}^{\square}} \models 
(\varphi[\top]_{p})^{\circ}
(\alpha \uplus \beta)
\]
as desired.  This finishes the proof of  {\bf Case} (3),
and therefore that of the Lemma. $\Box$
\end{proof}

\noindent  Call two inductive theories
$[\overline{X},\mathcal{E},H]$ and $[\overline{X},\mathcal{E},H']$
\emph{semantically equivalent}, denoted
$[\overline{X},\mathcal{E},H]\equiv [\overline{X},\mathcal{E},H]$,
iff they have the same models.  The following Lemma
gives a useful sufficient condition for semantic equivalence.

\begin{lemma} \label{ind-th-equiv}
  Let $[\overline{X},\mathcal{E},H]$ be an
inductive theory, and $\overline{G},\overline{G'}$ be
two conjunctions of ground $\Sigma(\overline{X})$-equations.
Then , $[\overline{X},\mathcal{E},H] \models \overline{G}
\Leftrightarrow \overline{G'}$ implies
$[\overline{X},\mathcal{E},H \cup \{\overline{G}\}]\equiv 
[\overline{X},\mathcal{E},H \cup \{\overline{G'}\}]$.
\end{lemma}

\begin{proof}  We prove the $(\Rightarrow)$ implication
of the semantic equivalence $\equiv$.  The $(\Leftarrow)$
implication is entirely symmetric, changing the
roles of $\overline{G}$ and $\overline{G'}$.
For any model $(T_{\mathcal{E}^{\square}},[\overline{\alpha}])$
of $[\overline{X},\mathcal{E},H \cup \{\overline{G}\}]$
we of course have
$(T_{\mathcal{E}^{\square}},[\overline{\alpha}])\models
\overline{G}$.  But then,
$[\overline{X},\mathcal{E},H] \models \overline{G}
\Leftrightarrow \overline{G'}$
forces $(T_{\mathcal{E}^{\square}},[\overline{\alpha}])\models
\overline{G'}$, which in turn forces
$(T_{\mathcal{E}^{\square}},[\overline{\alpha}])$ to be a
model of $[\overline{X},\mathcal{E},H \cup \{\overline{G'}\}]$,
as desired. $\Box$
\end{proof}

\noindent We now resume our proof of the Soundness Theorem.
The cases are as follows:

\vspace{2ex}

\noindent {\bf EPS}.  By assumption we
have $[\overline{X},\mathcal{E},H] 
\not\models \Gamma
\rightarrow \Lambda$, but
$[\overline{X},\mathcal{E},H] 
\models (\Gamma \rightarrow
\Lambda)!_{\vec{\mathcal{E}}^{=}_{\overline{X}_{U}} \cup
  \vec{H}^{+}_{e_{U}}}$.
But by Lemma \ref{formula-E=-equiv}
we must have
$[\overline{X},\mathcal{E},H] 
\models \Gamma
\rightarrow \Lambda$
iff $[\overline{X},\mathcal{E},H] 
\models (\Gamma \rightarrow
\Lambda)!_{\vec{\mathcal{E}}^{=}_{\overline{X}_{U}} \cup
  \vec{H}^{+}_{e_{U}}}$, contradicting our original assumption.

\vspace{2ex}

\noindent {\bf CVUL}.  By the minimality assumption
we have: (i) $[\overline{X},\mathcal{E},H] 
\not\models \Gamma,\Gamma'
\rightarrow \Lambda$, where the
$\Gamma'$ are $\mathcal{E}_{1}$-equalities;
and (ii) $\{[\overline{X}\uplus
    \overline{Y}_{\alpha},\; \mathcal{E},\;
(H \cup
\widetilde{\overline{\alpha}}|_{\overline{X}_{\Gamma}})_{\mathit{simp}}
] \models  (\Gamma'\rightarrow\Lambda)
\overline{\alpha}\}_{\alpha\in\textit{Unif}^{\hspace{1pt}\Omega}_{\mathcal{E}_1}\!(\Gamma^{\circ})}$.
But (i) exactly means that  there is a ground constructor
substitution $\gamma$ with domain $X$
such that  $(T_{\mathcal{E}^{\square}},[\overline{\gamma}])$
is a model of $[\overline{X},\mathcal{E},H]$, 
which implies (a)
$(T_{\mathcal{E}^{\square}},[\overline{\gamma}]) \models 
H$; and that  for $Z = \mathit{vars}(\Gamma,\Gamma' \rightarrow
\Lambda)$ there is a ground constructor substitution
$\beta$ with domain $Z$ disjoint from $X$
such that (b)
 $T_{\mathcal{E}^{\square}} \not\models 
(\Gamma,\Gamma' \rightarrow 
\Lambda)^{\circ}
(\gamma \uplus \beta)$, which means that
(b).1 $T_{\mathcal{E}^{\square}} \models 
(\Gamma',\Gamma'')^{\circ}
(\gamma \uplus \beta)$, and (b).2
$T_{\mathcal{E}^{\square}} \not\models 
\Lambda^{\circ}
(\gamma \uplus \beta)$.
Let $X_{\Gamma} \uplus Z_{\Gamma}
 = \mathit{vars}(\Gamma^{\circ})$,
with  $X_{\Gamma} \subseteq X$, and $Z_{\Gamma} \subseteq Z$. But
(b).1 implies that $(\gamma \uplus \beta)|_{X_{\Gamma} \uplus Z_{\Gamma}}$
is a ground constructor unifier of $\Gamma^{\circ}$.  Therefore,
there is an idempotent variant
 constructor $\mathcal{E}_{1}$-unifier $\alpha$
of $\Gamma^{\circ}$ with domain $X_{\Gamma} \uplus Z_{\Gamma}$ and 
fresh range $Y$ (so that $Y \supseteq Y_{\alpha}$),
and a ground constructor substitution $\tau$ with domain $Y$ such that
$(\gamma \uplus \beta)|_{X_{\Gamma} \uplus Z_{\Gamma}} =_{B_{1}} \alpha \tau$.
This means that  $(T_{\mathcal{E}^{\square}},[\overline{\gamma}
\uplus \overline{\tau}|_{Y_{\alpha}}])$
 is a model of the inductive theory
\[
[\overline{X}\uplus
    \overline{Y}_{\alpha},\; \mathcal{E},\;
(H \cup
\widetilde{\overline{\alpha}}|_{\overline{X}_{\Gamma}})_{\mathit{simp}}
]
\]
since: (i) $(T_{\mathcal{E}^{\square}},[\overline{\gamma}
\uplus \overline{\tau}|_{Y_{\alpha}}]) \models H$
because $(T_{\mathcal{E}^{\square}},[\overline{\gamma}]) \models 
H$ (the variables $Y_{\alpha}$ are fresh and therefore
do not appear in $H$); and
(ii) $(T_{\mathcal{E}^{\square}},[\overline{\gamma}
\uplus \overline{\tau}|_{Y_{\alpha}}]) 
\models \widetilde{\overline{\alpha}}|_{\overline{X}_{\Gamma}}$,
since for each $\overline{x} \in \overline{X}$
we have $(\gamma \uplus \tau|_{Y_{\alpha}})(x)=\gamma(x)
=(\gamma \uplus \beta)|_{X_{\Gamma} \uplus Z_{\Gamma}}) (x) 
=_{B_{1}} (\alpha \; \tau)(x) = \alpha(x) \; \tau|_{Y_{\alpha}}
=\alpha(x) (\gamma \uplus \tau|_{Y_{\alpha}})$.  Therefore, 
by (ii) we have
$(T_{\mathcal{E}^{\square}},[\overline{\gamma}
\uplus \overline{\tau}|_{Y_{\alpha}}]) \models  
(\Gamma'\rightarrow\Lambda) \overline{\alpha}$.
But since $\mathit{vars}((\Gamma'\rightarrow\Lambda)
\overline{\alpha})=Z\setminus Z_{\Gamma} \uplus Y\setminus Y_{\alpha}$,
in particular, for the ground constructor substitution
$\beta|_{Z\setminus Z_{\Gamma}} \uplus \tau|_{Y\setminus Y_{\alpha}}$
 we must have
$T_{\mathcal{E}^{\square}}\models
(\Gamma'\rightarrow\Lambda)^{\circ} \alpha
(\gamma \uplus \tau|_{Y_{\alpha}}\uplus
\beta|_{Z\setminus Z_{\Gamma}} \uplus
\tau|_{Y\setminus Y_{\alpha}})$, that is,
$T_{\mathcal{E}^{\square}}\models
(\Gamma'\rightarrow\Lambda)^{\circ} \alpha
(\gamma \uplus \tau \uplus \beta|_{Z\setminus Z_{\Gamma}})$.
But since $(\gamma \uplus \beta)|_{X_{\Gamma} \uplus Z_{\Gamma}}
=_{B_{1}} \alpha \tau$, this forces
$T_{\mathcal{E}^{\square}}\models
(\Gamma'\rightarrow\Lambda)^{\circ}(\gamma \uplus \beta)$,
which by (b).1 forces
$T_{\mathcal{E}^{\square}}\models
\Lambda^{\circ}(\gamma \uplus \beta)$,
contradicting (b).2, as desired.

\vspace{2ex}

\noindent {\bf CVUFR}. By the minimality assumption
 we have (i).1 $[\overline{X},\mathcal{E},H] 
\not\models \Gamma \rightarrow \Lambda \wedge (u=v,\Delta)$, 
(i).2 $\textit{Unif}^{\hspace{2pt}\Omega}_{\mathcal{E}_1}((u=v)^{\circ})=\emptyset$, and
(ii) $[\overline{X},\mathcal{E},H] 
\models \Gamma \rightarrow \Lambda \wedge \Delta$.
But (i).1 means 
that, for $Y = \mathit{vars}(\Gamma \rightarrow
\Lambda \wedge (u=v,\Delta))$, we have constructor
ground substitutions $\alpha$ and $\beta$ with respective domains
$X$ and $Y$ such that: (a) $(T_{\mathcal{E}^{\square}},[\overline{\alpha}])$
is a model of $[\overline{X},\mathcal{E},H]$, and (b)
 $T_{\mathcal{E}^{\square}} \not\models 
(\Gamma \rightarrow 
\Lambda \wedge (u=v,\Delta))^{\circ}
(\alpha \uplus \beta)$.
That is,
(b).1 $T_{\mathcal{E}^{\square}} \models 
\Gamma^{\circ}
(\alpha \uplus \beta)$, and (b).2
$T_{\mathcal{E}^{\square}} \not\models (\Lambda \wedge (u=v,\Delta))^{\circ}
(\alpha \uplus \beta)$.  But (b).2 is equivalent to:
$T_{\mathcal{E}^{\square}} \models \neg(\Lambda)^{\circ}
(\alpha \uplus \beta)$
or
($T_{\mathcal{E}^{\square}} \models (u \not=v)^{\circ}
(\alpha \uplus \beta)$
and $T_{\mathcal{E}^{\square}} \models (\neg \Delta)^{\circ}
(\alpha \uplus \beta)$).
And (i).2 means that $T_{\mathcal{E}^{\square}} \models (u\not=v)^{\circ}
(\alpha \uplus \beta)$ is necessarily true.
Therefore, (b).2 is equivalent to:
$T_{\mathcal{E}^{\square}} \not\models (\Lambda \wedge \Delta)^{\circ}
(\alpha \uplus \beta)$, which, together with (b).1
and (a), contradicts (ii), as desired.

\vspace{2ex}

\noindent {\bf SUBL}.  We prove soundness for the two different cases
of the rule.

\vspace{1ex}

\noindent {\bf Case} $x$ is a variable.
By the minimality assumption 
we have (i) $[\overline{X},\mathcal{E},H] 
\not\models x=u,\;\Gamma
\rightarrow \Lambda$ with 
with $x$ a variable of sort $s$, $\mathit{ls}(u) \leq s$,
and $x$ not appearing in $u$,
and (ii) $[\overline{X},\mathcal{E},H] 
\models (\Gamma
\rightarrow \Lambda) \{x \mapsto u\}$.
But (i)  means 
that, for $Y = \mathit{vars}(x=u,\;\Gamma \rightarrow
\Lambda)$, we have constructor
ground substitutions $\alpha$ and $\beta$ with respective domains
$X$ and $Y$ such that: (a) $(T_{\mathcal{E}^{\square}},[\overline{\alpha}])$
is a model of $[\overline{X},\mathcal{E},H]$, and (b)
 $T_{\mathcal{E}^{\square}} \not\models 
(x=u,\Gamma \rightarrow 
\Lambda)^{\circ}
(\alpha \uplus \beta)$.
That is,
(b).1 $T_{\mathcal{E}^{\square}} \models 
(x=u,\Gamma)^{\circ}
(\alpha \uplus \beta)$, and
(b).2 $T_{\mathcal{E}^{\square}} \not\models 
\Lambda^{\circ}
(\alpha \uplus \beta)$.
But (b).1 and $x$ not appearing in $u$
imply that $\beta(x) =_{\mathcal{E}} u^{\circ} (\alpha \uplus\beta|_{Y\setminus \{x\}})$.
Therefore, $(\alpha \uplus \beta) =_{\mathcal{E}} 
(\alpha \uplus \{ x \mapsto u^{\circ} (\alpha \uplus\beta|_{Y\setminus \{x\}})\} 
\uplus \beta|_{Y\setminus \{x\}})$.
But (ii) and (a) imply that
$T_{\mathcal{E}^{\square}} \models 
\Gamma^{\circ}\{ x \mapsto u^{\circ}\} 
(\alpha \uplus \beta|_{Y\setminus \{x\}})$, that is,
$T_{\mathcal{E}^{\square}} \models 
\Gamma^{\circ} (\alpha \uplus \{ x \mapsto u^{\circ} 
(\alpha \uplus\beta|_{Y\setminus \{x\}})\} 
\uplus \beta|_{Y\setminus \{x\}})$,
which by
the semantic equivalence 
$(\alpha \uplus \beta) =_{\mathcal{E}} 
(\alpha \uplus \{ x \mapsto u^{\circ} (\alpha \uplus\beta|_{Y\setminus \{x\}})\} 
\uplus \beta|_{Y\setminus \{x\}})$
 contradicts (b), as desired.

\vspace{2ex}

\noindent {\bf Case} $\overline{x} \in \overline{X}$.
By the minimality assumption 
we have (i) $[\overline{X},\mathcal{E},H] \not\models
\overline{x}=u,\; \Gamma\rightarrow\Lambda$,
with $\overline{x}$ a fresh constant of sort $s$, $\mathit{ls}(u) \leq s$,
and $\overline{x}$ not appearing in $u$,
and (ii) $[\overline{X} \uplus
  \overline{Y}_{u} ,\mathcal{E},(H \cup \{\overline{x} = \overline{u}\})_{\mathit{simp}}] 
\models (\Gamma\rightarrow\Lambda) (\{y \mapsto \overline{y}\}_{y \in  Y_{u}} \uplus
\{\overline{x}\mapsto \overline{u}\})$.
But (i)  means 
that, for $Y = \mathit{vars}(\overline{x}=u,\;\Gamma \rightarrow
\Lambda)$, we have constructor
ground substitutions $\alpha$ and $\beta$ with respective domains
$X$ and $Y$ such that: (a) $(T_{\mathcal{E}^{\square}},[\overline{\alpha}])$
is a model of $[\overline{X},\mathcal{E},H]$, and (b)
 $T_{\mathcal{E}^{\square}} \not\models 
(\overline{x}=u,\Gamma \rightarrow 
\Lambda)^{\circ}
(\alpha \uplus \beta)$.
That is,
(b).1 $T_{\mathcal{E}^{\square}} \models 
(\overline{x}=u,\Gamma)^{\circ}
(\alpha \uplus \beta)$, and
(b).2 $T_{\mathcal{E}^{\square}} \not\models 
\Lambda^{\circ}
(\alpha \uplus \beta)$.
But (b).1 and $\overline{x}$ not appearing in $u$
imply that $\alpha(x) =_{\mathcal{E}} 
u^{\circ} (\alpha|_{X\setminus\{x\}} \uplus\beta|_{Y_{u}})$.
Therefore, $(\alpha \uplus \beta) =_{\mathcal{E}} 
(\alpha|_{X\setminus \{x\}} \uplus \{ x \mapsto u^{\circ} (\alpha|_{X\setminus \{x\}} \uplus\beta|_{Y_{u}})\} 
\uplus \beta)$.
Now note that 
$\mathit{vars}( (\Gamma\rightarrow\Lambda) (\{y \mapsto \overline{y}\}_{y \in  Y_{u}} \uplus
\{\overline{x}\mapsto \overline{u}\}))=Y\setminus Y_{u}$,
and that $(T_{\mathcal{E}^{\square}},[\overline{\alpha} \uplus 
\overline{\beta}|_{Y_{u}}])$ is a model of
$[\overline{X} \uplus
  \overline{Y}_{u} ,\mathcal{E},(H \cup \{\overline{x} =
  \overline{u}\})_{\mathit{simp}}]$, because
it satisfies $H$ since $(T_{\mathcal{E}^{\square}},[\overline{\alpha}])$
does, and it satisfies $\overline{x} =
  \overline{u}$ by $\alpha(x) =_{\mathcal{E}} u^{\circ} (\alpha|_{X\setminus
  \{x\}} \uplus\beta|_{Y_{u}})$.
Therefore,
$(T_{\mathcal{E}^{\square}},[\overline{\alpha}\uplus\overline{\beta}|_{Y_{u}}]),  
[\beta|_{Y\setminus Y_{u}}] \models  (\Gamma\rightarrow\Lambda) (\{y \mapsto \overline{y}\}_{y \in  Y_{u}} \uplus
\{\overline{x}\mapsto \overline{u}\})$, that is,
$T_{\mathcal{E}^{\square}} \models  ((\Gamma\rightarrow\Lambda) (\{y \mapsto \overline{y}\}_{y \in  Y_{u}} \uplus
\{\overline{x}\mapsto \overline{u}\}))^{\circ}(\alpha \uplus \beta)$,
which means,
$T_{\mathcal{E}^{\square}} \models  (\Gamma\rightarrow\Lambda)^{\circ}
\{ x\mapsto u^{\circ}\} (\alpha \uplus \beta)$, which
by the semantic equivalence $(\alpha \uplus \beta) =_{\mathcal{E}} 
(\alpha|_{X\setminus \{x\}} \uplus \{ x \mapsto u^{\circ} 
(\alpha|_{X\setminus \{x\}} \uplus\beta|_{Y_{u}})\} \uplus \beta)$
forces
$T_{\mathcal{E}^{\square}} \models  ((\Gamma\rightarrow\Lambda)^{\circ}
(\alpha \uplus \beta)$, which by 
(b).1 forces $T_{\mathcal{E}^{\square}} \models 
\Lambda^{\circ}
(\alpha \uplus \beta)$, contradicting (b).2, as desired.

\vspace{2ex}

\noindent {\bf SUBR}. We prove soundness for the
somewhat more involved case $\overline{x}\in \overline{X}$
and leave the simper and similar
case when $x$ is a variable for the reader.

\vspace{1ex}

\noindent {\bf Case} $\overline{x}\in \overline{X}$.  By the
minimality assumption we have:
(i) $[\overline{X},\mathcal{E},H] \not\models
  \Gamma\rightarrow\Lambda\wedge 
\overline{x}=u$,  where $\Lambda \not= \top$,
$\overline{x}$ has sort $s$, $ls(u)\leq s$, and
$\overline{x}$ does not appear in $u$;
(ii).1 $[\overline{X},\mathcal{E},H] \models
 \Gamma\rightarrow \overline{x}=u$; and
(ii).2 $[\overline{X} \uplus
  \overline{Y}_{u},\mathcal{E},
(H \cup \{ \overline{x} =  \overline{u}\}) _{\mathit{simp}}] \models
  (\Gamma\rightarrow\Lambda)(\{y \mapsto \overline{y}\}_{y \in Y_{u}} \uplus
\{ \overline{x}\mapsto \overline{u}\})$.
But (i)  means 
that, for $Y = \mathit{vars}(\Gamma\rightarrow\Lambda\wedge 
\overline{x}=u)$, we have constructor
ground substitutions $\alpha$ and $\beta$ with respective domains
$X$ and $Y$ such that: (a) $(T_{\mathcal{E}^{\square}},[\overline{\alpha}])$
is a model of $[\overline{X},\mathcal{E},H]$, and (b)
 $T_{\mathcal{E}^{\square}} \not\models 
(\Gamma\rightarrow\Lambda\wedge 
\overline{x}=u)^{\circ}
(\alpha \uplus \beta)$.
That is,
(b).1 $T_{\mathcal{E}^{\square}} \models 
(\Gamma)^{\circ}
(\alpha \uplus \beta)$, and
(b).2 $T_{\mathcal{E}^{\square}} \not\models 
(\Lambda^{\circ}\wedge 
x=u^{\circ})
(\alpha \uplus \beta)$, which is equivalent
to ($T_{\mathcal{E}^{\square}} \not\models 
\Lambda^{\circ} (\alpha \uplus \beta)$
or
$T_{\mathcal{E}^{\square}} \not\models  x=u^{\circ}
(\alpha \uplus \beta)$).
But (b).1 and (ii).1 force $T_{\mathcal{E}^{\square}} \models  x=u^{\circ}
(\alpha \uplus \beta)$, making (b).2 equivalent to
(b).$2 '$ $T_{\mathcal{E}^{\square}} \not\models 
\Lambda^{\circ} (\alpha \uplus \beta)$.
And $\overline{x}$ not appearing in $u$
imply that $\alpha(x) =_{\mathcal{E}} 
u^{\circ} (\alpha|_{X\setminus\{x\}} \uplus\beta|_{Y_{u}})$.
Therefore, $(\alpha \uplus \beta) =_{\mathcal{E}} 
(\alpha|_{X\setminus \{x\}} \uplus \{ x \mapsto u^{\circ} (\alpha|_{X\setminus \{x\}} \uplus\beta|_{Y_{u}})\} 
\uplus \beta)$.
Now note that 
$\mathit{vars}((\Gamma\rightarrow\Lambda)(\{y \mapsto \overline{y}\}_{y \in Y_{u}} \uplus
\{ \overline{x}\mapsto \overline{u}\}) )=Y\setminus Y_{u}$,
and that $(T_{\mathcal{E}^{\square}},[\overline{\alpha} \uplus 
\overline{\beta}|_{Y_{u}}])$ is a model of
$[\overline{X} \uplus
  \overline{Y}_{u} ,\mathcal{E},(H \cup \{\overline{x} =
  \overline{u}\})_{\mathit{simp}}]$, because
it satisfies $H$ since $(T_{\mathcal{E}^{\square}},[\overline{\alpha}])$
does, and it satisfies $\overline{x} =
  \overline{u}$ by $\alpha(x) =_{\mathcal{E}} u^{\circ} (\alpha|_{X\setminus
  \{x\}} \uplus\beta|_{Y_{u}})$.
Therefore,
$(T_{\mathcal{E}^{\square}},[\overline{\alpha}\uplus\overline{\beta}|_{Y_{u}}]),  
[\beta|_{Y\setminus Y_{u}}] \models
 (\Gamma\rightarrow\Lambda)(\{y \mapsto \overline{y}\}_{y \in Y_{u}} \uplus
\{ \overline{x}\mapsto \overline{u}\})
$, that is,
$T_{\mathcal{E}^{\square}} \models  ((\Gamma\rightarrow\Lambda) (\{y \mapsto \overline{y}\}_{y \in  Y_{u}} \uplus
\{\overline{x}\mapsto \overline{u}\}))^{\circ}(\alpha \uplus \beta)$,
which means,
$T_{\mathcal{E}^{\square}} \models  (\Gamma\rightarrow\Lambda)^{\circ}
\{ x\mapsto u^{\circ}\} (\alpha \uplus \beta)$, which
by the semantic equivalence $(\alpha \uplus \beta) =_{\mathcal{E}} 
(\alpha|_{X\setminus \{x\}} \uplus \{ x \mapsto u^{\circ} 
(\alpha|_{X\setminus \{x\}} \uplus\beta|_{Y_{u}})\} \uplus \beta)$
forces
$T_{\mathcal{E}^{\square}} \models  ((\Gamma\rightarrow\Lambda)^{\circ}
(\alpha \uplus \beta)$, which by 
(b).1 forces $T_{\mathcal{E}^{\square}} \models 
\Lambda^{\circ}
(\alpha \uplus \beta)$, contradicting (b).$2'$, as desired.

\vspace{2ex}

\noindent {\bf NS}.  We prove soundness for the fully general
version of the  {\bf NS} rule.  By the minimality assumption 
we have: (i) $[\overline{X},\mathcal{E},H] \not\models
 (\Gamma\rightarrow \Lambda) [f(\vec{v})]_{p}$; and (ii)
$\{[\overline{X}
\uplus\overline{Y}_{i,j},\mathcal{E},H \uplus 
\widetilde{\overline{\alpha}}_{i,j}|_{\overline{X}_{f(\vec{v})}}
]
 \models (\Gamma_{i},(\Gamma\rightarrow\Lambda)[r_{i}]_{p})
\overline{\alpha}_{i,j}
\}_{ i \in I_{0}}^{j \in J_{i}}$.
But (i) exactly means that there
is a model $(T_{\mathcal{E}^{\square}},[\overline{\alpha}])$
of $[\overline{X},\mathcal{E},H]$ 
such that 
 $(T_{\mathcal{E}^{\square}},[\overline{\alpha}])
\not\models  (\Gamma\rightarrow\Lambda) [f(\vec{v})]_{p}$.
Therefore, there is  
a ground constructor substitution $\beta$ with domain the
variables $Z$ of $\Gamma\rightarrow\Lambda$
such that $(T_{\mathcal{E}^{\square}},[\overline{\alpha}]), [\beta]
\not\models (\Gamma\rightarrow\Lambda) [f(\vec{v})]_{p}$,
that is,
(a) $T_{\mathcal{E}^{\square}}
\not\models (\Gamma\rightarrow\Lambda)^{\circ} [f(\vec{v})^{\circ}]_{p}
(\alpha \uplus \beta)$.  But, by sufficient completeness
of $f$, this means that there exist $i \in I_{0}$,
$\alpha_{i,j} \in
\mathit{Unif}_{B_{0}}(f(\vec{v})^{\circ}=f(\vec{u_{i}}))$,
with $\mathit{ran}(\alpha_{i,j})$ fresh variables,
and $\rho \in [\mathit{ran}(\alpha_{i,j}) \rightarrow T_{\Omega}]$
such that, denoting $Z_{f(\vec{v})} =\mathit{vars}(f(\vec{v}))$,
(b.1) $f(\vec{v})^{\circ}(\alpha \uplus \beta)
|_{X_{f(\vec{v})}\uplus Z_{f(\vec{v})}}
=_{B_{0}}f(\vec{u_{i}})\alpha_{i,j}\rho$; which forces
(b.2)
$(\alpha \uplus \beta)=_{B_{0}} \alpha|_{X \setminus X_{f(\vec{v})}}
\uplus \alpha_{i,j}\rho \uplus \beta|_{Z \setminus Z_{f(\vec{v})}}$; and
we furthermore have
(b.3) $T_{\mathcal{E}^{\square}} \models \Gamma_{i} \alpha_{i,j}\rho$.
We claim that 
 $(T_{\mathcal{E}^{\square}},[\overline{\alpha} 
\uplus \overline{\rho}|_{Y_{i,j}}])$ is a model of 
$[\overline{X}
\uplus\overline{Y}_{i,j},\mathcal{E},H \uplus 
\widetilde{\overline{\alpha}}_{i,j}|_{\overline{X}_{f(\vec{v})}}]$.
Showing this boils down to showing
 $T_{\mathcal{E}^{\square}} \models x = \alpha_{i,j}(x)
(\alpha \uplus \rho|_{Y_{i,j}})$,
$x \in X_{f(\vec{v})}$,
which follows from (b.2) and $T_{\mathcal{E}^{\square}} \models
B_{0}$.
Therefore, by (ii) we have,
 $T_{\mathcal{E}^{\square}} \models 
 (\Gamma_{i},(\Gamma\rightarrow\Lambda)[r_{i}]_{p})^{\circ}
\alpha_{i,j}  (\alpha \uplus \rho|_{Y_{i,j}})$,
which forces
$T_{\mathcal{E}^{\square}} \models 
 (\Gamma_{i},(\Gamma\rightarrow\Lambda)[r_{i}]_{p})^{\circ}
\alpha_{i,j}  (\alpha \uplus \rho
\uplus \beta|_{Z \setminus Z_{f(\vec{v})}})$, which by
the freshness assumption on $\mathit{ran}(\alpha_{i,j})$
is just
$T_{\mathcal{E}^{\square}} \models 
 (\Gamma_{i},(\Gamma\rightarrow\Lambda)[r_{i}]_{p})^{\circ}
(\alpha|_{X \setminus X_{f(\vec{v})}} \uplus \alpha_{i,j} \rho
\uplus \beta|_{Z \setminus Z_{f(\vec{v})}})$,
which by,  (b.1), (b.3) and $T_{\mathcal{E}^{\square}}$
satisfying equation $[i]$ and $B_{0}$, forces
$T_{\mathcal{E}^{\square}} \models 
 (\Gamma\rightarrow\Lambda)^{\circ}
(\alpha|_{X \setminus X_{f(\vec{v})}} \uplus \rho \alpha_{i,j} 
\uplus \beta|_{Z \setminus Z_{f(\vec{v})}})$,
which by (b.2) forces
$T_{\mathcal{E}^{\square}} \models 
 (\Gamma\rightarrow\Lambda)^{\circ}
\alpha_{i,j}  (\alpha \uplus \beta)$,
contradicting (a), as desired.

\vspace{2ex}

\noindent {\bf CS}.  By the minimality assumption
we have (i) $[\overline{X},\mathcal{E},H\cup\{\Gamma\rightarrow\Delta\}] \not\models
 \Gamma\theta,\Gamma'\rightarrow\Lambda\wedge(\Delta\theta,\Delta')$;
and (ii)
$[\overline{X},\mathcal{E},H\cup\{\Gamma\rightarrow\Delta\}] \models
 \Gamma\theta,\Gamma'\rightarrow\Lambda$.  Therefore, there is a model
$(T_{\mathcal{E}^{\square}},[\overline{\alpha}])$ of 
$[\overline{X},\mathcal{E},H\cup\{\Gamma\rightarrow\Delta\}]$
and a constructor ground substitution $\beta$ with domain $Y=
\mathit{vars}(\Gamma\theta,\Gamma'\rightarrow\Lambda\wedge(\Delta\theta,\Delta'))$
such that: (a).1 
$(T_{\mathcal{E}^{\square}},[\overline{\alpha}]) 
\models \Gamma\rightarrow\Delta$;
(a).2 $T_{\mathcal{E}^{\square}} \models 
(\Gamma\theta,\Gamma')^{\circ}(\alpha \uplus \beta)$; and
(b) $T_{\mathcal{E}^{\square}} \not\models 
(\Lambda\wedge(\Delta\theta,\Delta'))^{\circ}(\alpha \uplus \beta)$, i.e.,
either (b).1 $T_{\mathcal{E}^{\square}} \not\models 
\Lambda^{\circ}(\alpha \uplus \beta)$ holds,
or (b).2 $T_{\mathcal{E}^{\square}} \not\models (\Delta\theta,\Delta')^{\circ}
(\alpha \uplus \beta)$ does.  But,  (a).1 and (a).2
force $T_{\mathcal{E}^{\square}} \models  (\Delta\theta)^{\circ}
(\alpha \uplus \beta)$, which makes (b)
equivalent to (b).1.  And (a).2 and (ii) force
$T_{\mathcal{E}^{\square}} \models (\Delta\theta)^{\circ}
(\alpha \uplus \beta)$, contradicting (b).1, as desired.

\vspace{2ex}

\noindent {\bf ERL} and {\bf ERR}.  We give the proof of soundness
for  {\bf ERL}; the proof for  {\bf ERR} is entirely similar.
By the minimality assumption we have:
(i) $[\overline{X},\mathcal{E},H] \not\models (u=v)\theta,\Gamma\rightarrow\Lambda$,
 (ii) $[\overline{X},\mathcal{E},H] \models (u'=v')\theta,\Gamma\rightarrow\Lambda$
and (iii) $[\emptyset,\mathcal{E},\emptyset] \models u=v\Leftrightarrow
u'=v'$.  This means that there is a model $(T_{\mathcal{E}^{\square}},[\overline{\alpha}])$
of $[\overline{X},\mathcal{E},H]$
and a ground constructor substitution $\beta$ with domain the
variables of $(u=v)\theta,\Gamma\rightarrow\Lambda$
such that: (a)
$T_{\mathcal{E}^{\square}}\not\models((u=v)\theta,\Gamma\rightarrow\Lambda)^{\circ}
(\alpha \uplus \beta)$,  and (b)
$T_{\mathcal{E}^{\square}} \models((u'=v')\theta,\Gamma\rightarrow\Lambda)^{\circ}
(\alpha \uplus \beta)$ (note that by the $\mathit{vars}(u=v) \supseteq
\mathit{vars}(u=v)$ assumption, the same $\beta$ works for both
satisfaction statements).  But this is impossible, since,
by (iii), in the inductive evaluation of the Tarskian semantics of
(a) we can replace the truth value of
$(u=v)\theta (\alpha \uplus \beta)$ by that of $(u'=v')\theta (\alpha
\uplus \beta)$, so that the formulas in (a) and (b) must
evaluate to the same truth value, contradicting (a), as desired.

\vspace{2ex}

\noindent {\bf ICC}.  By the minimality assumption 
we have $[\overline{X},\mathcal{E},H] 
\not\models \Gamma
\rightarrow \Lambda$, and
$[\overline{X},\mathcal{E},H] \models  \bigwedge_{i \in I}
\Gamma^{\sharp}_{i} \rightarrow
\Lambda^{\sharp}_{i}$.
We will show that this is impossible if we show the equivalence
$(\dagger) \; \; [\overline{X},\mathcal{E},H] 
\models \Gamma
\rightarrow \Lambda \Leftrightarrow
[\overline{X},\mathcal{E},H]  \models \bigwedge_{i \in I}
\Gamma^{\sharp}_{i}
\rightarrow \Lambda^{\sharp}_{i}$.
But (1) $[\overline{X},\mathcal{E},H] 
\models \Gamma
\rightarrow \Lambda$ holds
iff (2) $[\overline{X} \uplus \overline{Y},\mathcal{E},H \uplus 
\{\overline{\Gamma}\}] \models \overline{\Lambda}$ does,
where $Y = \mathit{vars}(\Gamma \rightarrow
\Lambda)$, and
$\overline{\Gamma}$ (resp. $\overline{\Lambda}$)
 is obtained from $\Gamma$ (resp. $\Lambda$)
 by replacing
each $y \in Y$ by its corresponding constant $\overline{y} \in
\overline{Y}$.  And since $cc^\succ_{B_0}(\overline{\Gamma})$
is the ground Knuth-Bendix completion modulo $B_{0}$ of
$\overline{\Gamma}$, by construction
we have the equivalence
$B_{0} \models \bigwedge
\overline{\Gamma}
\Leftrightarrow  \bigwedge cc^\succ_{B_0}(\overline{\Gamma})$, and,
since $B_{0}$ belongs to  $\mathcal{E}$,    a fortiori we have the equivalence
$[\overline{X} \uplus \overline{Y},\mathcal{E},H] \models \bigwedge
\overline{\Gamma}
\Leftrightarrow  \bigwedge cc^\succ_{B_0}(\overline{\Gamma})$.
Therefore, (2)
holds iff
(3) $[\overline{X} \uplus \overline{Y},\mathcal{E},H \uplus 
\{ cc^\succ_{B_0}(\overline{\Gamma})\}] \models \overline{\Lambda}$ does.
Furthermore, by Lemma \ref{formula-E=-equiv}
and the construction of  $\overline{\Gamma}^{\sharp}$ we
have the equivalence
$[\overline{X} \uplus \overline{Y},\mathcal{E},H \uplus 
\{ cc^\succ_{B_0}(\overline{\Gamma})\}] \models
 \bigwedge cc^\succ_{B_0}(\overline{\Gamma}) \Leftrightarrow
 \overline{\Gamma}^{\sharp}$.  Therefore, (3) holds iff
 (4)  $[\overline{X} \uplus \overline{Y},\mathcal{E},H \uplus
 \overline{\Gamma}^{\sharp}] \models \overline{\Lambda}$ does.
 Of course, if $\overline{\Gamma}^{\sharp}= \bot$
 (4) trivially holds.
 But this is the case where, by convention,
$[\overline{X},\mathcal{E},H]  \models \bigwedge_{i \in I} \Gamma^{\sharp}_{i}
\rightarrow \Lambda^{\sharp}_{i}$ denotes $\top$;  so
in this case have thus proved that (1) holds and therefore $(\dagger)$.
 From now on we may assume that 
 $\overline{\Gamma}^{\sharp} = \bigvee_{i \in I} \overline{\Gamma}^{\sharp}_{i}$
 with $I \not= \emptyset$.  Therefore, (4) holds iff
(5)  $[\overline{X} \uplus \overline{Y},\mathcal{E},H] \models  \bigwedge_{i \in I}
\overline{\Gamma}^{\sharp}_{i} \rightarrow \overline{\Lambda}$ does.
But, by Lemma \ref{formula-E=-equiv}, for each $i \in I$,
 $[\overline{X} \uplus \overline{Y},\mathcal{E},H \uplus
\overline{\Gamma}_{i}^{\sharp}] \models  \overline{\Lambda}$ holds
iff 
$[\overline{X} \uplus \overline{Y},\mathcal{E},H \uplus
\overline{\Gamma}^{\sharp}_{i}] \models
\overline{\Lambda}^{\sharp}_{i}$ does,
and therefore
$[\overline{X} \uplus \overline{Y},\mathcal{E},H ] \models
\overline{\Gamma}_{i}^{\sharp} \rightarrow \overline{\Lambda}$ holds iff
$[\overline{X} \uplus \overline{Y},\mathcal{E},H ] \models
\overline{\Gamma}_{i}^{\sharp} \rightarrow \overline{\Lambda}^{\sharp}_{i}$
does.  Therefore, (5) holds iff 
(6) $[\overline{X} \uplus \overline{Y},\mathcal{E},H] \models  \bigwedge_{i \in I}
\overline{\Gamma}^{\sharp}_{i} \rightarrow
\overline{\Lambda}^{\sharp}_{i}$ does; which itself holds
iff
(7) $[\overline{X},\mathcal{E},H] \models  \bigwedge_{i \in I}
\Gamma^{\sharp}_{i} \rightarrow
\Lambda^{\sharp}_{i}$ does.
This means that we have proved the equivalence $(\dagger)$ for all cases, as
desired.

\vspace{2ex}

\noindent {\bf VARSAT}.  The contradiction of the
minimality assumption is in this
case is quite immediate, since for an $\mathcal{E}_{1}$-formula
$\Gamma^{\circ} \rightarrow \Lambda^{\circ}$
its negation is unsatisfiable in $T_{\mathcal{E}}$ iff
$T_{\mathcal{E}}\models \Gamma^{\circ} \rightarrow \Lambda^{\circ}$,
which, a fortiori, implies that any model 
$(T_{\mathcal{E}^{\square}},[\overline{\alpha}])$
of $[\overline{X},\mathcal{E},H]$ must satisfy
$\Gamma \rightarrow \Lambda$, contradicting the minimality
assumption that $[\overline{X},\mathcal{E},H] \not\models  \Gamma \rightarrow \Lambda$, 
as desired.

\vspace{2ex}

\noindent  {\bf GSI}. By the minimality assumption we have:
(i) $[\overline{X},\mathcal{E},H] \not\models
 \Gamma\rightarrow\bigwedge_{j\in J}\Delta_j$; and
(ii) $\{[\overline{X}\uplus\overline{Y}^{\bullet}_{i},\mathcal{E},
H\uplus H_{i}] \models
(\Gamma\rightarrow\bigwedge_{j\in
    J}\Delta_j)\{z\mapsto\overline{u}^{\bullet}_i\}
\}_{1 \leq i \leq n}$, where $z$ is a variable of sort $s$.  
But (i) exactly means that there
is a model $(T_{\mathcal{E}^{\square}},[\overline{\alpha}])$
of $[\overline{X},\mathcal{E},H]$ 
such that 
 $(T_{\mathcal{E}^{\square}},[\overline{\alpha}])
\not\models \Gamma\rightarrow\bigwedge_{j\in J}\Delta_j$.
Therefore, there is  
a ground constructor substitution $\beta$ with domain the
variables $Z$ of $\Gamma\rightarrow\bigwedge_{j\in J}\Delta_j$
such that (a) $(T_{\mathcal{E}^{\square}},[\overline{\alpha}]), [\beta]
\not\models \Gamma\rightarrow\bigwedge_{j\in J}\Delta_j$.
We can choose among such $\beta$ one such that the term size
$|\beta(z)|$ is a  smallest possible number $m$.  Let us denote 
such a choice  by  $\beta_{\mathit{min}}$.  This means
that 
(b) for any $\beta'\in[Z \rightarrow T_{\Omega}]$ such that
$|\beta'(z)|< m$ we must have
$(T_{\mathcal{E}^{\square}},[\overline{\alpha}]), [\beta']
\models \Gamma\rightarrow\bigwedge_{j\in J}\Delta_j$.
Furthermore, there is a $u_{k}$ in the
$B_{0}$-generator set $\{u_{1},\; \ldots \; u_{n}\}$ for sort $s$
and a ground substitution $\rho$ with domain $Y_{k}$ such that:
(c) $\beta_{\mathit{min}}(z) =_{B_{0}}u_{k}\rho$ and, since
$B_{0}$ is size-preserving, (d) $|u_{k}\rho|=m$, and
since each $v \in \mathit{PST}_{B_{0}, \leq s}(u_{k})$
 is a proper $B_{0}$-subterm of $u_{k}$,
 $|\rho(v)|< m$ for each $v \in \mathit{PST}_{B_{0}, \leq s}(u_{k})$.
We now can distinguish two cases:

\vspace{1ex}

\noindent {\bf Case} (1): $Y_{k} \not= \emptyset$
and the hypotheses 
\[
H_{k}=\{((\Gamma\rightarrow\Delta_j)\{z\rightarrow\overline{v}\}) 
!\,_{\vec{\mathcal{E}}_{{\overline{X} \uplus\overline{Y}_{i}}_U}^{=}}
\mid v \in \mathit{PST}_{B_{0}, \leq s}(u_{k})\; \wedge \; j \in J \}
\]
are non-trivial (therefore,  $\mathit{PST}_{B_{0}, \leq s}(u_{k}) 
\not= \emptyset$).
Then,  $\overline{Y}^{\bullet}_{k} = \overline{Y}_{k}$
and $\overline{u_{k}}^{\bullet} = \overline{u_{k}}$.  We claim that
$(T_{\mathcal{E}^{\square}},[\overline{\alpha} \uplus \overline{\rho}])$
is a model of $[\overline{X}\uplus\overline{Y}_{k},
\mathcal{E},H\uplus
  \{(\Gamma\rightarrow\Delta_j)\{z\mapsto \overline{v}\}
\}_{v \in \mathit{PST}_{B_{0}, \leq s}(u_{k})}^{j\in J}]$.  This is the case
because: (e) $(T_{\mathcal{E}^{\square}},[\overline{\alpha} \uplus
\overline{\rho}]) \models H$ since
$(T_{\mathcal{E}^{\square}},[\overline{\alpha}])\models H$, and
(f) $(T_{\mathcal{E}^{\square}},[\overline{\alpha} \uplus
\overline{\rho}]) \models
\{(\Gamma\rightarrow\Delta_j)\{z\mapsto \overline{v}\}
\}_{v \in \mathit{PST}_{B_{0}, \leq s}(u_{k})}^{j\in J}$, that is,
$T_{\mathcal{E}^{\square}}  \models
\{(\Gamma\rightarrow\Delta_j)^{\circ}\{z\mapsto  v\}
(\alpha \uplus
\rho)
\}_{v \in \mathit{PST}_{B_{0}, \leq s}(u_{k})}^{j\in J}$.
(f) holds  because,
 for each $v \in \mathit{PST}_{B_{0}, \leq s}(u_{k})$,
$\mathit{vars}((\Gamma\rightarrow\Delta_j)\{z\mapsto \overline{v}\})
\subseteq Z\setminus \{z\}$, and
since the variables $Y_{k}$ are fresh 
and therefore do not appear in $\Gamma\rightarrow\Delta_j$, 
$j \in J$, (f) is equivalent to
$T_{\mathcal{E}^{\square}}  \models
\{(\Gamma\rightarrow\Delta_j)^{\circ}(\{z\mapsto  \rho(v)\}
\uplus \alpha)
\}_{v \in \mathit{PST}_{B_{0}, \leq s}(u_{k})}^{j\in J}$,
wich must hold because
any constructor ground substitution $\gamma$
of the variables $Z\setminus \{z\}$
gives us a ground substitution
$\beta' = \{z\rightarrow \rho(v)\} \uplus \gamma$
with domain $Z$ such that we must have
$T_{\mathcal{E}^{\square}}  \models
\{(\Gamma\rightarrow\Delta_j)^{\circ}(\{z\mapsto \rho(v)\}
\uplus \alpha \uplus \gamma)
\}_{v \in \mathit{PST}_{B_{0}, \leq s}(u_{k})}^{j\in J}$, that is,
$T_{\mathcal{E}^{\square}}  \models
\{(\Gamma\rightarrow\Delta_j)^{\circ}(\alpha \uplus \beta')
\}_{v \in \mathit{PST}_{B_{0}, \leq s}(u_{k})}^{j\in J}$, which must
hold by (b), since  $|\beta'(v)|=|\rho(v)|< m$ for each 
$v \in \mathit{PST}_{B_{0}, \leq s}(u_{k})$.
Therefore, (ii) forces
(g) $(T_{\mathcal{E}^{\square}},[\overline{\alpha} \uplus
\overline{\rho}]) \models (\Gamma\rightarrow\bigwedge_{j\in
    J}\Delta_j)\{z\mapsto\overline{u}_k\}$, that is,
$T_{\mathcal{E}^{\square}}  \models (\Gamma\rightarrow\bigwedge_{j\in
    J}\Delta_j)^{\circ} \{z\mapsto u_k\}(\alpha \uplus
\rho)$, which, again, since 
the variables $Y_{k}$ are fresh 
and therefore do not appear in
$ (\Gamma\rightarrow\bigwedge_{j\in
    J}\Delta_j)^{\circ}$,   is just
$T_{\mathcal{E}^{\square}}  \models (\Gamma\rightarrow\bigwedge_{j\in
    J}\Delta_j)^{\circ} (\{z\mapsto u_k\rho\} \uplus 
\alpha)$; but since 
$\beta_{\mathit{min}} =_{B_{0}} \{z\mapsto u_k\rho\} \uplus 
\beta|_{Z \setminus \{z\}}$, (g) is equivalent
to $T_{\mathcal{E}^{\square}}  \models (\Gamma\rightarrow\bigwedge_{j\in
    J}\Delta_j)^{\circ} (\alpha \uplus \beta_{\mathit{min}})$,
contradicting (a), as desired.

\vspace{1ex}

\noindent {\bf Case} (2): 
Either $Y_{k} = \emptyset$ or
the hypotheses $H_{k}$
are trivial, so that $\overline{Y}^{\bullet}_{k} = \emptyset$
and $\overline{u_{k}}^{\bullet} =  u_{k}$.  Therefore, 
for $i=k$, (ii) becomes:
$[\overline{X},\mathcal{E},
H] \models
(\Gamma\rightarrow\bigwedge_{j\in
    J}\Delta_j)\{z\mapsto u_{k}\})$.  
But since $(T_{\mathcal{E}^{\square}},[\overline{\alpha}])$ is a model
of $[\overline{X},\mathcal{E},H]$, 
for the ground constructor substitution
$\beta_{\mathit{min}}|_{Z \setminus \{z\}} \uplus \rho$
(where in case $Y_{k} = \emptyset$, $\rho$ is the
empty substitution) we must have
$T_{\mathcal{E}^{\square}} \models
(\Gamma\rightarrow\bigwedge_{j\in
    J}\Delta_j)^{\circ}\{z\mapsto u_{k}\}
(\alpha \uplus \beta_{\mathit{min}}|_{Z \setminus \{z\}} \uplus \rho)$, which
since the variables $Y_{k}$, if any, are fresh, just means
$T_{\mathcal{E}^{\square}} \models
(\Gamma\rightarrow\bigwedge_{j\in
    J}\Delta_j)^{\circ}(\{z\mapsto u_{k}\rho\} \uplus
\alpha \uplus \beta_{\mathit{min}}|_{Z \setminus \{z\}})$.
But since 
$\beta_{\mathit{min}} =_{B_{0}} \{z\mapsto u_k\rho\} \uplus 
\beta|_{Z \setminus \{z\}}$ this is equivalent to
$T_{\mathcal{E}^{\square}}  \models (\Gamma\rightarrow\bigwedge_{j\in
    J}\Delta_j)^{\circ} (\alpha \uplus \beta_{\mathit{min}})$,
contradicting (a), as desired.

\vspace{2ex}

\noindent {\bf NI}.  By the minimality assumption
we have: (i) $[\overline{X},\mathcal{E},H] \not\models
 (\Gamma\rightarrow\bigwedge_{l\in L}\Delta_l) [f(\vec{v})]_{p}$, and
(ii) $\{[\overline{X}
\uplus\overline{Y}^{\bullet}_{i,j},\mathcal{E},(H \uplus 
H_{i,j}]
 \models  (\Gamma_{i},(\Gamma\rightarrow\bigwedge_{l\in L}\Delta_l)[r_{i}]_{p})
\overline{\alpha}^{\bullet}_{i,j}
\}_{ i \in I_{0}}^{j \in J_{i}}$.
But (i) exactly means that there
is a model $(T_{\mathcal{E}^{\square}},[\overline{\alpha}])$
of $[\overline{X},\mathcal{E},H]$ 
such that 
 $(T_{\mathcal{E}^{\square}},[\overline{\alpha}])
\not\models  (\Gamma\rightarrow\bigwedge_{l\in L}\Delta_l) [f(\vec{v})]_{p}$.
Therefore, there is  
a ground constructor substitution $\beta$ with domain the
variables $Z$ of $\Gamma\rightarrow\bigwedge_{j\in L}\Delta_l$
such that (a) $(T_{\mathcal{E}^{\square}},[\overline{\alpha}]), [\beta]
\not\models (\Gamma\rightarrow\bigwedge_{l\in L}\Delta_l) [f(\vec{v})]_{p}$.
We can choose among such $\beta$ one such that the term size
$|f(\vec{v})\beta|$ is a smallest possible number $m$.  Let us denote 
such a choice by  $\beta_{\mathit{min}}$.  This means
that 
(b) for any $\beta'\in[Z \rightarrow T_{\Omega}]$ such that
$|f(\vec{v})\beta'|< m$ we must have
$(T_{\mathcal{E}^{\square}},[\overline{\alpha}]), [\beta']
\models (\Gamma\rightarrow\bigwedge_{l\in L}\Delta_l) [f(\vec{v})]_{p}$.
Furthermore, by sufficient completeness,
there must be a rule 
$[i]: f(\vec{u_{i}}) \rightarrow r_{i} \;\; \mathit{if} \;\;
\Gamma_{i}$ among those defining $f$,
a unifier $\alpha_{i,j} \in \mathit{Unif}_{B_{0}}(f(\vec{v})=f(\vec{u_{i}}))$
with $\mathit{ran}(\alpha_{i,j})=Y_{i,j}$ fresh variables,
and a constructor ground substitution
$\rho \in [Y_{i,j} \rightarrow T_{\Omega}]$ such that,
defining  $Z_{\vec{v}}= \mathit{vars}(f(\vec{v}))$,
we have: (c).1 $\beta_{\mathit{min}}|_{Z_{\vec{v}}}=_{B_{0}}
(\alpha_{i,j} \rho) |_{Z_{\vec{v}}}$, and therefore
$\beta_{\mathit{min}} =_{B_{0}} \beta_{\mathit{min}}|_{Z \setminus Z_{\vec{v}}}  
\uplus (\alpha_{i,j} \rho) |_{Z_{\vec{v}}}$;
(c).2 $f(\vec{v}) \beta_{\mathit{min}} =_{B_{0}}   f(\vec{u_{i}}) \alpha_{i,j}\rho$, 
so that $|f(\vec{u_{i}}) \alpha_{i,j}\rho| = m$;
(c).3  $\Gamma_{i} \beta_{\mathit{min}} =_{B_{0}}   \Gamma_{i}
\alpha_{i,j}\rho$, and $\mathcal{E} \vdash \Gamma_{i}\alpha_{i,j}\rho$
(since rule $[i]$ applies); and therefore
(c).4 $f(\vec{v}) \beta_{\mathit{min}} =_{\mathcal{E}} r_{i}\beta_{\mathit{min}}$.
We can now distinguish two cases:

\vspace{1ex}

\noindent {\bf Case} (1):  The simplified hypotheses
$(H_{i,j}) !\,_{\vec{\mathcal{E}}_{{\overline{X} \uplus \overline{Y}_{i,j}}_U}^{=}}$
are nontrivial, so that $\overline{Y}^{\bullet}_{i,j}=
\overline{Y}_{i,j}$ and
$\overline{\alpha}^{\bullet}_{i,j} = \overline{\alpha}_{i,j}$.
We claim that 
$(T_{\mathcal{E}^{\square}},[\overline{\alpha} \uplus\overline{\rho}])$ is a model
of  $[\overline{X}
\uplus\overline{Y}_{i,j},\mathcal{E}, H \uplus  H_{i,j}]$.
Fists of all, 
$(T_{\mathcal{E}^{\square}},[\overline{\alpha}\uplus\overline{\rho}])\models H$
since $(T_{\mathcal{E}^{\square}},[\overline{\alpha}])\models H$.
Let us see that we also have
$(T_{\mathcal{E}^{\square}},[\overline{\alpha}\uplus\overline{\rho}])\models H_{i,j}$.
That is, we need to show that for each induction hypothesis
$(\Gamma\rightarrow\Delta_l) 
\overline{\gamma}$ in $H_{i,j}$, and for each ground constructor
substitution
$\delta \in [(Z \setminus Z_{\vec{v}}) \rightarrow T_{\Omega}]$,
$T_{\mathcal{E}^{\square}} \models
(\Gamma\rightarrow\Delta_l)^{\circ}
\gamma (\alpha \uplus \rho \uplus \delta)$ holds, which
by the freshness assumption on $Y_{i,j}
$ just means that
(d) $T_{\mathcal{E}^{\square}} \models
(\Gamma\rightarrow\Delta_l)^{\circ}
(\alpha \uplus (\gamma \rho) \uplus \delta)$ holds.
Furthermore, by the definition 
of the matching substitution $\gamma$ we must have
 $f(\vec{v}) \gamma (\alpha \uplus \rho \uplus \delta)  = 
f(\vec{v}) \gamma \rho = f(\vec{w}) \alpha_{i,j}\gamma$,
with $f(\vec{w})$ a \emph{proper subcall} of rule $[i]$,
so that, by (c).2, $|f(\vec{w}) \alpha_{i,j}\rho| < m$.
But, since $(T_{\mathcal{E}^{\square}},[\overline{\alpha}])$ is a model
of $[\overline{X},\mathcal{E},H]$, 
 viewing $(\gamma \rho) \uplus \delta$
as a decomposition of a ground substitution
$\beta'$ with domain $Z$ and noticing that
$|f(\vec{v}) \beta'| = |f(\vec{v}) \gamma\rho| < m$, (d) indeed holds
because of (b).  Therefore, by (ii), we must have
$(T_{\mathcal{E}^{\square}},[\overline{\alpha} \uplus\overline{\rho}])
\models (\Gamma_{i},(\Gamma\rightarrow\bigwedge_{l\in L}\Delta_l)[r_{i}]_{p})
\overline{\alpha}_{i,j}$. 
In particular, since
$\mathit{vars}((\Gamma_{i},(\Gamma\rightarrow
\bigwedge_{l\in L}\Delta_l)[r_{i}]_{p})
\overline{\alpha}_{i,j}) = Z\setminus Z_{\vec{v}}$,
this must be the case for the ground constructor substitution
$\beta_{\mathit{min}}|_{Z\setminus Z_{\vec{v}}}$.  That is, we must have
$T_{\mathcal{E}^{\square}}
\models (\Gamma_{i},(\Gamma^{\circ}\rightarrow\bigwedge_{l\in
  L}\Delta^{\circ}_l)
[r_{i}]_{p})
\alpha _{i,j} (\alpha \uplus \rho \uplus \beta_{\mathit{min}}|_{Z\setminus Z_{\vec{v}}})$,
that is,
$T_{\mathcal{E}^{\square}}
\models (\Gamma_{i},
(\Gamma^{\circ}\rightarrow\bigwedge_{l\in
  L}\Delta^{\circ}_l)
[r_{i}]_{p}) (\alpha \uplus(\alpha _{i,j} \rho) \uplus
\beta_{\mathit{min}}|_{Z\setminus Z_{\vec{v}}})$.  But, by (c).3,
this is equivalent to
$T_{\mathcal{E}^{\square}}
\models (\Gamma^{\circ}\rightarrow\bigwedge_{l\in
  L}\Delta^{\circ}_l)
[r_{i}]_{p} (\alpha \uplus(\alpha _{i,j} \rho) \uplus
\beta_{\mathit{min}}|_{Z\setminus Z_{\vec{v}}})$, which by 
(c).1 is equivalent to
$T_{\mathcal{E}^{\square}}
\models (\Gamma^{\circ}\rightarrow\bigwedge_{l\in
  L}\Delta^{\circ}_l)
[r_{i}]_{p} (\alpha \uplus \beta_{\mathit{min}})$, which by (c).4 is
equivalent to
$T_{\mathcal{E}^{\square}}
\models (\Gamma^{\circ}\rightarrow\bigwedge_{l\in
  L}\Delta^{\circ}_l)
[f(\vec{v})]_{p} (\alpha \uplus \beta_{\mathit{min}})$,
contradicting (a), as desired.

\vspace{1ex}

\noindent {\bf Case} (2):  the simplified  hypotheses
$(H_{i,j}) !\,_{\vec{\mathcal{E}}_{{\overline{X} \uplus \overline{Y}_{i,j}}_U}^{=}}$
 are trivial, so that $Y^{\bullet}_{i,j} = \emptyset$ and
$\overline{\alpha}^{\bullet}_{i,j} = \alpha_{i,j}$.
This means that, for rule $[i]$, (ii) gives us
$[\overline{X},\mathcal{E},H]
 \models  (\Gamma_{i},(\Gamma\rightarrow\bigwedge_{l\in L}\Delta_l)[r_{i}]_{p})
\alpha_{i,j}$. But, since $(T_{\mathcal{E}^{\square}},[\overline{\alpha}])$ is a model
of $[\overline{X},\mathcal{E},H]$ and 
$\mathit{vars}((\Gamma_{i},(\Gamma\rightarrow
\bigwedge_{l\in L}\Delta_l)[r_{i}]_{p}) \alpha_{i,j}) \subseteq
(Z\setminus Z_{\vec{v}}) \uplus Y_{i,j}$,
in particular, for the ground constructor substitution
$\beta_{\mathit{min}}|_{Z\setminus Z_{\vec{v}}} \uplus
(\alpha _{i,j} \rho)$ we must have
$T_{\mathcal{E}^{\square}}
\models (\Gamma_{i},
(\Gamma^{\circ}\rightarrow\bigwedge_{l\in
  L}\Delta^{\circ}_l)
[r_{i}]_{p}) (\alpha \uplus(\alpha _{i,j} \rho) \uplus
\beta_{\mathit{min}}|_{Z\setminus Z_{\vec{v}}})$.  But, by (c).3,
this is equivalent to
$T_{\mathcal{E}^{\square}}
\models (\Gamma^{\circ}\rightarrow\bigwedge_{l\in
  L}\Delta^{\circ}_l)
[r_{i}]_{p} (\alpha \uplus(\alpha _{i,j} \rho) \uplus
\beta_{\mathit{min}}|_{Z\setminus Z_{\vec{v}}})$, which by 
(c).1 is equivalent to
$T_{\mathcal{E}^{\square}}
\models (\Gamma^{\circ}\rightarrow\bigwedge_{l\in
  L}\Delta^{\circ}_l)
[r_{i}]_{p} (\alpha \uplus \beta_{\mathit{min}})$, which by (c).4 is
equivalent to
$T_{\mathcal{E}^{\square}}
\models (\Gamma^{\circ}\rightarrow\bigwedge_{l\in
  L}\Delta^{\circ}_l)
[f(\vec{v})]_{p} (\alpha \uplus \beta_{\mathit{min}})$,
contradicting (a), as desired.

\vspace{2ex}

\noindent {\bf Existential} ($\exists$).   By the minimality 
assumption we have:
(i) $[\emptyset,\mathcal{E},\emptyset] \not\models
 (\exists \chi)(\Gamma\rightarrow\Lambda)$, that is,
$T_{\mathcal{E}^{\square}}
\not\models (\exists \chi)(\Gamma\rightarrow\Lambda)$;
 and
(ii) $[\emptyset,\mathcal{E},\emptyset] \models
 I(\Gamma\rightarrow\Lambda)$,
that is,
$T_{\mathcal{E}^{\square}}
\models  I(\Gamma\rightarrow\Lambda)$,
where $\chi$ is a Skolem signature,
$\Gamma\rightarrow\Lambda$
is a $\Sigma \cup \chi$-multiclause, 
and $I: \chi \rightarrow
\mathcal{E}$ is a theory interpretation.
But, as pointed out in Section \ref{Skolem-subsection},
$T_{\mathcal{E}^{\square}}
\models  I(\Gamma\rightarrow\Lambda)$
gives a \emph{constructive proof} of
$T_{\mathcal{E}^{\square}}
\models (\exists \chi)(\Gamma\rightarrow\Lambda)$
by proving that the intepretation $I(f)$
of the symbols  $f$ in the Skolem signature 
satisfies the interpretation $I(\Gamma\rightarrow\Lambda)$
of the $\Sigma \uplus \chi$-multiclause
$\Gamma\rightarrow\Lambda$, in direct contradiction of (i),
as desired.

\vspace{2ex}

\noindent {\bf LE}.  By the minimality assumption we have:
(i) $[\overline{X},\mathcal{E},H] \not\models
 \Gamma\rightarrow\Lambda$;
(ii) $[\overline{X}_{0},\mathcal{E},H_{0}] \models
  \Gamma'\rightarrow\bigwedge_{j\in J}\Delta'_j$;
(iii) $[\overline{X},\mathcal{E},H] \models H_{0}$;
and (iv) $[\overline{X},\mathcal{E},H\uplus
\{\Gamma'\rightarrow\Delta'_j\}_{j\in J}] \models
 \Gamma\rightarrow\Lambda$,
where 
$\emptyset \subseteq \overline{X}_{0} \subseteq \overline{X}$.  
 But (i) exactly means that there
is a model $(T_{\mathcal{E}^{\square}},[\overline{\alpha}])$
of $[\overline{X},\mathcal{E},H]$ 
such that 
(a) $(T_{\mathcal{E}^{\square}},[\overline{\alpha}])
\not\models  \Gamma\rightarrow\Lambda$.
However, $(T_{\mathcal{E}^{\square}},[\overline{\alpha}])$
is also a model of $[\overline{X},\mathcal{E},H\uplus
\{\Gamma'\rightarrow\Delta'_j\}_{j\in J}]$.
This is so because we obviously have (b)
$(T_{\mathcal{E}^{\square}},[\overline{\alpha}]) \models H$;
which by (iii)  forces  
$(T_{\mathcal{E}^{\square}},[\overline{\alpha}]) \models H_{0}$;
which, by (ii) and $\overline{X}_{0} \subseteq \overline{X}$,
 forces (c) $(T_{\mathcal{E}^{\square}},[\overline{\alpha}]) \models
\Gamma'\rightarrow\Delta'_j$ for each $j \in J$.
Therefore, (iv) gives us
$(T_{\mathcal{E}^{\square}},[\overline{\alpha}])
\models  \Gamma\rightarrow\Lambda$, in flat contradiction
of (a), as desired.

\vspace{2ex}

\noindent {\bf SP}.  By the minimality assumption
we have: (i) $[\overline{X},\mathcal{E},H] 
\not\models \Gamma
\rightarrow \Lambda$;
(ii) $\{[\overline{X},\mathcal{E},H] 
\models \Gamma_{i} \theta ,\Gamma
\rightarrow \Lambda\}_{i \in I}$;
(iii) $[\overline{X},\mathcal{E},H] \models H_{0}$;
and (iv) $[\overline{X}_{0},\mathcal{E},H_{0}]   \models 
\bigvee_{i \in I} \Gamma_{i}$,
where $\overline{X}_{0} \subseteq \overline{X}$,
and $\mathit{vars}((\bigvee_{i \in I} \Gamma_{i})\theta)
\subseteq \mathit{vars}(\Gamma \rightarrow
\Lambda) = Y$.
But (i) means we have constructor
ground substitutions $\alpha$ and $\beta$ with respective domains
$X$ and $Y$ such that: (a) $(T_{\mathcal{E}^{\square}},[\overline{\alpha}])$
is a model of $[\overline{X},\mathcal{E},H]$, and (b)
 $T_{\mathcal{E}^{\square}} \not\models 
(\Gamma\rightarrow 
\Lambda)^{\circ}
(\alpha \uplus \beta)$,
that is,
(b).1 $T_{\mathcal{E}^{\square}} \models 
(\Gamma)^{\circ}
(\alpha \uplus \beta)$, and
(b).2 $T_{\mathcal{E}^{\square}} \not\models 
\Lambda^{\circ}
(\alpha \uplus \beta)$.
But, since by (iii) and (iv),  we must have
$T_{\mathcal{E}^{\square}} \models 
((\bigvee_{i \in I} \Gamma_{i})\theta)^{\circ}
(\alpha \uplus \beta)$, 
there must be an $i \in I$ such that
$T_{\mathcal{E}^{\square}} \models 
(\Gamma_{i}\theta)^{\circ}
(\alpha \uplus \beta)$.
Therefore, by (b).1, we have
(c)  $T_{\mathcal{E}^{\square}} \models 
(\Gamma_{i}\theta ,\Gamma)^{\circ}
(\alpha \uplus \beta)$,
which by (ii)  and 
$(T_{\mathcal{E}^{\square}},[\overline{\alpha}])$
being a model of $[\overline{X},\mathcal{E},H]$
forces $T_{\mathcal{E}^{\square}} \models 
\Lambda^{\circ}
(\alpha \uplus \beta)$, contradicting (b).2, as desired.

\vspace{2ex}

\noindent {\bf CAS}. We prove the two modalities: for
a $z$ variable, and for $\overline{z}\in \overline{X}$.

\vspace{1ex}

\noindent {\bf Case} (1). We have $z$ of sort $s$ occurring in
$\Gamma\rightarrow\Lambda$, and $\{u_1,\cdots,u_n\}$ is
a $B_{0}$-generator set for sort $s$, with each $u_i$, $1\leq i\leq  n$, 
having fresh variables.  By the minimality assumption
we have: (i) $[\overline{X},\mathcal{E},H] \not\models 
\Gamma\rightarrow\Lambda$; and
(ii) $\left\{[\overline{X},\mathcal{E},H] \models
 (\Gamma\rightarrow\Lambda)\{z\mapsto u_i\}\right\}_{1\leq i\leq n}$.
A moment's reflection helps us realize that the task of proving this
case \emph{coincides} with the ---already accomplished---  task
 of proving  {\bf Case} (2) for the {\bf GSI} rule, i.e., the case when
either $Y_{i} = \emptyset$ or the induction 
hypotheses are trivial: the proof is identical here.

\vspace{1ex}  

\noindent {\bf Case} (2). We have
$\overline{z}\in \overline{X}$ of sort $s$ occurs in
$\Gamma\rightarrow\Lambda$, 
$Y_{i} = \mathit{vars}(u_i)$, $\overline{Y}_{i}$
are the corresponding new fresh constants, and
 $\overline{u}_i \equiv u_{i} \{y \mapsto \overline{y}\}_{y \in
   Y_{i}}$.  By the minimality assumption we have:
(i) $[\overline{X},\mathcal{E},H] \not\models \Gamma\rightarrow
\Lambda$; and (ii) 
$\{[\overline{X}
  \uplus\overline{Y}_{i},\mathcal{E},
H  \uplus \{\overline{z} = \overline{u}_i\}] \models
(\Gamma\rightarrow\Lambda)\{\overline{z}\mapsto\overline{u}_i\} 
\}_{1 \leq i \leq n}$.   But (i) exactly means that there
is a model $(T_{\mathcal{E}^{\square}},[\overline{\alpha}])$
of $[\overline{X},\mathcal{E},H]$ 
such that 
$(T_{\mathcal{E}^{\square}},[\overline{\alpha}])
\not\models  \Gamma\rightarrow\Lambda$, that is,
(a) $T_{\mathcal{E}^{\square}} \not\models  
(\Gamma\rightarrow\Lambda)^{\circ}\alpha$.
But then there is a $k$, $1 \leq k \leq n$,
and a ground constructor substitution
$\rho$ with domain $Y_{k}$ such that
(b).1 $\alpha(x)=_{B_{0}} u_{k}\; \rho$; 
which implies
(b).2 $\alpha =_{B_{0}} \alpha|_{X \setminus \{z\}} \uplus \{z \mapsto u_{k}\; \rho\}$.
We claim that 
$(T_{\mathcal{E}^{\square}},[\overline{\alpha} \uplus \overline{\rho}])$
is a model of $[\overline{X}
  \uplus\overline{Y}_{i},\mathcal{E},
H  \uplus \{\overline{z} = \overline{u}_{k}\}]$.
Indeed: (b)
$(T_{\mathcal{E}^{\square}},[\overline{\alpha} \uplus\overline{\rho}]) \models H$
because $(T_{\mathcal{E}^{\square}},[\overline{\alpha}]) \models H$; and
(c) $(T_{\mathcal{E}^{\square}},[\overline{\alpha} \uplus\overline{\rho}])
\models \overline{z} = \overline{u}_{k}$, that is,
$T_{\mathcal{E}^{\square}} \models z = u_{k} (\alpha \uplus \rho)$,
which holds by (b).1 and $T_{\mathcal{E}^{\square}} \models B_{0}$.
But then (ii) forces
$(T_{\mathcal{E}^{\square}},[\overline{\alpha} \uplus\overline{\rho}])
\models (\Gamma\rightarrow\Lambda)\{\overline{z}\mapsto\overline{u}_k\}$,
that is,
$T_{\mathcal{E}^{\square}} \models  (\Gamma\rightarrow\Lambda)^{\circ}
\{z \mapsto u_k\}(\alpha \uplus \rho)$,
which by (b).2 is equivalent to
$T_{\mathcal{E}^{\square}} \models  (\Gamma\rightarrow\Lambda)^{\circ}\alpha$,
flatly contradicting (a), as desired.

\vspace{2ex}

\noindent {\bf VA}.
By the minimality assumption we have:
(i)  $[\overline{X},\mathcal{E},H] 
\not\models [\overline{X},\mathcal{E},H] \Vdash u=v[w]_{p},\;\Gamma\rightarrow\Lambda$;
and (ii) $[\overline{X},\mathcal{E},H] \models
u=v[z]_{p},\;z=w,\; \Gamma\rightarrow\Lambda$, where $z$ is
fresh variable of sort the least sort of $w$.
But (i) exactly means 
that, for $Y = \mathit{vars}( u=v[w]_{p},\;\Gamma \rightarrow
\Lambda)$
we have constructor
ground substitutions $\alpha$ and $\beta$ with respective domains
$X$ and $Y$ such that: (a) $(T_{\mathcal{E}^{\square}},[\overline{\alpha}])$
is a model of $[\overline{X},\mathcal{E},H]$, and (b)
 $T_{\mathcal{E}^{\square}} \not\models 
(u=[w]_{p},\;\Gamma' \rightarrow 
\Lambda)^{\circ}
(\alpha \uplus \beta)$.
That is,
(b).1 $T_{\mathcal{E}^{\square}} \models 
(u=v[w]_{p} ,\Gamma)^{\circ}
(\alpha \uplus \beta)$, and
(b).2  $T_{\mathcal{E}^{\square}} \not\models 
\Lambda^{\circ}
(\alpha \uplus \beta)$.
Consider now the constructor ground substitution
$\alpha \uplus \beta \uplus \{z \mapsto w(\alpha \uplus \beta)\}$.
Since $v[z]_{p}(\alpha \uplus \beta \uplus \{z \mapsto w(\alpha
\uplus \beta)\}) =  v[w]_{p} (\alpha \uplus \beta)$, (b).1 forces
$T_{\mathcal{E}^{\square}} \models 
(u=v[z]_{p},\;z=w\; ,\Gamma)^{\circ}
(\alpha \uplus \beta \uplus \{z \mapsto w(\alpha \uplus \beta)\})$.
But then (ii) forces
$T_{\mathcal{E}^{\square}} \models 
\Lambda^{\circ}
(\alpha \uplus \beta \uplus \{z \mapsto w(\alpha \uplus \beta)\})$.
But since $z$, being fresh, does not occur in $\Lambda^{\circ}$,
this in turn forces
$T_{\mathcal{E}^{\square}} \models 
\Lambda^{\circ}
(\alpha \uplus \beta)$, in direct contradiction with (b).2.

\vspace{2ex}

\noindent $\textbf{EQ}$.  By the minimality assumption we have:
(i) $[\overline{X},\mathcal{E},H]
  \not\models
  (\Gamma\rightarrow\Lambda)[w]_{p}$
  and (ii)  $[\overline{X},\mathcal{E},H] 
  \models (\Gamma\rightarrow\Lambda)[v \theta\gamma]_{p}$, with
  $w=_{B_{0}}u\theta\gamma$.  But, since the conditional equation $\Gamma'\rightarrow u=v$
  used in this inference step belongs to either $\mathcal{E}$ or $H$,
 this is impossible, since, by
  Lemma \ref{formula-E-equiv}, we have
  $[\overline{X},\mathcal{E},H] 
  \models (\Gamma\rightarrow\Lambda)[v \theta\gamma]_{p}
  \; \Leftrightarrow \;  (\Gamma\rightarrow\Lambda)[w]_{p} $,
  which by (ii) forces
   $[\overline{X},\mathcal{E},H]
  \models
  (\Gamma\rightarrow\Lambda)[w]_{p}$.

  \vspace{2ex}

\noindent $\textbf{Cut}$.  By the minimality assumption
we have: (i) $[\overline{X},\mathcal{E},H] \not\models \Gamma\rightarrow\Lambda$,
 (ii) $[\overline{X},\mathcal{E},H] \models
 \Gamma\rightarrow \Gamma'$ and (iii) $[\overline{X},\mathcal{E},H]
 \models \Gamma,\Gamma' \rightarrow\Lambda$, with
 $\mathit{vars}(\Gamma') \subseteq \mathit{vars}(\Gamma
 \rightarrow \Lambda)=Y$.  But (i)--(iii) are respectively equivalent
 to: (i$'$) $[\overline{X} \uplus \overline{Y},\mathcal{E},H] \not\models
 \overline{\Gamma}\rightarrow\overline{\Lambda}$,
 (ii$'$)  $[\overline{X} \uplus \overline{Y},\mathcal{E},H] \models
 \overline{\Gamma}\rightarrow \overline{\Gamma'}$ and (iii$'$)
 $[\overline{X} \uplus \overline{Y},\mathcal{E},H]
 \models \overline{\Gamma},\overline{\Gamma'}
 \rightarrow\overline{\Lambda}$, which is impossible,
 since $((G \Rightarrow G') \wedge ((G \wedge G') \Rightarrow L))
 \Rightarrow (G \Rightarrow L)$ is a tautology of Propositional Logic.
 Therefore, by the Tarskian semantics of formulas,  (ii$'$) and
  (iii$'$) imply $[\overline{X} \uplus \overline{Y},\mathcal{E},H] \models
  \overline{\Gamma}\rightarrow\overline{\Lambda}$, which is
  equivalent to $[\overline{X},\mathcal{E},H] \models \Gamma\rightarrow\Lambda$.

\vspace{2ex}

\noindent This finishes the proof of the Soundness Theorem.  $\Box$

\end{document}